%% file: main.tex
\documentclass{article}
\pdfoutput=1

\usepackage[utf8]{inputenc}
\usepackage{pdfsync}

\def\showauthornotes{1}
\def\showkeys{0}

\def\confversion{0}
\def\widemargin{0}

\input{macros}
\input{abbrev}

\usepackage{enumitem}

\title{List Decoding Expander-Based Codes up to Capacity\\ in Near-Linear Time}
 \author{
 Shashank Srivastava\thanks{{\tt DIMACS (Rutgers) \& Institute for Advanced Study}. {\tt shashanks@ias.edu}.}
 \and
 Madhur Tulsiani\thanks{{\tt Toyota Technological Institute at Chicago}. {\tt madhurt@ttic.edu}. Supported by the NSF grant CCF-2326685.} 
 }

\allowdisplaybreaks
\begin{document}

\date{}
\maketitle

\thispagestyle{empty}

\begin{abstract}
We give a new framework based on graph regularity lemmas, for list decoding and list recovery of
codes based on spectral expanders. 
Using existing algorithms for computing regularity decompositions of sparse
graphs in (randomized) near-linear time, and 
appropriate choices for the constant-sized inner/base codes, we prove the following:
\medskip
\begin{itemize}
\item Expander-based codes constructed using the distance amplification technique of Alon, Edmonds and Luby [FOCS 1995] with rate $\rho$, can be list decoded to a radius $1 - \rho - \epsilon$ in near-linear time. 
By known results, the output list has size $O(1/\epsilon)$.
\item The above codes of Alon, Edmonds and Luby, with rate $\rho$,  can also be list recovered to radius $1 - \rho - \epsilon$ in near-linear time, with constant-sized output lists.
\item The Tanner code construction of Sipser and Spielman [IEEE Trans. Inf. Theory 1996] with distance $\delta$, can be list decoded to radius $\delta - \epsilon$ in near-linear time, with constant-sized output lists.
\end{itemize}
\medskip
Our results imply novel combinatorial as well as algorithmic bounds for each of the above explicit constructions. 
All of these bounds are obtained via combinatorial rigidity phenomena, proved using (weak) graph regularity. 
The regularity framework allows us to lift the list decoding and list recovery properties for the local base codes, to the global codes obtained via the above constructions.

%
%
%
\end{abstract}

\newpage

\pagenumbering{roman}
\tableofcontents
\clearpage

\newpage
\pagenumbering{arabic}
\setcounter{page}{1}

\input{intro}
%
\input{prelims}

\input{enumeration}

\input{reg_lemma}

\input{ael}

\input{list_recovery}

\input{tanner}

\section*{Acknowledgements}
We are grateful to Tushant Mittal for helpful discussions in the early phase of this work. 
We also thank the FOCS 2025 reviewers for helpful comments and suggestions.


\bibliographystyle{alphaurl}
\bibliography{macros,madhur}

\appendix
\input{appendix}

\end{document}

%% file: macros.tex
\usepackage{xspace,enumerate}
\usepackage{amsmath,amssymb}
\usepackage{amsthm}
\ifnum\confversion=0
	\usepackage[toc,page]{appendix}
\fi
\usepackage{thmtools}
\usepackage{thm-restate}
\usepackage{graphicx}
\usepackage{boxedminipage}
\usepackage{csquotes}
\usepackage{makecell}
\usepackage{tabularx}
\usepackage{relsize}
\usepackage{scalerel}
\usepackage[dvipsnames]{xcolor} 
\usepackage{tikz}

\ifnum\showkeys=1
\usepackage[color]{showkeys}
\fi

\definecolor{darkred}{rgb}{0.5,0,0}
\definecolor{darkgreen}{rgb}{0,0.35,0}
\definecolor{darkblue}{rgb}{0,0,0.55}

\ifnum\confversion=1
\usepackage[pdfstartview=FitH,pdfpagemode=UseNone,colorlinks=false,linkcolor=darkblue,filecolor=darkred,citecolor=darkgreen,urlcolor=darkred,pagebackref,bookmarks=false]{hyperref}
\else
\usepackage[pdfstartview=FitH,pdfpagemode=UseNone,colorlinks,linkcolor=NavyBlue,filecolor=blue,citecolor=OliveGreen,urlcolor=NavyBlue,pagebackref]{hyperref}
\fi

\usepackage[capitalise,nameinlink]{cleveref}
\usepackage[T1]{fontenc}

\usepackage{mathtools,dsfont}

\usepackage{mathpazo}
\usepackage{microtype}
\ifnum\widemargin=0
\usepackage[top=1in, bottom=1in, left=1in, right=1in]{geometry}
\else
\usepackage[top=1in, bottom=1in, left=1.25in, right=1.25in]{geometry}
\fi

\ifnum\showauthornotes=1
\newcommand{\Authornote}[3]{{\sf\small\color{#3}{[#1: #2]}}}
\newcommand{\Authorcomment}[2]{{\sf \small\color{gray}{[#1: #2]}}}
\newcommand{\Authorfnote}[2]{\footnote{\color{red}{#1: #2}}}
\else
\newcommand{\Authornote}[3]{}
\newcommand{\Authorcomment}[2]{}
\newcommand{\Authorfnote}[2]{}
\fi

\declaretheorem[numberwithin=section]{theorem}
\declaretheorem[sibling=theorem]{lemma}
\declaretheorem[sibling=theorem]{claim}
\declaretheorem[sibling=theorem]{proposition}

\declaretheorem[sibling=theorem]{corollary}

\theoremstyle{definition}
\declaretheorem[sibling=theorem]{definition}
\declaretheorem[sibling=theorem]{remark}

\newtheorem{algo}[theorem]{Algorithm}

\newenvironment{algorithm}[3]
        {\noindent\begin{boxedminipage}{\textwidth}\begin{algo}[#1]\ \par\medskip
        {\begin{tabular}{r l}
        \textsf{Input:} & #2\\[3 pt]
        \textsf{Output:} & #3 
        \end{tabular}\par\enskip}}
        {\end{algo}\end{boxedminipage}}

\newenvironment{new-algorithm}[1]
        {\begin{figure}[!ht]\noindent\begin{boxedminipage}{\textwidth}\begin{algo}[#1]\ \par\medskip
        {\enskip}}
        {\end{algo}\end{boxedminipage} \end{figure}}

\def\FullBox{\hbox{\vrule width 6pt height 6pt depth 0pt}}

\def\qed{\ifmmode\qquad\FullBox\else{\unskip\nobreak\hfil
\penalty50\hskip1em\null\nobreak\hfil\FullBox
\parfillskip=0pt\finalhyphendemerits=0\endgraf}\fi}

\def\qedsketch{\ifmmode\Box\else{\unskip\nobreak\hfil
\penalty50\hskip1em\null\nobreak\hfil$\Box$
\parfillskip=0pt\finalhyphendemerits=0\endgraf}\fi}

\def\to{\rightarrow}
\def\eps{\varepsilon}
\def\epsilon{\varepsilon}

\def\eps{\epsilon}

\def\phi{\varphi}
\def\cal{\mathcal}

\def\implies{\Rightarrow}
\newcommand{\sub}{\ensuremath{\subseteq}}

\newcommand{\defeq}{:=}

\renewcommand{\bar}{\overline} 

\newcommand{\given}{\;\ifnum\currentgrouptype=16 \middle\fi \vert\;}

\newcommand{\ie}{i.e.,\xspace}

\newcommand{\etal}{et al.\xspace}

\newcommand{\mper}{\,.}
\newcommand{\mcom}{\,,}

\usepackage{nicefrac}

\let\nfrac=\nicefrac
\let\ffrac=\flatfrac

\DeclarePairedDelimiter\parens{\lparen}{\rparen}
\DeclarePairedDelimiter\abs{\lvert}{\rvert}
\DeclarePairedDelimiter\norm{\lVert}{\rVert}

\DeclarePairedDelimiter\braces{\lbrace}{\rbrace}
\DeclarePairedDelimiter\brackets{\lbrack}{\rbrack}
\DeclarePairedDelimiter\angles{\langle}{\rangle}
\DeclarePairedDelimiterXPP\lnorm[1]{}\lVert\rVert{_2}{#1}

\newcommand{\inparen}[1]{\parens*{#1}}             
\newcommand{\inbraces}[1]{\braces*{#1}}           
\newcommand{\insquare}[1]{\brackets*{#1}}             

\newcommand{\one}{{\mathbf{1}}}

\DeclareMathOperator{\sgn}{sgn}

\newcommand{\nnz}[1]{\mathbf{nnz}(#1)}

\newcommand{\Esymb}{\mathbb{E}}
\newcommand{\Psymb}{\mathbb{P}}

\DeclareMathOperator*{\ExpOp}{\Esymb}

\makeatletter
\def\Pr#1{%
    \ProbabilityRender{\Psymb}{#1}%
}

\def\Ex#1{%
    \ProbabilityRender{\Esymb}{#1}%
}

\def\condPE#1#2{%
	\@ifnextchar\bgroup
	{\ConditionalProbabilityRender{\widetilde{\Esymb}}{#1}{#2}}
	{\ProbabilityRender{\widetilde{\Esymb}}{#1 \given #2}}
}

\def\ConditionalProbabilityRender#1#2#3#4{
	\renderwithdist{#1}{#2}{#3 \given #4}	
}

\def\ProbabilityRender#1#2{
  \@ifnextchar\bgroup%
  {\renderwithdist{#1}{#2}}
   {\singlervrender{#1}{#2}}
}
\def\singlervrender#1#2{%
   \ensuremath{\mathchoice
       {{#1}\insquare{ #2 }}
       {{#1}\insquare{ #2 }}
       {{#1}\insquare{ #2 }}
       {{#1}\insquare{ #2 }}
   }
}
\def\renderwithdist#1#2#3{%
   \@ifnextchar\bgroup
   {\superfancyrender{#1}{#2}{#3}}
   {\ensuremath{\mathchoice
      {\underset{#2}{#1}\insquare{ #3 }}
      {{#1}_{#2}\insquare{ #3 }}
      {{#1}_{#2}\insquare{ #3 }}
      {{#1}_{#2}\insquare{ #3 }}
     }
   }
}
\def\superfancyrender#1#2#3#4#5{
   \ensuremath{\mathchoice
      {\underset{#1}{{#1}}\left#4 #3 \right#5}
      {{#1}_{#2}#4 #3 #5}
      {{#1}_{#2}#4 #3 #5}
      {{#1}_{#2}#4 #3 #5}
   }
}
\makeatother

\newcommand{\bigoh}{O}

%% file: abbrev.tex
\def\li{{ \ell }}
\def\ri{{ r }}

\newcommand{\ip}[2] {\ensuremath{\angles*{ #1 , #2 }}}
\newcommand{\trans}{\mathsf{T}}

\newcommand{\factor}{\calB}
\newcommand{\family}{\calF}

\def\Enc{\mathrm{Enc} }

\def\delL{{\delta_L}} 
\def\delR{{\delta_R}} 
\def\decL{{\delta^{\dec}_L}} 
\def\decR{{\delta^{\dec}_R}} 
\def\decT{{ \delta^{\dec}_{\mathrm{Tan}} }} 

\newcommand{\treg}[1]{T_{\mathrm{reg}}(#1)}
\newcommand{\tlocal}[1]{T_{\mathrm{local}}(#1)}
\newcommand{\tlocalL}[1]{T_{\mathrm{local}, L}(#1)}
\newcommand{\tlocalR}[1]{T_{\mathrm{local}, R}(#1)}
\newcommand{\tdec}[1]{T_{\dec}(#1)}

\def\codeL{\calC_1}
\def\codeR{\calC_2}
\def\TanC{{\calC_{\mathrm{Tan}} }}
\def\lsL{{ K_1 }} 
\def\lsR{{ K_2 }} 
\def\lsLR{{ k }} 
\def\listl{{ \calL_{\li} }}
\def\listr{{ \calL_{\ri} }}
\newcommand{\Ot}[1]{\widetilde{\bigoh}\parens*{#1}}
\newcommand{\Oh}[1]{\bigoh\parens*{#1}}

\def\AELC{{ \calC_{\mathrm{AEL}} }}

\newcommand{\R}{{\mathbb R}}

\newcommand{\N}{{\mathbb{N}}}

\newcommand{\F}{{\mathbb F}}

\newcommand{\B}{\{0,1\}\xspace}

\newcommand{\indicator}[1]{\one_{#1}}

\newcommand{\inn}{\mathrm{in}}
\newcommand{\out}{\mathrm{out}}
\newcommand{\dec}{\mathrm{dec}}
\newcommand{\dis}{{\Delta}}

\newcommand{\deffont}{\sf}

\newcommand{\calB}{{\cal B}}
\newcommand{\calC}{{\cal C}}

\newcommand{\calF}{{\cal F}}

\newcommand{\calJ}{{\cal J}}
\newcommand{\calL}{{\cal L}}
\newcommand{\calN}{{\cal N}}

\newcommand{\calR}{{\cal R}}
\newcommand{\calX}{{\cal X}}

\newcommand{\cC}{{\cal C}}

\newcommand{\Zemor}{Z\'emor\xspace}

%% file: intro.tex
\section{Introduction} \label{sec:intro}
Expander graphs have been a powerful and versatile tool for the construction of codes with several interesting
properties.
The simple combinatorial structure of codes based on expander graphs, often makes them adaptable  for a variety of applications, leading to several elegant constructions.
A (very) small sample of the list of applications already includes the seminal
constructions of expander codes~\cite{SS96}, widely used distance amplification
constructions~\cite{ABNNR92, AEL95}, as well as recent breakthrough constructions of
$\epsilon$-balanced codes~\cite{TS17}, locally testable codes~\cite{DELLM22}, and quantum LDPC
codes~\cite{PK22, LZ22}. 
A detailed account of the rich interactions between coding theory and expander graphs, and
pseudorandom objects in general, can be found in several excellent surveys and
textbooks on these areas~\cite{GSurvey04, Vadhan12, HLW06, GRS23}. 

One particular advantage of codes based on expanders, has been the amenability of their combinatorial and spectral structure, to techniques for highly efficient encoding and decoding algorithms.
However, while several expander-based codes admit linear time (and even parallelizable) algorithms for unique decoding, this has not necessarily been the case for \emph{list decoding} algorithms. 
Although some constructions of expander-based codes~\cite{GI03, TS17} do allow for fast list decoding, we know of few general algorithmic techniques to exploit expansion for list decoding. Moreover, even the known algorithms do not extend to the optimal tradeoffs between the the code parameters such as rate and distance, and the fraction of correctible errors.

For an error-correcting code $\calC \subseteq \Sigma^n$ over a finite alphabet $\Sigma$, the distance $\delta(\calC)$ and the rate $\rho(\calC)$ are fundamental parameters, which measure the error-correction capacity and redundancy of a code
\[
\delta(\calC) 
~\defeq~ \min_{\underset{g\neq h}{g,h \in \calC}} \frac{1}{n} \cdot \abs{\inbraces{i \in [n] ~|~ g_i \neq h_i}}
\qquad \text{and} \qquad
\rho(\calC) ~\defeq~ \frac{1}{n} \cdot \log_{\abs{\Sigma}}\abs{\calC} \mper
\]
The task of list decoding, defined by Elias~\cite{E57} and Wozencraft~\cite{W58}, requires finding the list of all codewords within a given distance (say) $\beta$ of a given $g \in \Sigma^n$, denoted as $\calL(g,\beta)$, where
\[
\calL(g,\beta) ~\defeq~ \inbraces{~h \in \calC ~\mid~ \Delta(g,h) \leq \beta~} \mper
\]
Understanding list decoding then presents two challenges: the \emph{combinatorial} problem of finding the largest $\beta$ (as a function of other parameters of the code) for which the list size remains bounded, and the \emph{algorithmic} problem of efficiently recovering the list even when the size is known to be bounded.
\vspace{-5 pt}
\paragraph{Combinatorial bounds for list decoding.} 
Any code $\calC$ with distance $\delta$ can be unique decoded from $\delta/2$ errors, and it is known that the list size remains bounded for a threshold $\calJ(\delta) \in (\delta/2, \delta)$ known as the Johnson bound. 
Decoding beyond the Johnson bound requires relying on the structure of the code to obtain bounds on the list size. 
This has often been achieved by incorporating additional algebraic structure in the code construction, to take advantage of the significant machinery for list decoding using polynomials~\cite{Guruswami:survey}.
For a code $\calC$ with rate $\rho$, the Singleton bound implies that $\delta \leq 1 - \rho$, and a generalization by Shangguan and Tamo~\cite{ST20} shows that the list size at radius $\beta = 1 - \rho - \eps$ must be $\Omega(1/\eps)$.
%
%
Algebraic code families such as folded Reed-Solomon codes and multiplicity codes, are indeed known to be list decodable close to the optimal radius of $1 - \rho - \eps$ (known as the list decoding capacity)~\cite{GR08, Kop15}, with asymptotically optimal list size~\cite{CZ25}.
However, algebraic structure can be incompatible with other desirable properties such as the presence of low-density parity checks (LDPC), or may require an alphabet size growing with $n$. 
Also, this line of work does not yield list size bounds for several of the expander-based code constructions mentioned above, which may not posses such algebraic structure. 

Recently, list-size bounds for codes constructed via the expander-based distance amplification technique of Alon, Edmonds and Luby~\cite{AEL95} (which we will refer to as ``AEL codes'') were also obtained using ``local-to-global" phenomena for such codes~\cite{JMST25}. 
These results, which also yield list-size bounds at near optimal radii of $\beta \geq 1 - \rho - \eps$ (with an optimal list size of $O(1/\epsilon)$), also have algorithmic versions based on the Sum-of-Squares (SoS) hierarchy of semidefinite programs. 
However, while the algorithms based on the SoS hierarchy run in polynomial time, the exponents of these polynomials are quite large (and grow with $1/\eps$). 
Moreover, these techniques still do not yield list-size bounds (beyond the Johnson bound) for other expander-based codes of interest, nor for generalizations of list decoding like list recovery.
Indeed, it is not known if even the long-studied families of expander codes, such the Tanner code constructions of by Sipser and Spielman~\cite{SS96}, have polynomial (in $n$) list sizes beyond the Johnson bound.
\vspace{-5 pt}
\paragraph{Algorithms for list decoding.}
While the task of bounding the list size is easy below the Johnson bound, the algorithmic question still remains nontrivial even if the list is known to be of small size. 
In the case of algebraic codes, such algorithms were obtained in the foundational work of Guruswami and Sudan~\cite{Sudan97, GS98}. 
Since then, a number of beautiful algorithms have been developed for these and related algebraic codes, and we refer the reader to excellent textbooks and surveys on the subject~\cite{GSurvey04, GRS23} for a more detailed account.
Broadly speaking, these algorithms can be seen as characterizing the list via certain ``algebraic rigidity'' phenomena. They show that certain low-degree polynomials satisfying conditions also satisfied by the list codewords, must in fact \emph{contain} the list of codewords as roots.
However, as mentioned earlier, such algebraic structure is not always available for codes based on expanders.

The use of expansion for list decoding was explored in an early work of Guruswami and Indyk~\cite{GI03} who used expanders to construct linear time list decodable codes (for very small rates).
For the case AEL codes, the study of list decoding algorithms relying on (high-dimensional) expansion was initiated by Dinur \etal~\cite{DHKLNTS19}, who used semidefinite programs to list decode a special case of AEL codes (to a radius smaller than the Johnson bound).
Recent works have also led to list decoding algorithms up to the Johnson bound for AEL and Tanner codes~\cite{JST23}, and also for Ta-Shma's expander-based constructions of $\eps$-balanced codes~\cite{RR23}, using the Sum-of-Squares (SoS) hierarchy.
These algorithms can be seen as the application of a general "proofs to algorithms" framework for designing SoS algorithms~\cite{FKP19} which shows that the spectral proofs used for obtaining distance and list-size bounds for expander codes, can generically be translated to list decoding algorithms.
However, the runtime of algorithms based on the SoS hierarchy involve large polynomial exponents, growing with the proximity to the Johnson bound in the above cases.
As mentioned earlier, for the case of AEL codes, the SoS hierarchy also yields algorithms beyond the Johnson bound, and arbitrarily close to the list decoding capacity~\cite{JMST25}, but again with similarly large exponents.

For the $\eps$-balanced code construction of Ta-Shma~\cite{TS17}, a faster list decoding algorithm was obtained by~\cite{JST21} using structural characterizations for expanding structures, known as ``regularity lemmas''. 
This algorithm runs in near-linear time, but is only known to achieve list decoding up to a threshold significantly \emph{below} the Johnson bound for code.
Moreover, it also requires the development of new regularity lemmas specially adapted to the code construction. 
In this work, we consider if such techniques can be applied to obtain new combinatorial and algorithmic results for a broader class of codes.

\subsection{Our results}
We develop a new framework based on graph regularity lemmas, to obtain near-linear time algorithms for list decoding Tanner codes and AEL codes. 
We show that the regularity lemmas lead to ``combinatorial rigidity'' phenomena, which can be used to design fast algorithms for list decoding various codes up to optimal list decoding radii. 
Our results also imply new combinatorial bounds on list sizes, for codes which were not previously known to be list decodable beyond the Johnson bound. 

We describe the relevant code constructions and results below. 
For simplicity, we restrict the discussion in this paper to versions based on balanced bipartite expanding graphs, although the techniques can be generalized to other variants.
In the discussion below, we always take $G=(L,R,E)$ to be a $d$-regular bipartite graph with $\abs{L} = \abs{R} = n$. It is said to be an $(n,d,\lambda)$-expander if the second singular value of the biadjacency matrix is at most $\lambda \cdot d$.

\subsubsection{Alon-Edmonds-Luby (AEL) codes}
The distance amplification procedure of Alon, Edmonds, and Luby~\cite{AEL95} (AEL) based on
expander graphs, is a versatile method used for constructing various code families with near-optimal rate-distance tradeoff, and other desirable properties. 
The AEL construction works with an expanding graph $G$, an ``inner code'' $\calC_{\inn} \subseteq \Sigma_{\inn}^d$ and an ``outer code'' $\calC_{\out} \subseteq \Sigma_{\out}^n$, with $|\Sigma_{\out}| = |\calC_{\inn}|$.
We view a codeword of $\calC_{\out}$ as being written on the left vertices (formally, $\calC_{\out} \subseteq \Sigma_{\out}^L$. The symbol at each vertex $\ell \in L$ is then viewed as an element of $\Sigma_{\inn}^d$, which gives a labeling of the edges (with an ordering fixed in advance) and thus an element of $\Sigma_{\inn}^E$. 
Each right vertex $r \in R$ the
n collects the $d$ symbols from the incident edges, forming a symbol in $\Sigma_{\inn}^d$.
The procedure can be viewed as using $\calC_{\inn}$ to map $\calC_{\out}$ to a new code $\AELC \subseteq (\Sigma_{\inn}^d)^R \cong \Sigma_{\inn}^E$. 
As illustrated in \cref{fig:ael_intro}, we can think of symbols being on the edges of $G$, which form a codeword in $\AELC$ when grouped from the right. 
The distance of two codewords is the fraction of right vertices $r \in R$ on which they differ.
\begin{figure}[htb]
\begin{center}	
\begin{tikzpicture}[scale = 0.6]
\begin{scope}
\draw[] (0,0) ellipse (1.5cm and 3cm);
\node[below] at (0,-3.1) {$L$};
\draw[] (6,0) ellipse (1.5cm and 3cm);
\node[below] at (6,-3.1) {$R$};


\node[fill,circle,red] (a1) at (0,2) {};
\node[fill,circle,red] (a2) at (0,1) {};
\node[left,darkred] at (-0.2,1) {$\cC_\inn \ni \parens[\big]{h(e_1),h(e_1), h(e_3)} = h_{\ell}  $};
\node[fill,circle,red] (a3) at (0,0) {};
\node[fill,circle,red] (a4) at (0,-1) {};
\node[fill,circle,red] (a5) at (0,-2) {};
\node[fill,circle,blue] (b1) at (6,2) {};
\node[fill,circle,blue] (b2) at (6,1) {};
\node[fill,circle,blue] (b3) at (6,0) {};
\node[below,darkblue] at (10.5,0.45) {{$h_r = \parens[\big]{h(e_2),h(e_4), h(e_5)} \in \Sigma_\inn^d$}};
\node[fill,circle,blue] (b4) at (6,-1) {};
\node[fill,circle,blue] (b5) at (6,-2) {};

\draw[](a2)--node[pos=0.45,sloped, above] {\small{$h(e_1)$}}(b1);
\draw[](a2)--node[pos=0.45,sloped, above] {\small $h(e_2)$}(b3);
\draw[](a2)--node[pos=0.45,sloped, above] {\small $h(e_3)$}(b5);
\draw[] (a4)--node[pos=0.17,sloped, above] {\small{$h(e_4)$}}(b3);
\draw[] (a5)--node[pos=0.3,sloped, below] {\small{$h(e_5)$}}(b3);

\end{scope}
\end{tikzpicture}
\vspace{-10 pt}
\end{center}
\caption{Illustration of the AEL procedure}
\label{fig:ael_intro}
\end{figure}
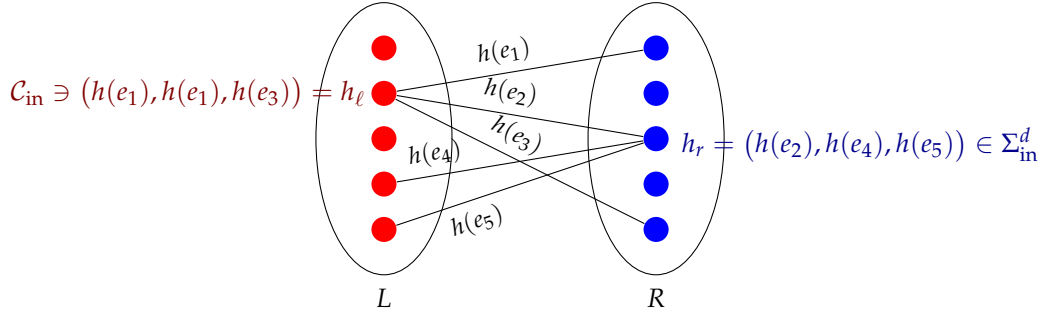

%
An easy computation shows that $\rho(\AELC) \geq \rho(\calC_{\out}) \cdot \rho(\calC_{\inn})$. Moreover, an application of the expander mixing lemma can be used to show that 
\[
\delta(\AELC) ~\geq~ \delta(\calC_{\inn}) - \frac{\lambda}{\delta(\calC_{\out})} \mper
\]
Thus, when $\rho(\calC_{\out}) \geq 1 - \eps$ and $\lambda \leq \eps \cdot \delta(\calC_{\out})$, this gives a ``local-to-global'' amplification of the rate-distance tradeoffs obtained by the code $\calC_{\inn}$:  we get $\rho(\AELC) \geq \rho(\calC_{\inn}) - \eps$ and $\delta(\AELC) \geq \delta(\calC_{\inn}) - \eps$. 
In particular, instantiating the construction with a ``constant-sized'' code $\calC_{\inn}$ achieving the Singleton bound, yields an explicit infinite family of codes (corresponding to an explicit family of expander graphs) which achieve a tradeoff $\eps$-close to the Singleton bound. 

The AEL procedure is highly versatile and has been adapted for a number of applications, particularly because it also preserves structural properties of the outer code $\calC_{\out}$. 
It is easy to see $\calC_{\out}$ and $\calC_{\inn}$ are both linear over $\F_q$, then so is the code $\AELC$. 
It also preserves the property of being LDPC, local testability or locally correctability~\cite{GKORZS17, KMRZS16}, linear time unique decodability \cite{GI05}, and also the duality properties needed for quantum codes~\cite{BGG24}. 
We refer the reader to the excellent discussion in \cite{KMRZS16} for more details.
\vspace{-5 pt}
\paragraph{List decoding.} 
It was recently proved by Jeronimo \etal~\cite{JMST25} that AEL codes also yield ``local-to-global" phenomena for list decodability. 
They proved that if the inner code $\cC_{\inn}$ satisfies a (slight strengthening of) the generalized Singleton bound, then $\AELC$ also satisfies such a bound (for an outer code $\cC_{\out}$ is chosen as above). 
In particular, their result gives codes of rate $\rho$ with asymptotically optimal list size of $O(1/\eps)$ at radius $\beta = 1 - \rho - \eps$, approaching the list decoding capacity.
We give a near-linear time
algorithm for recovering such a list, at any radius $\beta \leq 1 - \rho - \eps$ (assuming near-linear time unique decoding for $\cC_{\out}$).
\begin{theorem}[Informal version of \cref{thm:ael_decoding_instantiation}]\label{thm:ael-dec-intro}
Let $\AELC$ be the code obtained via the AEL construction, applied to an outer code
$\calC_{\out}$ and an inner code $\calC_{\inn}$ over alphabet $\Sigma_{\inn}$, using an
$(n,d,\lambda)$-expander $G=(L,R,E)$. 
Let $\cC_{\out}$ be unique-decodable from fraction $\delta_{\dec}$ of errors in time $T_{\dec}$.
For $\eps > 0$, let the list size $\abs{\calL(\cdot, \delta_{\inn}-\eps)}$ for $\calC_{\inn}$ be
bounded by $K=K(\eps)$, and let $\lambda \leq ((\eps\cdot \delta_{\dec})/(200 K))^2$. 

Then, for $B(\eps,K,\delta_{\dec}) = \exp(O(K^5/(\eps^3\cdot \delta_{\dec}^3)))$, the code $\AELC$ can be decoded up to radius $\delta_{\inn}-2\eps$ in randomized time
$\Ot{n\cdot |\Sigma_{\inn}|^d + nK}+T_{\dec}\cdot B(\eps,K,\delta_{\dec})$. 
Further, the produced list is of size at most $B(\eps,K,\delta_{\dec})$.
\end{theorem} 
We note that in addition to a decoding algorithm, \cref{thm:ael-dec-intro} also gives new
\emph{combinatorial bounds} on the list size, using weaker properties of the inner code $\cC_{\inn}$
than the ones used by~\cite{JMST25}. 
Assuming only a bound $K$ on the list-sizes for the ``local" inner code, where $K$ is independent of inner code's blocklength $d$, we get a bound on the lists size for the global code $\AELC$. 
%

We also emphasize that requiring the inner code $\cC_{\inn}$ to have bounded lists sizes (independent
of the block-length $d$) is easy to satisfy. 
The (constant-sized) random linear codes or folded Reed-Solomon used as inner codes by~\cite{JMST25}
already have list size $O(1/\eps)$ at distance $\delta_{\inn} - \eps$.
They use these to obtain \emph{global} list sizes $O(1/\eps)$, but need the alphabet to be doubly
exponential in $1/\eps$. 
On the other hand, we only need the alphabet to be $\exp((1/\eps)^{O(1)})$ at the cost of possibly
having a larger (constant) global list size. 
Of course when we do have a larger alphabet such that the result of~\cite{JMST25} applies, then our
algorithm also outputs a list of size $O(1/\eps)$.
\cref{thm:ael-dec-intro} can be instantiated to yield the following.
\begin{corollary}[Informal version of \cref{cor:ael_decoding instantiation_final}]\label{cor:ael-dec-intro}
For every $\rho, \eps \in (0,1)$, there exists an infinite family of explicit codes $\AELC \subseteq
(\F_q^d)^n$ obtained via the AEL construction, with rate $\rho$, distance $\delta \geq 1 -
\rho - \eps$, alphabet size $2^{1/\eps^{O(1)}}$, and characterized by parity checks of size
$1/\eps^{O(1)}$ such that
\begin{enumerate}
\item For any $g\in (\F_q^d)^R$, the list $\calL(g,1-\rho-\eps)$ is of size at most $\exp(O(1/\eps^{O(1)}))$.
\item There is an algorithm that given any $g\in (\F_q^d)^R$, runs in time $\widetilde{O}_{\eps}(n)$ and outputs the list $\calL(g,1-\rho-\eps)$.
\end{enumerate}
\end{corollary}
\vspace{-10 pt}
\paragraph{List recovery.} 
Our results also extend to the setting of list recovery, which generalizes the list decoding problem. Instead of a $g \in \Sigma^n$ which specifies a symbol in $\Sigma$ for every $i \in [n]$, one is now given lists $\calL_i$ of size (say) $k$, for each $i \in [n]$. 
The goal is now to recover the list of all codewords, which agree with at least one element of the list $\calL_i$, for $1 - \beta$ fraction of the positions $i \in [n]$. Thus, we aim to find
\[
\calL\inparen{\inbraces{\calL_i}_{i \in [n]}, \beta} ~\defeq~ \inbraces{h \in \calC ~\mid~ \Pr{i \in [n]}{h_i \notin \calL_i} \leq \beta} \mper
\]
Several known list decoding bounds for algebraic codes are also known to extend to the case of list recovery~\cite{GR08, GW13, KRSW23, Tamo24}.
On the other hand, list recovery has often been a \emph{requirement} on the outer code $\cC_{\out}$ to even obtain list decoding for AEL codes ~\cite{GR08}.
We show that AEL codes also support near-linear time list recovery up to the optimal radius of $1 - \rho - \eps$, assuming only the \emph{unique decodability} of the outer code $\cC_{\out}$ and list recovery bounds for the inner code $\cC_{\inn}$.
\begin{theorem}[Informal version of \cref{thm:ael_recovery_instantiation}] \label{thm:ael-rec-intro}
Let $\AELC$ be the code obtained via the AEL construction, using an applied to an outer code
$\calC_{\out}$and an inner code $\calC_{\inn}$ over alphabet $\Sigma_{\inn}$, using an
$(n,d,\lambda)$-expander $G=(L,R,E)$. 
Let $\cC_{\out}$ be unique-decodable from fraction $\delta_{\dec}$ of errors in time $T_{\dec}$.
For $\eps > 0$, let the list size $\abs{\calL(\cdot, \delta_{\inn}-\eps)}$ for $\calC_{\inn}$ be
bounded by $K=K(\eps)$ for list recovery with input lists of size $k$, 
and let $\lambda \leq ((\eps\cdot \delta_{\dec})/(250 k \cdot K))^2$. 

Then, for $B(\eps,\lsLR, K,\delta_{\dec}) = \exp(O(\lsLR^5K^4/(\eps^3\cdot \delta_{\dec}^3)))$, the code $\AELC$ can be list recovered to radius $\delta_{\inn}-2\eps$, from inputs lists of
size $\lsLR$, in randomized time
$\Ot{n\cdot |\Sigma_{\inn}|^d + \lsLR K \cdot n}+T_{\dec}\cdot B(\eps,\lsLR, K,\delta_{\dec})$, with a list size of
at most $B(\eps,\lsLR, K,\delta_{\dec})$.
\end{theorem}
As in the case of list decoding, we note that \cref{thm:ael-rec-intro} yields both combinatorial and algorithmic results for list recovery. 
For the case of AEL codes, even combinatorial bounds for list recovery were not known previously,
and thus both aspects of the above result are new.
Another implication of \cref{thm:ael-rec-intro} is the first explicit family of Low-Density Parity Check (LDPC) codes, which achieve capacity for list recovery. 
%
%
Using appropriate instantiations with an LDPC outer code, we
also get an LDPC code list recoverable to capacity. 
\begin{corollary}[Informal version of \cref{cor:ael_decoding_instantiation_final}]\label{cor:ael_decoding_instantiation_final}
For every $\rho, \eps \in (0,1)$ and $\lsLR \in \N$, there exists an infinite family of explicit
codes $\AELC \subseteq (\F_q^d)^n$ obtained via the AEL construction, with rate $\rho$, distance
$\delta \geq 1 - \rho - \eps$, alphabet size $\exp(\exp(O(\frac{k}{\eps} \log \frac{k}{\eps})))$,
and characterized by parity checks of size $\exp(\frac{k}{\eps}\log\frac{k}{\eps})$ such that
\begin{enumerate}
\item For any $\{\listr\}_{\ri\in R}$, with $|\listr|\leq \lsLR$, the list $\calL\parens*{\{\listr\}_{\ri\in R},1-\rho-\eps}$ is of size $\exp(\exp(O(\frac{\lsLR}{\eps}  \log \frac{k}{\eps} )))$.
\item There is an algorithm that given any $\{\listr\}_{\ri\in R}$, with each $\listr \sub \F_q^d$ and $|\listr|\leq \lsLR$, runs in time $\widetilde{O}_{\lsLR, \eps}(n)$ and outputs the list $\calL\parens*{\{\listr\}_{\ri\in R},1-\rho-\eps}$.
\end{enumerate}
\end{corollary}
%
%
Note, however that the alphabet size of the LDPC codes obtained via the AEL construction is a large
constant (depending on $\eps$).
Our framework also gives new list decoding results for classical LDPC code constructions with a fixed constant alphabet (such as $\F_2$), which we discuss next.

\subsubsection{Tanner codes}
LDPC codes are codes whose codewords can be characterized by local checks. 
Constructions of LDPC codes based on sparse graphs were introduced by the seminal work of
Tanner~\cite{Tanner81}, who gave lower bounds on distance based on the girth of the graph.
Sipser and Spielman~\cite{SS96, Spi96} gave the first constructions of Tanner codes with distance bounds
based on the expansion of the graph, which also admitted elegant linear time encoding and (unique) decoding
algorithms. These was later improved by \Zemor~\cite{Zemor01} to (unique) decode from a much larger
fraction of errors. 
Variants of these constructions have led to several applications~\cite{RU:book} and also serve as the 
building blocks in recent constructions of locally testable codes by Dinur \etal~\cite{DELLM22}
and quantum LDPC codes by Panteleev and Kalachev~\cite{PK22} (see also \cite{LZ22, DHLV23}).

In the setting of \emph{erasures} where the location of the corruptions in the transmitted codeword is known,
recent work has also led to linear-time list decoding algorithms~\cite{RWZ21, HW15}\footnote{In the
erasure setting, where the error locations are known, polynomial time list decoding is often
easier. For example, linear codes can be list decoded from erasures in polynomial time via Gaussian elimination by simply ignoring the known error
locations. Thus, the focus there is to give a \emph{faster} algorithms. However, for the case for
unknown error locations, even polynomial time algorithms are not always known.}, which also
work for list recovery in the large alphabet (and high-rate) case~\cite{HW15}.
In the more challenging setting of decoding from errors, list decoding up to the Johnson bound was
only achieved recently~\cite{JST23} using the SoS semidefinite programming hierarchy. 
At radii beyond the Johnson bound, neither combinatorial nor algorithmic results for list decoding
were known in the case of expander-based Tanner codes.
Our methods yield both combinatorial and algorithmic list decoding for such codes, at radii
arbitrarily close to their distance. We state our results for version of the construction on
balanced bipartite expander graphs~\cite{SS96, Zemor01}. 

Given $G=(L,R,E)$ which is an $(n,d,\lambda)$-expander, the code $\TanC$ is a code with 
block-length $\abs{E} = nd$. 
Given an alphabet $\Sigma$, the code consists of all edge-labelings $h \in \Sigma^E$, such that the
labels in the neighborhood of every vertex
\footnote{Our results also apply to other variants, such as one with different base codes $\calC_1,
  \calC_2$ for the left and right vertices.}, 
belong to a ``base code'' $\calC_0 \subseteq \Sigma^d$.
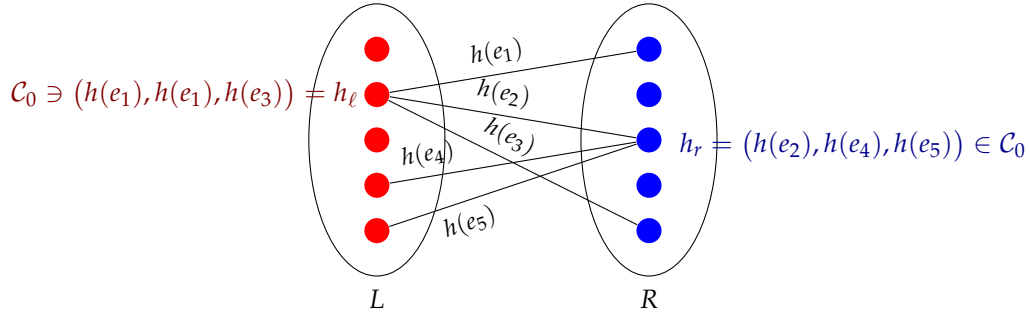
\begin{figure}[htb]
\begin{center}	
\begin{tikzpicture}[scale = 0.6]
\begin{scope}
\draw[] (0,0) ellipse (1.5cm and 3cm);
\node[below] at (0,-3.1) {$L$};
\draw[] (6,0) ellipse (1.5cm and 3cm);
\node[below] at (6,-3.1) {$R$};


\node[fill,circle,red] (a1) at (0,2) {};
\node[fill,circle,red] (a2) at (0,1) {};
\node[left,darkred] at (-0.2,1) {$\cC_0 \ni \parens[\big]{h(e_1),h(e_1), h(e_3)} = h_{\ell}  $};
\node[fill,circle,red] (a3) at (0,0) {};
\node[fill,circle,red] (a4) at (0,-1) {};
\node[fill,circle,red] (a5) at (0,-2) {};
\node[fill,circle,blue] (b1) at (6,2) {};
\node[fill,circle,blue] (b2) at (6,1) {};
\node[fill,circle,blue] (b3) at (6,0) {};
\node[below,darkblue] at (10.5,0.45) {{$h_r = \parens[\big]{h(e_2),h(e_4), h(e_5)} \in \cC_0$}};
\node[fill,circle,blue] (b4) at (6,-1) {};
\node[fill,circle,blue] (b5) at (6,-2) {};

\draw[](a2)--node[pos=0.45,sloped, above] {\small{$h(e_1)$}}(b1);
\draw[](a2)--node[pos=0.45,sloped, above] {\small $h(e_2)$}(b3);
\draw[](a2)--node[pos=0.45,sloped, above] {\small $h(e_3)$}(b5);
\draw[] (a4)--node[pos=0.17,sloped, above] {\small{$h(e_4)$}}(b3);
\draw[] (a5)--node[pos=0.3,sloped, below] {\small{$h(e_5)$}}(b3);

\end{scope}
\end{tikzpicture}
\vspace{-10 pt}
\end{center}
\caption{Bipartite construction for Tanner codes}
\label{fig:ael}
\end{figure}
Note that as opposed to AEL codes, the alphabet $\Sigma$ of $\TanC$ is same as that of base code, and the distance of
two codewords is measured as the fraction of \emph{edges} on which they differ.
When the base code has (fractional) distance $\delta_0$ and $G$ is an $(n,d,\lambda)$-expander, the distance of the code is known to be at least $\delta = \delta_0 \cdot
(\delta_0 - \lambda)$. 

We show that if the base code $\cC_0$ has bounded list sizes (independent of the block-length $d$)
at radius $\delta_0 - \eps$, then code $\TanC$ also has bounded list sizes (independent of $n$) at
radius $\delta_0^2 -2 \eps$, for sufficiently small $\lambda$. Moreover, the list can be recovered in
time $\Ot{\abs{E}}$.
\begin{theorem}[Informal version of \cref{cor:tanner-instantiation}]\label{thm:tanner-dec-intro}
	Let $\TanC$ be the Tanner code defined using an $(n,d,\lambda)$-expander and a base code $\calC_0 \subseteq \Sigma^d$. Suppose $\calC_0$ has distance at least $\delta$, and is list decodable up to radius $\delta^{\dec}$ with a list size of $K$ that is independent of $d$. Then given any $\eps>0$, if $\lambda \leq c\cdot \frac{\eps^4}{K^4}$ for some small constant $c$, then there is an algorithm that takes as input $g\in \Sigma^E$, runs in time $\Ot{n}$, and returns the list $\calL(g,\delta\cdot \delta^{\dec}-\eps)$, of size at most $\exp(O(K^9/\eps^3))$.
\end{theorem}
Note that if we choose the (constant-sized) base code so that $\delta^{\dec} \geq \delta -\eps$, then above algorithm can decode $\TanC$ to radius $\delta^2 - 2\eps$.
To the best of our knowledge, both the combinatorial and algorithmic aspects of
\cref{thm:tanner-dec-intro} are novel for Tanner codes.
All the above results are obtained via a common framework for obtaining ``combinatorial rigidity''
results via graph regularity lemmas, which we describe below.

%

\subsection{Technical overview}
The starting point for our argument is the simple observation that the proofs of distance and other
local-to-global phenomena for expander-based codes essentially rely on estimating the number of
edges between subsets of vertices in the graph, via the expander mixing lemma. 
For carefully chosen subsets $S, T$, the proofs rely on relating the number of edges with some
property (say where two codewords differ) to sizes $\abs{S}, \abs{T}$. 
To develop an algorithmic analogue of such proofs, we consider regularity
lemmas, which provide a framework to capture any such ``edge counting'' proofs for subgraphs $H$ of an expanding graph $G$.
\vspace{-5 pt}
\paragraph{Regularity lemmas.} 
The well-known weak regularity lemma of Frieze and Kannan~\cite{FK96:focs} gives an approximate
decomposition of the adjacency matrix of any (dense) graph as a linear combination of few cut matrices of the form $\indicator{S_t}\indicator{T_t}^{\trans}$. 
Formally, given any graph $H$ on $n$ vertices and $\gamma \in (0,1)$, there exist $\inbraces{(c_t,
  S_t, T_t)}_{t \in [p]}$ for $p \leq 1/\gamma^2$, where $c_t \in \R$ and $S_t, T_t \subseteq [n]$ and $\sum_t \abs{c_t} \leq 1/\gamma$,
such that for any $S, T \subseteq [n]$ 
\[
\abs*{\abs{E_H(S,T)} - \sum_{t \in [p]} c_t \cdot \abs{S_t \cap S}\abs{T_t \cap T}} ~\leq~ \gamma n^2
\quad \Leftrightarrow \quad
\abs*{\indicator{S}^{\trans}{\inparen{A_H - \sum_{t \in [p]} c_t
      \indicator{S_t}\indicator{T_t}^{\trans}}\indicator{T}}} ~\leq~ \gamma n^2 \mper
\]
Several known statements generalize the above to subgraphs of an expanding
graph~\cite{ReingoldTTV08, TrevisanTV09, CCFA09, OGT15, BV22}. A version for subgraph $H$ of an
$(n,d,\lambda)$-expander $G$, would give (see \cref{sec:reg_lemma})
\[
\abs*{\abs{E_H(S,T)} - \sum_{t \in [p]} c_t \cdot \frac{d}{n} \cdot\abs{S_t \cap S}\abs{T_t \cap T}}
~\leq~  \inparen{\gamma + \frac{\lambda}{\gamma}} \cdot nd 
~\leq~ 2\gamma \cdot nd \qquad (\text{for}~\lambda \leq \gamma^2)\mper
\]
Thus, any proof or statement regarding edge counts between any sets $S, T$ can equivalently be
phrased as a statement about their intersections with \emph{fixed} sets $S_1, \ldots, S_p$ and $T_1,
\ldots, T_p$.
Moreover there exist near-linear time algorithms to compute the decomposition $\inbraces{(c_t, S_t,
  T_t)}_{t \in [p]}$, which allow us to efficiently search for sets $S, T$ with given intersection
sizes (up to a small error) and thus, also a given number of edges between them.
We will next show how (for example) the list decoding problem for AEL codes can be characterized in
terms of this search.
\vspace{-5 pt}
\paragraph{Local lists for AEL codes.} 
Recall that given $g \in (\Sigma_{\inn}^d)^R$, the list decoding task for $\AELC$ is to find all $h$
such that $\dis_R(g,h) \leq \beta$, where we use $\dis_R(g,h)$ to denote the fraction of right vertices
$r \in R$ on which $g$ and $h$ differ.
For arbitrarily small $\eps$, we will take $\beta = \delta(\AELC) - 2\eps \leq \delta(\cC_{\inn}) - 2\eps$.
We can also view $g, h \in \Sigma_{\inn}^E$ as labeling the edges, which defines for each
$\ell \in L$ \emph{local projections} $g_{\ell}, h_{\ell} \in \Sigma_{\inn}^d$, defined to be the restrictions of $g$ and $h$ respectively to the edges incident on $\li$.  For $g_{\ell}$ and
$h_{\ell}$ we measure local distances $\Delta(g_{\ell}, h_{\ell})$ as a fraction of the $d$ edges
incident on $\ell$.

Given $g \in (\Sigma_{\inn}^d)^R$, we can compute the local projections $\inbraces{g_{\ell}}_{\ell
  \in L}$, and list decode each $g_{\ell}$ within the inner code $\cC_{\inn}$ to find the \emph{local
lists} $\calL_{\ell}(g_{\ell}, \beta + \eps) ~=~ \inbraces{f_{\ell} \in \cC_{\inn} ~|~
\dis(f_{\ell}, g_{\ell}) \leq \beta + \eps}$.
An application of the expander mixing lemma can be used to show that for any $h \in \AELC$ in the global
list, its local projection does occur in almost all local lists. In fact, for any $g, h \in
(\Sigma_{\inn}^d)^R$ with $\dis_R(g,h) \leq \beta$, we have (for, say $\lambda \leq \eps \cdot \gamma$)
\[
\Pr{\ell \in L}{\dis(g_{\ell}, h_{\ell}) \leq \beta+\eps} 
~~\geq~~ 1 - \nfrac{\lambda}{\eps} 
~~\geq~~ 1 - \gamma\mper
\]
If we have a bound of (say) $K$ on the list sizes for the inner code, we know that
$\abs{\calL_{\ell}(g_{\ell},\beta+\eps)} \leq K$ for each $\ell \in L$.
By padding and ordering arbitrarily, we can take each $\calL_{\ell}$ to be an ordered list of size
$K$, with the $i$-th element denoted as $\calL[i]$.
If we only knew \emph{which} of the $K$ elements in each $\calL_{\ell}$ equals $h_{\ell}$, we would
be able to recover the codeword $h$ in the global list, since we would know it on almost all $\ell
\in L$ (and then unique decode in the code $\cC_{\out}$).
The task is then find the right $i \in [K]$ for each $\ell \in L$, such that $\calL[i]=h_{\ell}$.
\begin{figure}[h]
\centering
\includegraphics[width=0.5\textwidth]{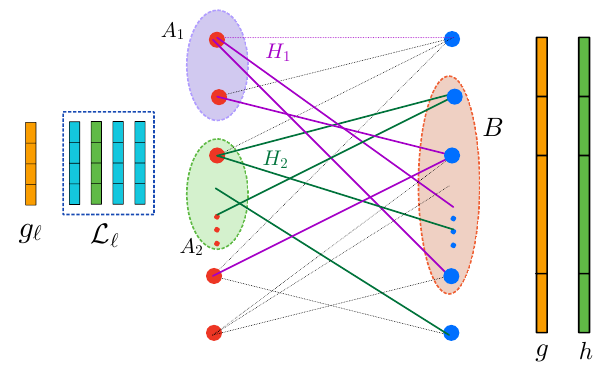}
\vspace{-10 pt}
\caption{Agreement graphs for list decoding AEL codes}
\vspace{-10 pt}
\label{fig:ael-graphs-intro}
\end{figure}
\vspace{-5 pt}
\paragraph{Agreement graphs.}
Fix a codeword $h \in \AELC$ in the global list, and for $i \in [K]$, let $A_i \subseteq L$ be the
set where $h_{\ell}$ appears in the $i$-th position in the local list. 
We will also ignore the $\gamma$ fraction of vertices where $h_{\ell} \notin \calL_{\ell}$ and assume
$A_1, \ldots, A_K$ form a partition of $L$.
Note that picking $h_{\ell}$ from the the local lists is equivalent to searching for the sets $A_1, \ldots,
A_K$. 
We will show that they can be found using the regularity lemma, by characterizing their
edge-densities in certain ``agreement graphs''.

For $i \in [K]$, define $H_i$ to be the graph with the set of edges where the $i$-th element of
$\calL_{\ell}$ agrees with $g$ 
\[
E(H_i) ~=~ \{e = (\ell, r) \in E(G) ~|~ (\calL_{\ell}[i])_e = g_e \} \mper
\]
Note that while we do not know the sets $A_i$, we can easily compute the graphs $H_i$ in linear
time, as we know $g$ and the local lists $\calL_{\ell}$.
Also, if $B = \inbraces{r \in R ~|~ g_r =h_r}$ denotes the set of right vertices where $g$ and $h$ agree, then we must have $E_{H_i}(A_i,B) = E_G(A_i,B)$ since for every $e = (\ell,r) \in E_G(A_i,B)$,
we have $(\calL[i])_e = h_e = g_e$.
Thus, while $H_i$ might have fewer edges than $G$, none of the edges in $E_G(A_i,B)$ are deleted, which allows the set $A_i$ to stand out in the regularity decomposition.

If $\inbraces{(c_{i,t}, S_{i,t}, T_{i,t})}_{t \in [p]}$ is a regularity decomposition of the graph
$H_i$, then $\abs{E_{H_i}(A_i,B)}$ can also be characterized using the intersection sizes
$\inbraces{\abs{A_i \cap S_{i,t}}, \abs{B \cap T_{i,t}}}_{t \in [p]}$. 
%
%
In fact, if we consider the family $\family$ of  all the sets $\inbraces{S_{i,t}, T_{i,t}}_{i \in [K], t \in [p]}$ 
appearing the regularity decompositions of all the graphs $H_i$, then the intersection sizes of of any $S$ and $T$ with all sets in the family $\family$, simultaneously determine $\abs{E_{H_i}(S,T)}$ for all $i \in [K]$.
%
%
\vspace{-5 pt}
\paragraph{Rigidity.}
Let $\family$ be the family of all the sets appearing the regularity decompositions of all
the agreement graphs $H_i$. 
For a set $S$, we refer to the intersection sizes of $S$ with all sets in $\family$ (normalized to be between 0 and 1, or correlations when viewed as functions) as the {\deffont signature} $\sigma(S)$ of $S$ with respect to the family $\family$. 
We will show the signatures essentially characterize the sets
$\inbraces{A_i}_{i \in K}$ \ie any collection of sets $\inbraces{A_i'}_{i \in [K]}$ which have (approximately) the same signatures, must actually be close to the sets $\inbraces{A_i}_{i \in K}$. 
%
%
This ``combinatorial rigidity'' will allow us to search for the sets $\inbraces{A_i}_{i \in K}$ by
simply searching for their \emph{signatures}. 
By choosing an appropriate level of discretization $\eta$ for the densities, it will be sufficient to
enumerate over a net of size at most $(1/\eta)^{\abs{\family}} \leq (1/\eta)^{{2pK}}$, where the
parameters $\eta, p, K$ are all constants independent of $n$.

The rigidity property follows from the fact that being close in their signatures
implies that $E_{H_i}(A_i', B) \approx E_{H_i}(A_i,B) = E_G(A_i,B)$ for all $i \in [K]$, using the regularity
lemma. Formally, we get
\[
\sum_{i \in [K]} \abs{E_{H_i}(A_i',B)} 
~~\geq~~ \abs{E_{G}(\cup_i A_i,B)} - K \cdot \gamma \cdot dn 
~~\geq~~ d \cdot \abs{B} - K\cdot\gamma \cdot dn \mcom
\]
for sufficiently close densities of $A_i'$ and $A_i$, and error parameter $\gamma$ from the
regularity lemma. Thus, a typical vertex in $L$, contributes $\approx d \cdot \abs{B}/n = d \cdot (1
- \beta)$ edges to the above sum.
On the other hand, any vertex in $A_i' \setminus A_i$ can contribute at most $d \cdot (1 - \delta(\cC_{\inn}))$
edges in $E_{H_i}(A_i',B)$. This is because $\calL_{\li}[i]$ and $h_{\ell}$ must differ on such a vertex on at least $\delta(\cC_{\inn}) \cdot d$ edges, and these cannot go to $B$ (where $g$ and $h$ agree) in the graph $H_i$ (where $\calL[i]$ and $g$ agree).
Edge counting using the expander mixing lemma yields
\[
\sum_{i \in [K]} \abs{E_{H_i}(A_i',B)} ~~\leq~~ \inparen{\sum_i \abs{A_i' \cap A_i}} \cdot \abs{B}
\cdot \frac{d}{n} + \lambda \cdot dn + \inparen{\sum_i \abs{A_i'\setminus A_i}} \cdot (1 - \delta(\cC_{\inn}))
\cdot d \mper 
\]
Comparing the two bounds and using $|B| \geq (1-\beta)n$ gives $\sum_i \abs{A_i' \setminus A_i} \leq
((K\gamma+\lambda)/(\delta(\cC_{\inn}) - \beta)) \cdot n$, which can be made arbitrarily small by
choosing the regularity error $\gamma$ and $\lambda$ sufficiently small compared to
$\delta(\cC_{\inn}) - \beta \geq 2\eps$.
Note that we only needed rigidity for the left sets $A_1, \ldots, A_K$, and it will actually suffice for our applications to only consider signatures with respect to families of subsets of the left set $L$.
\vspace{-5 pt}
\paragraph{List decoding via enumeration.}
Given the above rigidity property, we just need to (approximately) find the correct signatures for the sets $A_1, \ldots, A_K$ corresponding to each global codeword $h \in \calL$.
Enumerating with a discretization $\eta$, such that the error in the edge densities (relative to $\abs{E(G)} = nd$) is bounded by $\gamma$ (which is chosen to be small relative to $\eps$), we are guaranteed to find sets $A_1', \ldots, A_K'$ sufficiently close to $A_1, \ldots, A_K$. We can then recover $h$ by unique decoding in $\cC_{\out}$ to correct for the small differences.
Note that the enumeration argument also implies a combinatorial bound of $(1/\eta)^{\abs{\family}}$ on the size of the list, since we proved that every list codeword must be found by searching this discrete net.

One subtle issue here is that we do need to produce sets $A_1', \ldots, A_K'$ corresponding to each point (signature) in the discrete net we search over, for the decoding algorithm to proceed. 
However, not every signature $\sigma$ (collection of intersection sizes with sets in $\family$) might actually be \emph{realizable} as $\sigma(S)$ for some set $S$, due to possibly complex overlaps between the sets in the family $\family$ which constrain the intersection patterns. 
But we look at the partition $\factor(\family)$ generated by all the sets in $\family$, then the density of $S$ in any cell of this partition is unconstrained by densities in other cells. 
We test the realizability of the signatures in our enumeration algorithm by expressing this as a simple linear program in densities corresponding to the cells of the partition. 

It will be more convenient to work in the language of ``factors'' instead of partitions for the enumeration argument, since it facilitates a cleaner error analysis. This is described in \cref{sec:prelims}. The error estimates for set (function) families arising from regularity lemma are developed in \cref{sec:reg_lemma}.
\vspace{-5 pt}
\paragraph{List decoding for Tanner codes.} 
For the case of Tanner codes, we are given a $g \in \Sigma^E$ and the goal is now to find codewords $h \in \TanC$, which differ from $g$ on at most $\beta$ fraction of \emph{edges}.
Since the code constraints are now on both the left and right vertices, we list decode the local projections $g_{\ell}$ and $g_r$ to the left and right vertices, to form local lists $\calL_{\ell}$ and $\calL_r$ for each $\ell \in L$ and $r \in R$. Given a bound of (say) $K$ on the size of the local lists, for each global codeword $h$, we now define sets $A_1, \ldots, A_K$ on the left and $B_1, \ldots, B_K$ on the right, based on the position of $h_{\ell}$ and $h_r$ in the local lists.
To find these sets, we consider agreement graphs $H_{ij}$ for each $i,j \in [K]$, based on the agreement of the $i$-th and $j$-th elements in the left and right lists (see \cref{fig:tanner-graphs})
\[
E(H_{ij}) ~~=~~ \inbraces{e = (\ell,r) \in E(G) ~|~ (\calL_{\ell}[i])_e = (\calL_{r}[j])_e} \mper
\]
As before, we can argue that $E_{H_{ij}}(A_i, B_j) = E_G(A_i,B_j)$ and prove rigidity by considering a partition given by the sets appearing in the regularity decompositions for all the graphs $H_{ij}$.
\begin{figure}[h]
\centering
\includegraphics[width=0.55\textwidth]{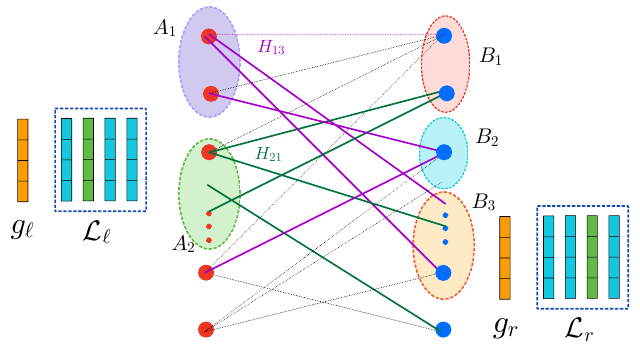}
\vspace{-10 pt}
\caption{Agreement graphs for list decoding Tanner codes}
\vspace{-10 pt}
\label{fig:tanner-graphs}
\end{figure}
The key difference is that now the sets $A_1, \ldots, A_K$ and $B_1, \ldots, B_K$ do not cover almost all the vertices, and we need to declare erasures in the locations not covered by these sets. 
The unique decoder for the outer code $\calC_{\out}$ in the case of AEL codes is then replaced by a unique decoder for $\TanC$, which can decode from erasures as well as (a small number of ) errors.
Such decoding is known for Tanner codes, and we include a proof in \cref{sec:appendix_tanner}.
The same argument also extends asymmetric variants of Tanner codes, when the left and right base
codes are different, say $\cC_1, \cC_2 \subseteq \Sigma^d$ with list size bounds $K_1, K_2$.
\vspace{-5 pt}
\paragraph{List recovery.}
An argument similar to the list decoding of Tanner codes, with two-sided local lists, can also be used for the list recovery of AEL codes. In this case, the lists for the right vertices $\inbraces{\calL_r}_{r \in R}$ (with $\abs{\calL_r} \leq k$) are actually the input for list recovery.
The local lists $\calL_{\li}$ are obtained via list recovery for the inner code $\calC_{\inn}$, with $\abs{\calL_{\ell}} \leq K$. We again set up agreement graphs $H_{ij}$ for each $i \in [K], j \in [k]$, and prove rigidity for the partitions given by regularity lemma. 
We believe a similar argument, with two-sided list recovery for base codes on left and right, may also give list recovery for Tanner codes. 
\vspace{-5 pt}
\paragraph{Remarks.}
We collect a few additional observations regarding our techniques and results below.
\begin{itemize}[leftmargin=0.5 cm]
\item The expander based constructions typically lift a property such as distance \cite{AEL95, GI05} or average-radius list decodability \cite{JMST25} from local code to the global code. Our results can be seen as lifting, both combinatorially and algorithmically, list size bounds from the local code to the global code. Given a collection of local lists, the regularity decomposition allows us to stitch together these lists in a globally consistent manner, as was done by \cite{HW15, RWZ21} in the erasures-only setting.
\item The only component of our algorithm that taken near-linear instead of linear time (or uses
  randomization), is the computation of the regularity decomposition. All other parts of our
  algorithm can be implemented in deterministic linear time. The regularity decomposition can also
  be computed in deterministic $O(n^{3.5})$ which also gives the same deterministic runtime for all
  our results. The randomized near-linear time computation can also be implemented in time $O(n \cdot (\log n)^2 \cdot (1/\eps)^{O(1)})$ using the approximate primal-dual SDP solvers of Arora and Kale~\cite{AK07}. Their solvers require $O(\log n)$ iterations of matrix multiplicative updates, with each step projecting SDP solutions to $O(\log n)$ random directions.
%
\item The only step which depends on the alphabet size of the code, is the computation of the local lists for the code $\calC_{\inn}$ in AEL, and the base code $\calC_0$ in Tanner codes. Given the local lists, the rest of the algorithm only depends on the size of these lists and is independent of the alphabet size.
\item Using standard concatenation methods, our list recovery algorithms can also be used to obtain binary codes which are list decodable up to the Zyablov bound in near-linear time.
\item Given a list recovery instance for the Tanner code, one can perform list recovery for the local code to obtain local lists, and then extract the list codewords from them in a way similar to list decoding. We omit the details of this list recovery algorithm in the Tanner code's case since it follows by a relatively straightforward extension of the list decoder.
\end{itemize}
\vspace{-5 pt}
\paragraph{Algebraic vs. combinatorial rigidity.}
As discussed above, regularity lemmas lead to rigidity properties for the sets corresponding to list codewords. 
These can be compared to algebraic rigidity phenomena which underlie decoding algorithms for codes based on polynomials, such as Reed-Solomon codes. 
In the algebraic setting, one often sets up a multivariate polynomial interpolation problem, which can be solved by solving linear equations depending on the input. One then needs to argue that the list codewords must always be roots of such a multivariate polynomial, using the fact that low-degree polynomials have some algebraic rigidity (once they vanish at sufficiently many points, they must be zero everywhere).

The arguments using regularity can be considered analogues of the above interpolation-based list decoding arguments. Instead of linear equations, one considers edge density conditions that must be satisfied by the sets corresponding to the true codewords. We consider a strengthening of this condition by matching not only edge densities, but all densities in the regularity partition. At this point, the regularity conditions and the local distances of the codes are sufficient to imply rigidity, saying that sets matching in densities according to the partition, should in fact match almost everywhere. 
%

\subsection{Related work}
The rigidity argument used in our work is also similar in spirit to the rigidity argument used by
Bhowmick and Lovett~\cite{BL18} for obtaining bounds on the list size for Reed-Muller codes, at
radii arbitrarily close to distance. 
They develop (much stronger) regularity lemmas for polynomials using tools from higher order Fourier analysis, and show that any list codeword must be ``measurable'' in a factor (partition) determined by a constant number of polynomials (independent of the block length). Counting the number of functions that are measurable in the factor then yields the list size bound in their setting. 
However, the regularity lemmas in their setting require significantly stronger properties, and yield
tower-type bounds on the list size, while our argument only relies on the weak regularity lemma and
yields exponential bounds.

Our results can also be seen as sparse analogues of the results by Gopalan, Guruswami and Raghavendra~\cite{GGR09}, for list decoding tensor codes and interleaved codes. 
When $G$ is taken to be the complete bipartite graph, the resulting Tanner code is simply the tensor
code $\calC_0 \otimes \calC_0$, and more generally $\cC_1 \otimes \cC_2$ when the left and right
codes are different.
Their algorithm for tensor codes chooses from the local (row and column) lists by randomly sampling rows and columns, while our algorithm extends it to sparse expanders via regularity.
%
%
Thinking of consistently picking from local lists as a constraint satisfaction problem (CSP) this
can be seen as analogous to literature on approximation of CSPs, where sampling-based algorithms for dense
instances~\cite{AKK95, FK96:focs}, can be extended to sparse expanding ones via regularity~\cite{OGT15}. 

As mentioned earlier, Guruswami and Indyk~\cite{GI03} also gave the first explicit construction of
expander-based codes, which are list decodable from $1 - \eps$ fraction of errors in (randomized)
\emph{linear} time. 
Their approach also relies on consistently choosing from $n$ local lists corresponding to vertices
of an expander, in time $O(n \log n)$, independent of the alphabet for the local codes. 
They then achieve linear run-time by using an alphabet of size $n^{O(1)}$ and concatenating, so that
the block size becomes $N = O(n \log n)$. However, their constructions use a variant of AEL with the
repetition code as the inner code~\cite{ABNNR92}, which causes the rate to decay as
$\exp(-\exp(1/\eps^3))$.
A similar concatenation with our AEL construction may also lead to linear-time list decodable codes
with improved rates, using the fact that our algorithm's runtime does not directly on the alphabet size (see remark above).

Another closely related work is that of~\cite{JST21}, which uses analogues of the weak regularity lemma for expanding (ordered) hypergraphs, to obtain a list decoding algorithms for Ta-Shma $\eps$-balanced code. 
Similar ideas were also used by Blanc and Doron~\cite{BD22} for near-linear time list decoding of a different (probabilistic) construction of an $\eps$-balanced code.
The key challenge in these applications is actually developing algorithmic versions of the
hypergraph regularity lemmas to do any nontrivial list decoding (both above results are
significantly below the Johnson bound), while the focus in our work is to using known graph
regularity lemmas to list decode from an \emph{optimal} fraction of errors.

There is also a very rich literature of algebraic code constructions with efficient list decoding
near capacity, and other interesting properties, that we do not exhaustively list here. 
It was shown by Goyal \etal~\cite{GHKS24} that folded Reed-Solomon codes and univariate multiplicity
codes can be list decoded to capacity in randomized near-linear time. 
Guo and Ron-Zewi~\cite{GRZ21} also obtained a construction with constant alphabet size and constant
list-sizes, that be list decoded $\eps$-close to capacity in $\mathrm{poly} (n, 1/\eps)$ time.\footnote{They
 output a subspace of dimension $\mathrm{poly}(1/\eps)$ containing the list in $\mathrm{poly}(n, 1/\eps)$ time, but pruning takes
  $\exp(1/\eps)$ time.}
Hemenway, Ron-Zewi, and Wootters~\cite{HRW17} also gave constructions of codes which are list decodable
to capacity in near-linear time (with list sizes growing with $n$), and even locally list decodable
and locally list recoverable, using repeated tensoring of algebraic codes.
This was also extended by Kopparty \etal~\cite{KRRSS21} to achieve \emph{deterministic} near-linear
time decoding (with alphabet and list size growing as functions of $n$).
%
%

There has also been work trying to optimize list size near list decoding capacity, and several random families of codes are known to achieve the optimal list size, such as random linear codes \cite{AGL24}, randomly punctured Reed-Solomon codes \cite{BGM23, GZ23, AGL24}, randomly punctured algebraic-geometric codes \cite{BDGZ24}. However, no polynomial-time list decoding algorithms are known for these random constructions.

Finally, our work can also be seen as a refinement of decoding algorithms using the SoS
hierarchy~\cite{JST23, JMST25} based on a ``proofs to algorithms'' paradigm. The algorithms using
SoS rely on the fact that a broad family of (sum-of-squares) proofs can be converted to algorithms
using programs in the SoS hierarchy. On the other hand, our framework captures the fact that a
subclass of these proofs based on edge-counts using expander mixing, can be turned into (faster)
algorithms using the regularity lemma.





%% file: prelims.tex
\section{Preliminaries}\label{sec:prelims}

%
For a bipartite graph $G=(L,R,E)$ where $L$ is the set of left vertices, and $R$ is the set of right
vertices, we index the left set by $\li$ and the right set by $\ri$. For a vertex $\li \in L$, 
we denote the set of edges incident to it by $N(\li)$ (left neighborhood), and the set of edges
incident to  $\ri\in R$ is denoted by $N(\ri)$ (right neighborhood). 
We use $\li \sim \ri$ to denote that the vertex $\li \in L$ is adjacent to the vertex $\ri \in R$, that is, $(\li,\ri)\in E$. For the graph $G$, we will also use $A_G$ to denote its $L\times R$ biadjacency matrix.

Fix an arbitrary ordering of the edges. Then there are bijections between the sets $E$, $L \times
[d]$, and $R \times [d]$, 
given by taking $(\li,i)$ to be the $i^{th}$ edge incident on $\li$, and similarly for $R \times [d]$.
%
Henceforth, we will implicitly assume such an ordering of the edges is fixed, and use the resulting bijections.


\subsection{Codes and distances}
\begin{definition}
	Let $\Sigma$ be a finite alphabet and let $g,h\in \Sigma^n$. Then the (fractional) distance
        between $g,h$ is defined as \[ \dis(g,h) = \Pr{i\in [n]}{ g_i \neq h_i} \mper \]
\end{definition}

Let $G=(L,R,E)$ be a bipartite graph. For any $g\in \Sigma^E$ and $\li \in L$, we use $g_{\li}$ to denote the restriction of $g$ to the edges in $N(\li)$, that is, $g_{\li} \defeq g\big\vert_{N(\li)} \in \Sigma^{N(\li)}$. Similarly, for any $g\in \Sigma^E$ and $\ri\in R$, we use $g_{\ri} \defeq g\big\vert_{N(\ri)} \in \Sigma^{N(\ri)}$.

For any $g,h \in \Sigma^E$, note that $\Delta(g,h)$ is simply $\Pr{e\in E}{g_e\neq h_e}$. For such $g$ and $h$, we will also use the following two distances.
\begin{definition}[Left and right distances]
	Let $G=(L,R,E)$ be a bipartite graph. Let $\Sigma$ be a finite alphabet and let $g,h\in \Sigma^E$. Then the left distance and the right distance
        between $g,h$ are denoted as $\Delta_L(g,h)$ and $\Delta_R(g,h)$, and defined respectivel y as  
        \[ \Delta_L(g,h) = \Pr{ \li \in L}{ g_{\li} \neq h_{\li}} \qquad \text{and} \qquad \Delta_R(g,h) = \Pr{\ri \in R}{g_{\ri} \neq h_{\ri}} \mper \]
\end{definition}

\begin{definition}[Code, distance and rate]
	A code $\calC$ of block length $n$, distance $\delta$ and rate $\rho$ over an alphabet $\Sigma$ is a set $\calC \subseteq \Sigma^n$ with the following properties
\[
\rho = \frac{\log_{|\Sigma|} |\calC|}{n} \qquad \text{and} \qquad \delta = \min_{h_1,h_2\in \calC} \dis(h_1,h_2) \mper
\]
We say $\calC$ is a linear code if $\Sigma$ can be identified with a finite field $\F_q$ and $\calC$ is a linear subspace of $\F_q^n$.
\end{definition}


\begin{definition}[List of codewords]
	For any $g\in \Sigma^n$, the list of codewords in $\calC$ that are at a distance at most
        $\beta$ from $g$ is denoted by $\calL(g,\beta)$. That is $\calL(g,\beta) = \inbraces{h\in \calC ~\vert~ \dis(g,h) \leq  \beta }$.
\end{definition}
%

\subsection{Expander graphs}
\begin{definition}[$(n,d,\lambda)$-expander]
A $d$-regular bipartite graph $G(L,R,E)$ with $|L|=|R|=n$ is said to be an $(n,d,\lambda)$-expander if
\[
	\sigma_2(A_G) \leq \lambda \cdot d
\]
where $A_G$ is the $L\times R$ biadjacency matrix, and $\sigma_2(A_G)$ is its second largest singular value.
\end{definition}

Infinite explicit families of $(n,d,\lambda)$--expanders, with growing $n$ as $d$ and $\lambda$ are constant, can be derived based on double covers of Ramanujan graphs of \cite{LPS88} as long as $\lambda \geq \frac{2\sqrt{d-1}}{d}$.

\begin{lemma}[Expander Mixing Lemma]
	Given an $(n,d,\lambda)$-expander $G=(L,R,E)$, the following well known property is a simple consequence of definition of $(n,d,\lambda)$-expanders:
	\[
		\forall S\sub L, T\sub R,\qquad \abs*{|E_G(S,T)| - \frac{d}{n}\cdot |S|\cdot |T|} ~\leq~ \lambda \cdot d \sqrt{|S|\cdot |T|} ~\leq~ \lambda \cdot dn \mper
	\]
\end{lemma}

%% file: enumeration.tex
\subsection{Factors, signatures, and enumeration}
%
%
It will be convenient to use the language of factors, to work with the decompositions identified by regularity lemmas.
This concept (from ergodic theory) takes a rather simple form in our finite settings: a factor over a base set $\calX$ is just a partition of $\calX$, but with an associated operation of averaging functions $f:\calX \to \R$ over each cell of the partition, which makes it quite convenient to work with.
\begin{definition}[Factors and measurable functions]
Let $\calX$ be a finite set. A {\deffont factor} $\factor$ is a partition of the set $\calX$, and
the cells of the partition are referred to as \deffont{atoms} of the factor. For $x \in \calX$, we
also use $\factor(x)$ to denote the atom $A \in \factor$ containing $x$. 
A function $f: \calX \to \calR$ is said to {\deffont measurable} with respect to $\factor$ ($\factor$-measurable) if $f$ is constant on each atom of $\factor$.
\end{definition}

\begin{definition}[Conditional averages]
If $f: \calX \to \R$ is a function and $\factor$ is a factor of $\calX$, then we define the {\deffont conditional average function} $\Ex{f|\factor}:\calX \to \R$ as
\[
\Ex{f|\factor}(x) ~\defeq~ \Ex{y \in \factor(x)}{f(y)} \mcom
\]
where $\factor(x)$ denotes the atom containing $x$.
\end{definition}

\begin{proposition}
	For any $f: \calX \to \R$, the conditional average function $\Ex{f|\factor}$ is measurable with respect to $\factor$. Further, $f$ itself is $\factor$-measurable if and only if $f = \Ex{f|\factor}$.
\end{proposition}
Note a function measurable in $\factor$ and can equivalently be thought of as a
function over the domain $|\factor|$. 
However, it will be convenient to treat these as functions over $\calX$ (which are constant in every atom)
when writing inner products and norms.
All norms and inner products below are using the expectation measure on $\calX$.
\begin{proposition}\label{prop:measurable-inner-product}
  Let $h: \calX \to \R$ be a $\factor$-measurable function, and let $f: \calX \to \R$ be any function. Then,
  \[
    \ip{h}{f} ~=~ \ip{h}{\Ex{f|\factor}}\mper
  \]
\end{proposition}
\begin{proof}
The proof simply follows by noticing that we can sample a random element in $\calX$ by first sampling a random atom of $\factor$ (with probability of an atom proportional to its size) and then sampling a random element of the atom. This gives
\[
\ip{h}{f} 
~=~ \Ex{y \in \calX}{h(y) \cdot f(y)} 
~=~ \ExpOp_{x \in \calX} \Ex{y \in \factor(x)}{h(y) \cdot f(y)} 
~=~  \Ex{x \in \calX}{h(x) \cdot \Ex{y \in \factor(x)}{f(y)}}
~=~ \ip{h}{\Ex{f|\factor}} \mcom
\]
where we used that $h(x) = h(y)$ for all $y \in \factor(x)$, since $h$ is $\factor$-measurable.
\end{proof}
The factors we will consider will be defined by a finite family of functions appearing in a regularity decomposition.
\begin{definition}[Function factors]
Let $\calX$ be a finite set and $\family = \inbraces{f_1, \ldots, f_p: \calX \to \B}$ be a family of boolean functions on $\calX$. We consider the factor $\factor(\family)$ defined by the functions $f_1, \ldots, f_p$ in $\family$, as the factor with atoms $\inbraces{x ~|~ f_1(x) = b_1, \ldots, f_p(x) = b_p}$ for all $(b_1,\ldots, b_p) \in \B^p$.
When the family $\family$ is clear from context, we will refer to the function factor simply as $\factor$, with the dependence on the family $\family$ being implicit.
\end{definition}

\begin{remark}\label{remark:factor_sigma_algebra_comp}
The above function can also be thought of as indicators for sets $S_1, \ldots, S_p$. Then the function factor $\factor$ is the same as the $\sigma$-algebra generated by these sets, and has at most $2^p$ many atoms.
Also, given the functions $f_1,\ldots,f_p$ as above, the function factor $\factor(\family)$ can be computed in time $O(\abs{\calX} \cdot 2^p)$.
\end{remark}

For our algorithms, we will need to enumerate over $\factor$-measurable functions, which contain a function close to any function of interest. We define the associated notion of distance and coverings below.
\begin{definition}[Distances] 
For a family of functions $\family$  on $\calX$, we define $\family$-distance between two functions $f_1, f_2:
\calX \to \R$ as the maximum difference in their correlations with functions from $\family$
\[
\norm{f_1 - f_2}_{\family} 
~\defeq~ \sup_{f \in \family}~\abs{\ip{f}{f_1} - \ip{f}{f_2}} ~=~ \sup_{f \in \family}~\abs{\ip{f}{f_1 - f_2}} \mper
\] 
For a factor $\factor$, we also define a notion of $\factor$-distance between $f_1$ and $f_2$ as the distance between their conditional expectations
\[
\norm{f_1 - f_2}_{\factor}
~\defeq~ \norm*{\Ex{f_1|\factor} - \Ex{f_2|\factor}}_1 \mper
\]
\end{definition}
The following claim relates the two notions of distance defined above.
\begin{claim}
For a family $\family \subseteq \B^{\calX}$ and function factor $\factor(\family)$, we have $\norm{f_1 - f_2}_{\family} \leq \norm{f_1 - f_2}_{\factor(\family)}$.
\end{claim}
\begin{proof}
For any function $f \in \family$ we have,
\[
\abs{\ip{f}{f_1 - f_2}} ~=~ \abs*{\ip{f}{\Ex{f_1|\factor}} - \ip{f}{\Ex{f_2|\factor}}} ~\leq~ \norm*{f}_{\infty} \cdot \norm*{\Ex{f_1|\factor} - \Ex{f_2|\factor}}_1 ~\leq~ \norm{f_1 - f_2}_{\factor(\family)} \mper \qedhere
\]
\end{proof}
We also need a small covering net which contains a point close in the above distance, to any
function with $f$ with bounded $\ell_{\infty}$ norm. Since we will only need to use these for
indicator functions of subsets, we define the following covering next for functions $f: \calX \to [0,1]$.
While one can define a net over the set of $\factor$-measurable functions in the distance $\norm*{\cdot}_{\factor}$, the size of the net is exponential in $\abs{\factor}$ and thus doubly exponential in $\abs{\family}$. 
We will instead search over a smaller net in the distance $\norm*{\cdot}_{\family}$ with size ``only'' exponential in $\abs{\family}$, which will suffice for our purposes. We next define the objects we search over.
\begin{definition}[Signatures]
Let $\factor{(\family)}$  be a function factor defined by a finite function family $\family = \inbraces{f_1, \ldots, f_p: \calX \to \B}$. For a $f: \calX \to [0,1]$, we call the vector of inner products $\sigma(f) = \inbraces{\ip{f}{f_j}}_{j \in [p]}$ the {\deffont signature} of $f$. 
We call a vector $\sigma \in [0,1]^p$ an $\eta$-{\deffont realizable signature} if there exists a $\factor(\family)$-measurable function $f: \calX \to [0,1]$ such that $\max_{j \in [p]} \abs{\sigma_j - \ip{f}{f_j}} = \norm{\sigma - \sigma(f)}_{\infty} \leq \eta$. 
\end{definition}
Note that since each $f_j$ is $\factor(\family)$-measurable, the condition that $f$ is also $\factor(\family)$ measurable is without loss of generality, by \cref{prop:measurable-inner-product}.
%
%
\begin{definition}[Covering Nets for Signatures]
%
Let $D_{\eta} = \{0, \eta, 2\eta, \ldots, 1\}$ and $\family = \inbraces{f_1, \ldots, f_p}$. 
We define the covering net $\calN_{\eta}(\family)$ to be the set of all signatures $\sigma \in D_{\eta}$, which are $\eta$-realizable using $\factor$-measurable functions.
Then, for any $f: \calX \to [0,1]$, there exists $\sigma \in \calN_{\eta}(\family)$ such that $\norm{\sigma - \sigma(f)}_{\infty} \leq \eta$.
We will also need coverings for $K$-tuple of functions $f_1, \ldots, f_K: \calX \to [0,1]$ satisfying $f_1 + \cdots + f_K \leq 1$ (corresponding to disjoint subsets of $\calX$). 
We define the covering net $\calN_{\eta}(\family,K)$ to be the set of all $K$-tuples of signatures $\sigma = (\sigma_1, \ldots, \sigma_K) \in D_{\eta}^{\abs{\family} \cdot K}$, which are $\eta$-realizable using functions measurable in $\factor = \factor(\family)$.
%
\[
\calN_{\eta}(\family,K) ~\defeq~ \inbraces{ \sigma = (\sigma_1, \ldots, \sigma_K) \in D_{\eta}^{\abs{\family} \cdot K} \quad \middle| \quad 
\begin{split}
&\exists~ \bar{f}_1, \ldots, \bar{f}_K: \factor \to [0,1] ~~\text{s.t.}~~ \\
&\sum_i \bar{f}_i \leq 1  \quad \text{and} \quad \forall i \in [K], ~\norm{\sigma_{i} - \sigma(\bar{f}_i)}_{\infty} \leq \eta
\end{split}
} \mper
\]
Note that we have $\abs{\calN_{\eta}(\family,K)} \leq (1/\eta + 1)^{K \cdot \abs{\family}} = (1/\eta + 1)^{K \cdot p}$.
\end{definition}
The following claim captures the equivalence of points in the covering net $\calN_{\eta}(\family, K)$ and collections of disjoint subsets of $\calX$. 
 Note that the bound $\ffrac{\abs{\factor}}{\abs{\calX}}$ in the second part should be thought of as negligible, since we will deal with factors where the number of atoms $\abs{\factor}$ is constant, while $\abs{\calX} = n$ is arbitrarily large. 
\begin{claim}\label{clm:enumeration}
Let $A_1, \ldots, A_K$ be $K$ disjoint subsets of $\calX$. Then there exists $\sigma \in \calN_{\eta}(\family, K)$ such that for all $i \in [K]$, we have $\norm{\sigma_i - \sigma(\indicator{A_i})}_{\infty} \leq \eta$. 
Also, for any $\sigma \in D_{\eta}^{\abs{\family} \cdot K}$, it is possible to test in time $O(\parens{K \cdot \abs{\factor(\family)}}^{O(1)})$ if $\sigma \in \calN_{\eta}(\family,K)$ and find a set of functions $\bar{f}_1, \ldots, \bar{f}_K$ realizing the $\sigma$ as above, if they exist.
Moreover, for any such functions $(\bar{f}_1, \ldots, \bar{f}_K)$, there exist disjoint sets $S_1, \ldots, S_K \subseteq \calX$ constructible in time $O(\abs{\calX})$, such that $\max_{i \in [K]} \norm{\sigma_i - \sigma(\indicator{S_i})}_{\infty} \leq \eta + \ffrac{\abs{\factor}}{\abs{\calX}}$.
\end{claim}
\begin{proof}
Given disjoint sets $A_1, \ldots, A_K$, the functions $\bar{f}_i = \Ex{\indicator{A_i}|\factor}$  are measurable in $\factor$ for each $i \in [K]$, and satisfy $\sum_{i} \bar{f}_i \leq 1$. Moreover, for each $i \in [K]$ and $f_j \in \calF$, defining $\sigma_{ij}$ as
\[
\sigma_{ij} ~=~ \left\lfloor \frac{1}{\eta} \cdot \ip{\Ex{\indicator{A_i}|\factor}}{f_j} \right\rfloor \cdot \eta \mcom 
\]
yields $\sigma \in \calN_{\eta}(\family,K)$ with the required properties.

Conversely, given $\sigma \in D_{\eta}^{\abs{\family} \cdot K}$, we consider the linear program in the $K \cdot \abs{\factor}$ variables corresponding to the values of the $\factor$-measurable functions $\bar{f}_1, \ldots, \bar{f}_K$ in each of the $\abs{\factor}$ atoms, and with the constraints
\[
\bar{f}_i \in [0,1], ~~\norm{\sigma_i - \sigma(\bar{f}_i)}_{\infty} \leq \eta \quad\forall i \in [K]
\qquad \text{and} \qquad \sum_i \bar{f}_i \leq 1  \mper
\]
The linear program can be solved\footnote{The linear program can be solved to give solutions up to a small error, which can be absorbed in the final error bound of $O(\eta)$. We suppress this here for clarity.}
in time $(K \cdot \abs{\factor})^{O(1)}$ to find functions $\bar{f}_1, \ldots, \bar{f}_K$ certifying $\sigma \in \calN_{\eta}(\family,K)$ if they exist.
Given $\factor$-measurable functions $(\bar{f}_1, \ldots, \bar{f}_K) \in \calN_{\eta}(\factor,K)$, in each atom $P \in \factor$, we can assign $\lfloor \abs{P} \cdot \bar{f}_i \rfloor$ distinct elements to each $S_i$ (say according to some fixed ordering over all elements) since $\sum_i \bar{f}_i \leq 1$. We then have for each atom $P$, $\abs*{\Ex{x \in P}{\bar{f}_i-\indicator{S_i}(x)}} \leq \frac{1}{\abs{P}}$, which gives 
\[
\norm{\bar{f}_i - \indicator{S_i}}_{\factor} 
~=~ \ExpOp_{x \in \calX} \abs*{\Ex{x \in \factor(x)}{\bar{f}_i-\indicator{S_i}(x)}} 
~\leq~ \Ex{x \in \calX}{\frac{1}{\abs{\factor(x)}}} 
~=~ \frac{\abs{\factor}}{\abs{\calX}}  \mper 
\]
Finally, we note that for each $i \in [K]$,
\[
\norm{\sigma_i - \sigma(\indicator{S_i})}_{\infty} 
~\leq~ \norm{\sigma_i - \sigma(\bar{f}_i)}_{\infty} + \norm{\sigma(\bar{f}_i) - \sigma(\indicator{S_i})}_{\infty} 
~\leq~ \eta + \norm{\bar{f}_i - \indicator{S_i}}_{\family} 
~\leq~ \eta + \norm{\bar{f}_i - \indicator{S_i}}_{\factor}  \mper
\]
Combining with the previous bound of $\norm{\bar{f}_i - \indicator{S_i}}_{\factor} ~\leq~ \frac{\abs{\factor}}{\abs{\calX}}$ then proves the claim.
\end{proof}



%% file: reg_lemma.tex
\section{Regularity Lemmas for Subgraphs of Expanders}\label{sec:reg_lemma}
We collect the statements we need about the regularity lemma in this section. The proofs below are adaptations of well-known arguments in the literature, and readers familiar with the literature can simply refer to the statements in \cref{sec:regularity-factor} which capture the properties we need. 
\newcommand{\decom}{\mathcal{D}}
\newcommand{\mat}{{M}}
\newcommand{\mask}[1]{\insquare{#1}_{{E(G)}}}
\newcommand{\muG}{\mu_G}

We first prove the following existential statement about the existence of a regularity decomposition $\decom$ for any subgraph $H$ of $G = (L,R,E)$, which is simply a concrete version of an abstract regularity lemma from~\cite{TrevisanTV09}.
For any matrix $\mat \in \R^{L \times R}$, we use $\mask{\mat} = \mat \circ A_G$ to denote the restriction of its entries to the edges in $G$ (such restrictions will be useful to ensure sparsity for our algorithms). Also, we use $\muG$ to denote the uniform measure on the edges of $G$, and will need to take inner products and norms with respect to this measure.
\begin{theorem}[Existential Regularity Lemma]\label{thm:existential_reg}
Let $G = (L,R,E)$ be a bipartite graph and let $H$ be a subgraph of $G$. Then for any $\gamma > 0$, there exists a collection of $p\leq 1/\gamma^2$ triples $(c_i,S_i,T_i) \in \R \times 2^L \times 2^R$ such that for any $S\sub L$ and $T\sub R$,
	\[
		\abs*{ |E_H(S,T)| - \sum_{i=1}^p c_i \cdot |E_G(S_i\cap S, T_i\cap T)| } \leq \gamma \cdot |E_G(L,R)|
	\]
	Further, the collections of reals $\{c_i\}_{i\in [p]}$ satisfy $\sum_i |c_i| \leq 1/\gamma$.
\end{theorem}
\begin{proof}
	The following (possibly inefficient) procedure finds such triples $(c_i,S_i,T_i)$.
	
	\begin{algorithm}{Existential Regularity Decomposition}{$\gamma > 0$, subgraph $H$ of $G = (L,R,E)$}{$\decom = \{(c_i, S_i, T_i)\}_{i\in [p]}$, with $ c_i \in\R, S_i \sub L, T_i \sub R$ for every $i\in [p]$} \label{algo:existential_regularity}
     \begin{itemize}
      \item Let $\mat_0 = A_H$, $j=1$ and $\decom = \emptyset$
      \item While $\exists ~ S_j \sub L,T_j \sub R$ such that $\abs*{\one_{S_j}^{\trans} \mat_{j-1}\one_{T_j}} > \gamma \cdot |E_G(L,R)|$:
    
  \begin{itemize}
			\item Let $c_j = \gamma \cdot \sgn\parens*{\one_{S_j}^{\trans} \mat_{j-1}\one_{T_j}}$
			\item $\mat_j \gets \mat_{j-1} - c_j \cdot \mask{\one_{S_j} \one_{T_j}^{\trans}}$
            \item $\decom \gets \decom \cup \{(c_j,S_j,T_j)\}$
            \item $j = j + 1$
          \end{itemize}
      \item Let $p = j-1$
      \item Return $\decom$
     \end{itemize}
\end{algorithm}

We relate the term $\one_S^{\trans} \mat_p \one_T$ to the edge counts as
\[
\one_S^{\trans} \mat_p \one_T ~=~ \one_S^{\trans} \inparen{A_H - \sum_{i-1}^p c_i \mask{\one_{S_i}\one_{T_i}^{\trans}}} \one_T ~=~ E_H(S,T) - \sum_{i=1}^p c_i E_G(S_i \cap S, T_i\cap T) \mper
\]
Also note that when the algorithm terminates, it holds that
\[
	\forall S \sub L, T\sub R, \qquad \abs*{\one_S^{\trans} \mat_{p} \one_T} \leq \gamma \cdot |E_G(L,R)| \mper
\]
Thus it suffices to show that the algorithm must terminate in at most $1/\gamma^2$ steps. We note that
\begin{align*}
\norm{\mat_{j-1}}_{\muG}^2 - \norm{\mat_{j}}_{\muG}^2 
&~=~ 2c_j \cdot \ip{\mat_{j-1}}{\mask{\one_{S_j} \one_{T_j}^{\trans}}}_{\muG} -~ c_j^2 \cdot \norm*{\mask{\one_{S_j} \one_{T_j}^{\trans}}}_{\muG}^2 \\
&~\geq~ 2c_j \cdot \frac{\one_{S_j}^{\trans} \mat_{j-1} \one_{T_j}}{\abs{E(G)}} -~ c_j^2 ~\geq~ 2\gamma^2 - \gamma^2  ~=~ \gamma^2 \mper
\end{align*}
Since the norm decreases by $\gamma^2$ in each step, and we have $\norm{\mat_0}_{\muG}^2 = \Ex{e \in G}{\mathds{1}_{\{e\in H\}}^2} \leq 1$, the algorithm must terminate in at most $1/\gamma^2$ steps. Finally, we also have $\sum_i \abs{c_i} \leq \gamma \cdot (1/\gamma^2) \leq 1/\gamma$.
\end{proof}
A reader may notice that the above proof is simply the analysis of gradient descent for minimizing the
convex function 
\[
F(M) ~=~ \sup_{S,T} ~\abs*{~\ip{M}{\one_S \one_T^{\trans}}_{\muG}} \mcom
\]
starting from the point $\mat_0 = A_H$. The minimum value is of course zero, but the proof follows
from the fact that gradient descent gets $\gamma$-close to the optimum in $1/\gamma^2$ steps with updates of the form $\one_{S_i} \one_{T_i}^{\trans}$.

\begin{corollary}
	Let $G$ be a bipartite $(n,d,\lambda)$-expander, and let $H$ be a subgraph of $G$. Then for any $\gamma >0$, if it holds that $\lambda\leq \gamma^2/4$, then there exists a collection of $p\leq 4/\gamma^2$ triples $(c_i,S_i,T_i) \in \R \times 2^L \times 2^R$ such that for any $S\sub L$ and $T\sub R$,
	\[
		\abs*{ |E_H(S,T)| - \frac{d}{n} \sum_{i=1}^p c_i \cdot |S_i\cap S|\cdot |T_i\cap T| } \leq \gamma \cdot nd
	\]
	The collection of reals $\{c_i\}_{i\in [p]}$ satisfy $\sum_i |c_i| \leq 2/\gamma$.
\end{corollary}

\begin{proof}
	We apply \cref{thm:existential_reg} to $H$ as a subgraph of $G$, with accuracy $\gamma/2$, to get a set of $p \leq 4/\gamma^2$ triples $(c_i,S_i,T_i)$ such that
	\[
		\abs*{ |E_H(S,T)| - \sum_{i=1}^p c_i \cdot |E_G(S_i\cap S, T_i\cap T)| } ~\leq~ \frac{\gamma}{2} \cdot |E_G(L,R)| ~=~ \frac{\gamma}{2} \cdot nd
	\]
	For each $i \in [p]$, we know by the expander mixing lemma that
	\[
		\abs*{|E_G(S_i\cap S, T_i\cap T)| - \frac{d}{n}\cdot |S_i\cap S| \cdot |T_i\cap T|} ~\leq~ \lambda nd
	\]
	Substituting this approximation back into the regularity lemma, we get
	\begin{align*}
		\abs*{ |E_H(S,T)| - \sum_{i=1}^p c_i \cdot |E_G(S_i\cap S, T_i\cap T)| } & ~\leq~ \frac{\gamma}{2} \cdot |E_G(L,R)| ~=~ \frac{\gamma}{2} \cdot nd \\
		\implies ~~\abs*{ |E_H(S,T)| - \frac{d}{n} \sum_{i=1}^p c_i \cdot |S_i\cap S| \cdot |T_i\cap T)| } & ~\leq~ \parens*{\frac{\gamma}{2}+ \lambda \cdot \sum_{i=1}^p |c_i|} \cdot nd \\
		& \leq \parens*{\frac{\gamma}{2}+ \lambda \cdot \frac{2}{\gamma}} nd ~\leq~ \gamma \cdot nd \qquad \qquad  \insquare{\text{for}~\lambda \leq \gamma^2/4} \qedhere
	\end{align*}
\end{proof}

\subsection{Efficient regularity lemma}
To make use of the regularity lemma algorithmically, we use a simple modification of \cref{algo:existential_regularity} to make it efficient. Note that the only inefficient step there is to check whether there exist $S_j \sub L, T_j \sub R$ such that $\abs{\one_{S_j}^{\trans} \mat_{j-1} \one_{T_j}} > \gamma \cdot |E(G)|$, and if so, to find these sets. This can be done if we can find sets $S_j, T_j$ maximizing $\abs{\one_{S_j}^{\trans} \mat_{j-1} \one_{T_j}}$, which is precisely the cut norm problem for the matrix $\mat_{j-1}$. Unfortunately, the cut-norm problem is NP-Hard to solve in general.

However, we can use the following result of Alon and Naor \cite{AN04} for a constant factor approximation of the cut norm. To be precise, we use the algorithm of Alon and Naor to find sets $S_j, T_j$ that satisfy $\abs{\one_{S_j}^{\trans} \mat_{j-1} \one_{T_j}} \geq \alpha_{AN}\cdot \max_{S\sub L, T\sub R} \abs{\one_S^{\trans} \mat_{j-1}\one_T}$. The Alon-Naor algorithm is based on a semidefinite programming (SDP) relaxation of the cut-norm problem, which they can then round to obtain a constant factor approximation. This algorithm is deterministic, and the runtime is dominated by the time needed to solve the SDP.
Further, it was shown by Jeronimo \etal \cite{JST21}, that using the ideas of Arora and Kale \cite{AK07} and Lee and Padmanabhan \cite{LeeP20}, the Alon-Naor algorithm can be made to run in randomized near-linear time.

\begin{theorem}[\cite{AN04, JST21}]\label{thm:cut-norm}
	Given any matrix $H \in \R^{L\times R}$, there is an algorithm that finds sets $S^*\sub L$ and $T^*\sub R$ such that
	\[
		\abs*{\one_{S^*}^{\trans} \mat \one_{T^*}} = \abs*{\sum_{\li \in S^*, \ri\in T^*} \mat_{\li\ri}} ~\geq~ \alpha_{AN}\cdot \max_{S\sub L, T\sub R} \abs*{\one_S^{\trans} \mat \one_T}
	\]
	Here $\alpha_{AN} \geq 0.03$ is an absolute constant. The algorithm can be made to run in deterministic $O(\nnz{H}^{3.5})$ time and randomized $\Ot{\nnz{H}}$ time, where $\nnz{H}$ denotes the number of non-zero entries of $H$.
\end{theorem}

\begin{theorem}[Efficient Regularity Lemma]\label{thm:generic_regularity}
	Let $G = (L,R,E)$ be a bipartite graph and let $H$ be a subgraph of $G$. Then for any $\gamma > 0$, there exists an algorithm that outputs a collection of $p\leq 1/(\alpha_{AN} \cdot \gamma )^2$ triples $(c_i,S_i,T_i) \in \R \times 2^L \times 2^R$ such that for any $S\sub L$ and $T\sub R$,
	\[
		\abs*{ |E_H(S,T)| - \sum_{i=1}^p c_i \cdot |E_G(S_i\cap S, T_i\cap T)| } \leq \gamma \cdot |E_G(L,R)|
	\]
	As before, the collections of reals $\{c_i\}_{i\in [p]}$ satisfy $\sum_i |c_i| \leq 1/(\alpha_{AN} \cdot \gamma )$. The algorithm can be made to run in randomized $\Ot{|E(G)|}$ time, or deterministic $\Oh{|E(G)|^{3.5}}$ time.
\end{theorem}

\begin{proof}
The following algorithm is a modification of \cref{algo:existential_regularity}.

\begin{algorithm}{Efficient Regularity Decomposition}{$\gamma > 0$, subgraph $H$ of $G = (L,R,E)$}{$M = \{(c_i, S_i, T_i)\}_{i\in [p]}$, with $ c_i \in\R, S_i \sub L, T_i \sub R$ for every $i\in [p]$} \label{algo:efficient_regularity}
     \begin{itemize}
      \item Let $\mat_0 = A_H$, $j=1$ and $\decom = \emptyset$
      \item Repeat:
      \begin{itemize}
      		\item Find $S_j \sub L,T_j \sub R$ such that $\abs*{\one_{S_j}^{\trans} \mat_{j-1} \one_{T_j}} \geq \alpha_{AN}\cdot \max_{S\sub L, T\sub R} \abs*{\one_S^{\trans} \mat_{j-1} \one_T}$
			\item If $\abs*{\one_{S_j}^{\trans} \mat_{j-1} \one_{T_j}} \leq \alpha_{AN}\cdot \gamma \cdot |E(G)|$, then exit loop.
			\item Let $c_j = \alpha_{AN} \cdot \gamma \cdot \sgn\parens*{\one_{S_j}^{\trans}\mat_{j-1}\one_{T_j}}$
			\item $\mat_j \gets \mat_{j-1} - c_j \cdot \mask{\one_{S_j} \one_{T_j}^{\trans}}$
            \item $\decom \gets \decom \cup \{(c_j,S_j,T_j)\}$
            \item $j = j + 1$
          \end{itemize}
      \item Let $p = j-1$
      \item Return $M$
     \end{itemize}
\end{algorithm}

As in the proof of \cref{thm:existential_reg}, we first need to show that when the algorithm stops, the set of triples $M = \{(c_i,S_i,T_i)\}_{i\in [p]}$ indeed satisfies the condition states in the theorem statement. The second claim we would need to argue is that the algorithm stops in at most $1/(\alpha_{AN}\cdot \gamma)^2$ many steps.

For the first claim, note that when the algorithm stops, we have
\[
\abs*{\one_{S_{p+1}}^{\trans} \mat_{j-1} \one_{T_{p+1}}} ~\leq~ \alpha_{AN}\cdot \gamma \cdot |E(G)|
\quad \implies \quad
\max_{S,T} \abs{\one_S^{\trans} \mat_p \one_{T}} \leq \gamma \cdot |E(G)| \mcom
\]
using the approximation guarantee for the Alon-Naor algorithm.
As in proof of \cref{thm:existential_reg}, we also have
\[
	\one_S^{\trans} \mat_p \one_T ~=~ E_H(S,T) - \sum_{i=1}^p c_i \cdot E_G(S_i \cap S, T_i\cap T) \mcom
\]
which shows that when the algorithm terminates, it satisfies the regularity condition. To bound the number of steps, we again compute
\begin{align*}
\norm{\mat_{j-1}}_{\muG}^2 - \norm{\mat_{j}}_{\muG}^2 
&~=~ 2c_j \cdot \ip{\mat_{j-1}}{\mask{\one_{S_j} \one_{T_j}^{\trans}}}_{\muG} -~ c_j^2 \cdot \norm*{\mask{\one_{S_j} \one_{T_j}^{\trans}}}_{\muG}^2 \\
&~\geq~ 2c_j \cdot \frac{\one_{S_j}^{\trans} \mat_{j-1} \one_{T_j}}{\abs{E(G)}} -~ c_j^2 ~\geq~ 2(\alpha_{AN} \cdot \gamma)^2 - (\alpha_{AN} \cdot \gamma)^2  
~=~ (\alpha_{AN} \cdot \gamma)^2 \mper
\end{align*}
For the runtime claim, note that it suffices to use the algorithm of \cref{thm:cut-norm} at most $p$ times, and throughout the algorithm, we maintain that the matrix $\mat_j$ is only supported on entries of $E(G)$, so that $\nnz{\mat_j} \leq nd$.
\end{proof}

Finally, we use the expander mixing lemma to get that when $G=(L,R,E)$ is a $\lambda$--spectral expander, the number of edges in $G$ between any two sets in $L$ and $R$ respectively can be approximated well just using the set sizes.

\begin{corollary}\label{cor:expander_regularity}
	Let $G$ be a bipartite $(n,d,\lambda)$--expander, and let $H$ be a subgraph of $G$. Then for any $\gamma > 0$, if it holds that $\lambda \leq \alpha_{AN}\cdot \gamma^2/4$, then there exists an algorithm that takes as input the graph $H$, and outputs a collection of $p\leq 4/(\alpha_{AN}^2\gamma^2)$ triples $(c_i,S_i,T_i) \in \R \times 2^L \times 2^R$ such that for any $S \sub L$ and $T\sub R$,
	\[
		\abs*{\abs{E_H(S,T)} - \frac{d}{n} \sum_{i-1}^p c_i \cdot \abs{S_i\cap S} \cdot \abs{T_i\cap T}} \leq \gamma \cdot nd
	\]
	The algorithm can be made to run in randomized $\Ot{|E(G)|}$ time, or deterministic $\Oh{|E(G)|^{3.5}}$ time. The collection of reals $\{c_i\}_{i\in [p]}$ satisfy $\sum_i |c_i| \leq 2/(\alpha_{AN}\cdot \gamma)$.
\end{corollary}
\begin{proof}
	We use the algorithm of \cref{thm:generic_regularity} on $H$ as a subgraph of $G$ with accuracy $\gamma/2$ to algorithmically obtain a set of triples $\decom = \{(c_i,S_i,T_i)\}_{i\in [p]}$ such that
	\begin{align}\label{eq:subgraph_regularity}
		\abs*{ |E_H(S,T)| - \sum_{i=1}^p c_i \cdot |E_G(S_i\cap S, T_i\cap T)| } ~\leq~ \frac{\gamma}{2} \cdot |E_G(L,R)| ~=~ \frac{\gamma}{2} \cdot nd
	\end{align}
	As guaranteed by \cref{thm:generic_regularity}, we can also assume $\sum_i |c_i| \leq 2/(\alpha_{AN}\cdot \gamma)$. For each $i\in [p]$, by the expander mixing lemma, we have
	\[
		\abs*{|E_G(S_i\cap S,T_i\cap T)| - \frac{d}{n} \cdot |S_i\cap S| \cdot |T_i\cap T|} ~\leq~ \lambda \cdot nd
	\]
	We substitute this approximation back into \cref{eq:subgraph_regularity}, and keep track of the error as
	\begin{align*}
		\abs*{ |E_H(S,T)| - \frac{d}{n} \sum_{i=1}^p c_i \cdot |S_i\cap S| \cdot |T_i\cap T)| } & ~\leq~ \inparen{\frac{\gamma}{2} + \sum_{i=1}^p \abs{c_i} \cdot \lambda} \cdot nd  \\
		&~\leq~ \parens*{\frac{\gamma}{2}+\frac{2\lambda}{\alpha_{AN}\cdot \gamma}}nd ~\leq~ \gamma \cdot nd \qquad \qquad \insquare{\text{using}~\lambda \leq \frac{\alpha_{AN}\cdot \gamma^2}{4}} \\
	\end{align*}
	The algorithm runtime simply reflects the analogous runtime in \cref{thm:generic_regularity}.
\end{proof}
In conclusion, given any subgraph $H$ of a $\lambda$--expander $G$, we can find a regularity decomposition for it with error $\gamma$ as long as $\gamma \geq C\sqrt{\lambda}$ for some absolute constant $C$ ($C=12$ suffices).

\subsection{Function families given by regularity lemma}\label{sec:regularity-factor}
The regularity lemma of \cref{cor:expander_regularity} can also be viewed as yielding an efficiently
constructible family of functions, given by the indicator function of the sets appearing in the lemma. 
Below we formally define the properties need from the family, and analyze the parameters given by
\cref{cor:expander_regularity}. 
Note that while the regularity lemma applied to a bipartite graph on vertex sets $(L,R)$ yields
subsets of both $L$ and $R$, it will suffice for our applications to only consider subsets of
$L$, and thus we define the family below accordingly.

\begin{definition}\label{regular-family}
Let $G = (L,R,E)$ be a $d$-regular bipartite graph with $\abs{L} = \abs{R} = n$, and let $H$ be a subgraph of $G$.
We say that a function family $\family$ defined on $L$ is $(\eta,\gamma)$-regular with
respect to $H$ if for any $S, S' \subseteq L$ and $T \subseteq R$, it holds that
\[
\norm{\indicator{S} - \indicator{S'}}_{\family} ~\leq~ \eta
\qquad \implies \qquad
\abs{E_H(S,T) - E_H(S',T)} ~\leq~ \gamma \cdot n d \mper
\]
\end{definition}
We now argue that the regularity lemma yields efficiently constructible $(\eta,\gamma)$-regular
families.
\begin{lemma}\label{lem:regularity-family}
Let $G = (L,R,E)$ be an $(n,d,\lambda)$-expander, and let $H$ be a subgraph of $G$. Then, for every
$\gamma \geq 50\sqrt{\lambda}$, there exists $\eta \geq C_0 \cdot \gamma^2$, and an
$(\eta,\gamma)$-regular family constructible in randomized $\tilde{O}(\abs{E(G)})$ or deterministic
$O(\abs{E(G)}^{3.5})$ time.
\end{lemma}
\begin{proof}
Given $\gamma \geq 50\sqrt{\lambda}$, we apply \cref{cor:expander_regularity} with parameter $\gamma/4$ to obtain a collection of triples $\inbraces{(c_i, S_i, T_i)}_{i \in [p]}$, for $p = O(1/\gamma^2)$, satisfying for all $S \subseteq L$ and $T \subseteq R$
\[
\abs*{ |E_H(S,T)| - \sum_{i=1}^p c_i \cdot \frac{d}{n} \cdot \abs{S \cap S_i}\abs{T \cap T_i} } 
~=~ \abs*{ |E_H(S,T)| - \sum_{i=1}^p c_i \cdot nd \cdot \angles{\indicator{S}, \indicator{S_i}} \cdot \angles{\indicator{T}, \indicator{T_i}} }
~\leq~ \frac{\gamma}{4} \cdot nd \mper
\]
By triangle inequality, we get that for any $S, S' \subseteq L$ and $T \subseteq R$
\[
\frac{1}{nd} \cdot \abs*{\abs{E_H(S,T)} - \abs{E_H(S',T)}} ~\leq~ \sum_{i \in [p]} \abs{c_i} \cdot \abs{\angles{\indicator{S}-\indicator{S'},\indicator{S_i}}} \cdot \angles{\indicator{T}, \indicator{T_i}} + \frac{\gamma}{2} \mper
\]
Let $\family = \braces{\indicator{S_i}}_{i \in [p]}$ be the function family corresponding to subsets of $L$ given by the regularity lemma. By definition of the distance $\norm{\cdot}_{\family}$, we get
\[
\angles{\indicator{S} - \indicator{S'}, \indicator{S_i}} ~\leq~ \norm{\indicator{S} - \indicator{S'}}_{\family} ~\leq~ \eta \mper
\]
This can be used to bound the difference of number of edges as
\[
\frac{1}{nd} \cdot \abs{\abs{E_H(S,T)} - \abs{E_H(S',T)}} ~\leq~ \sum_{i \in [p]}\abs{c_i} \cdot \eta ~+~ \frac{\gamma}{2} \mper
\]
From \cref{cor:expander_regularity}, we have that $p \leq (8/\alpha_{AN} \cdot \gamma)^2$ and $\sum_i \abs{c_i} \leq 8/(\alpha_{AN} \cdot \gamma)$. Thus, for $\eta = (\alpha_{AN} \cdot \gamma^2)/16$, we have that $\sum_i \abs{c_i}\cdot \eta \leq \gamma/2$, which completes the proof.
\end{proof}
We will also need to construct families which are simultaneously regular for multiple subgraphs of $G$. The following claim shows that this can be done by simply obtaining a regular family for each subgraph, and then taking a union.
\begin{claim}[Refinement]\label{clm:refinement}
Let $\family$ be an $(\eta,\gamma)$-regular family for $H$, and let $\family' \supseteq \family$ be a larger family. 
Then $\family'$ is also $(\eta,\gamma)$-regular for $H$.
\end{claim}
\begin{proof}
It suffices to show that for any two sets $S,S' \subseteq L$, we have $\norm{\indicator{S}-\indicator{S'}}_{\family} \leq \norm{\indicator{S}-\indicator{S'}}_{\family'}$. 
Since $\family' \supseteq \family$, we have that
\[
\norm{\indicator{S} - \indicator{S'}}_{\family} ~=~ \sup_{f \in \family} \abs{\angles{\indicator{S} - \indicator{S'}, f}} ~\leq~ \sup_{f \in \family'} \abs{\angles{\indicator{S} - \indicator{S'}, f}} ~=~ \norm{\indicator{S} - \indicator{S'}}_{\family'} \qedhere
\]
%
\end{proof}
\begin{remark}\label{remark:trivial_factor}
By the expander mixing lemma, we have that 
\[
\abs*{\abs{E_G(S,T)} - \frac{d}{n} \cdot \abs{S}\abs{T}} ~=~ \abs*{\abs{E_G(S,T)} - nd \cdot \angles{\indicator{S}, \indicator{L}} \cdot \angles{\indicator{T}, \indicator{R}}} ~\leq~ \lambda \cdot nd \mcom
\]
which shows that the trivial family $\family_0 = \{\indicator{L}\}$ is $(\eta,\gamma)$-regular for $G$, for any $\gamma \geq 4\lambda$ and $\eta = \gamma/2$. We always assume any family $\family$ we construct also contains the function $\indicator{L}$, so that by \cref{clm:refinement}, $\family$ will also be regular for $G$.
\end{remark}



%% file: ael.tex
\section{Algorithms for AEL Codes}\label{sec:ael}
%
%
Alon, Bruck, Naor, Naor and Roth \cite{ABNNR92} introduced a graph-based distance amplification scheme which was generalized by Alon, Edmonds and Luby \cite{AEL95}, and used by Guruswami and Indyk \cite{GI05} to design linear-time unique decodable near-MDS codes.

The scheme is a three-step process involving an outer code  ($\calC_\out$), an inner code ($\calC_\inn$), and a bipartite expander $G$: (i) concatenate the outer code $\cC_\out$ with inner code $\cC_\inn$ (ii) shuffle the symbols of concatenated code via edges on a $d$-regular bipartite expander graph $G$, and (iii) collect $d$-symbols on the right vertices and fold them back to produce the final code, $\AELC$. We now formally define this procedure.

Fix an $(n,d,\lambda)$-expander $G = (L,R,E)$. Let $\calC_{\out}$ be an $[n,r_{\out},\delta_{\out}]_{\Sigma_\out}$ code and let $\calC_{\inn}$ be a $[d,r_{\inn},\delta_{\inn}]_{\Sigma_\inn}$ code with $\abs{\Sigma_\out} = |\calC_\inn|$, and let $\Enc_{\inn}: \Sigma_\out\rightarrow \calC_\inn\subseteq \Sigma_\inn^d$ be the encoding map for $\calC_{\inn}$, along with its inverse $\Enc_{\inn}^{-1} : \calC_{\inn} \to \Sigma_{\out}$.

The code $\AELC$ is defined over the alphabet $\Sigma \defeq \Sigma_{\inn}^d$ and coordinate set $R$, so that $\AELC \sub \Sigma^R$. A codeword $f$ of $\AELC$ technically belongs to the space $\Sigma^R = (\Sigma_\inn^d)^R$ but by our choice of parameters, $(\Sigma_\inn^d)^R$ is in bijection with the set $\Sigma_\inn^E$ and one can just \emph{fold} or \emph{unfold} the symbols to move from one space to the other. 


%
%
%
%
%

\begin{definition}[AEL Codes]

Given inner code $\cC_\inn$ an outer code $\cC_\out$, a map $\varphi$, and a graph $G$ as above, the AEL code is defined as,
\[
\AELC ~=~ \braces[\big]{f \in \parens{\Sigma_\inn^d}^R \cong \Sigma_{\inn}^E ~\mid~ \forall \ell \in L, f_\ell \in \cC_\inn,\;\text{and } \braces[\big]{\Enc_{\inn}^{-1}(f_\ell)}_{\ell \in L} \in \cC_\out }.
\]
\end{definition}

One can also think of the code procedurally as, starting with $f_0 \in \cC_\out$, and defining $f \in \Sigma_\inn^E$ such that $f_\ell := \Enc_{\inn}((f_0)_\ell)$. 
%
%
This $f$ can then be folded on the right vertices to obtain an AEL codeword. Note that each codeword of $\cC_{\out}$ gives rise to a unique codeword of $\AELC$.


\begin{figure}[htb]
\begin{center}	
\begin{tikzpicture}[scale = 0.7]
\begin{scope}
\draw[] (0,0) ellipse (1.5cm and 3cm);
\node[below] at (0,-3.1) {$L$};
\draw[] (6,0) ellipse (1.5cm and 3cm);
\node[below] at (6,-3.1) {$R$};


\node[fill,circle,red] (a1) at (0,2) {};
\node[fill,circle,red] (a2) at (0,1) {};
\node[left,darkred] at (-0.2,1) {$\cC_\inn \ni \parens[\big]{f(e_1),f(e_1), f(e_3)} = f_{\ell}  $};
\node[fill,circle,red] (a3) at (0,0) {};
\node[fill,circle,red] (a4) at (0,-1) {};
\node[fill,circle,red] (a5) at (0,-2) {};
\node[fill,circle,blue] (b1) at (6,2) {};
\node[fill,circle,blue] (b2) at (6,1) {};
\node[fill,circle,blue] (b3) at (6,0) {};
\node[below,darkblue] at (10,0.45) {{$f_r = \parens[\Big]{f(e_2),f(e_4), f(e_5)} \in \Sigma_\inn^d$}};
\node[fill,circle,blue] (b4) at (6,-1) {};
\node[fill,circle,blue] (b5) at (6,-2) {};

\draw[](a2)--node[pos=0.45,sloped, above] {\small{$f(e_1)$}}(b1);
\draw[](a2)--node[pos=0.45,sloped, above] {\small $f(e_2)$}(b3);
\draw[](a2)--node[pos=0.45,sloped, above] {\small $f(e_3)$}(b5);
\draw[] (a4)--node[pos=0.17,sloped, above] {\small{$f(e_4)$}}(b3);
\draw[] (a5)--node[pos=0.17,sloped, above] {\small{$f(e_5)$}}(b3);

\end{scope}
\end{tikzpicture}
\vspace{-10 pt}
\end{center}
\caption{Illustration of the AEL procedure}
\label{fig:ael-main}
\end{figure}
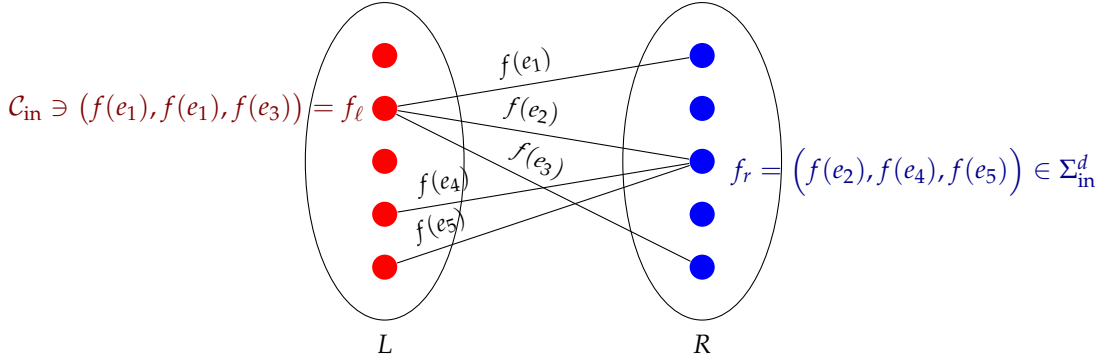

Let us briefly recall the details of the AEL amplification from \cref{sec:intro}. It is easy to see that if the rate of the outer code $\cC_{\out}$ is $\rho_{\out}$ and the rate of the inner code $\cC_{\inn}$ is $\rho_{\inn}$, then the rate of $\AELC$ is $\rho_{\out}\cdot \rho_{\inn}$. The following theorem gives a bound on the distance of $\AELC$ in terms of the distances $\delta_{\inn}$ and $\delta_{\out}$ of $\cC_{\inn}$ and $\cC_{\out}$ respectively.
%
%
%
\begin{theorem}[\cite{AEL95, GI05}]\label{thm:ael_distance}
	The distance of AEL code $\AELC$ is at least $\delta_{\inn} - \frac{\lambda}{\delta_{\out}}$.
\end{theorem}
\begin{proof}
Let $h_1, h_2 \in AELC$ be any two distinct codewords, and let $L' = \inbraces{\ell \in L ~|~ h_{1\li} \neq h_{2\li}}$. By the distance of $\cC_{\out}$, we must have $\abs{L'} \geq \Delta_L(h_1,h_2) \geq \delta_{\out} \cdot n$. Also, each $\ell \in L'$ has at least $\delta_{\inn} \cdot d$ incident edges on which $h_{1e} \neq h_{2e}$, since $h_{1\ell}, h_{2\ell}$ are codewords in $\cC_{\inn}$. Let $R' \subseteq R$ be the set of the right endpoints of these edges. We also have $\abs{R'} \leq \dis_R(h_1,h_2) \cdot n$ since $h_{1r} \neq h_{2r}$  on each $r \in R'$. Expander mixing lemma then gives
\[
\delta_{\inn} \cdot d \cdot \abs{L'} ~\leq~ E_G(L',R') 
~\leq~ \frac{d}{n} \cdot \abs{L'} \cdot \abs{R'}  + \lambda \cdot dn 
~\leq~  d \cdot \abs{L'} \cdot \dis_R(h_1,h_2)  + \lambda \cdot dn \mper
\]
Rearranging and using $\abs{L'} \geq \delta_{\out} \cdot n$ then proves the claim.
\end{proof}

Typically, the AEL construction is used with a high rate code, so that $\rho_{\out}$ is close to 1, and $\delta_{\out}$ is positive but arbitrarily small. This means that the rate of $\AELC$ is close to $\rho_{\inn}$, and by choosing $\lambda$ to be small enough, the distance of $\AELC$ can be made arbitrarily close to $\delta_{\inn}$. Thus, we get an infinite explicit family of codes with the parameters inherited from a small constant sized code. See \cref{sec:main_ael_decoding} for an example of a formal instantiation to obtain LDPC codes.

For the rest of this section, the alphabet of the AEL code, $\Sigma_\inn^d$ will be denoted simply as $\Sigma$.
\subsection{List decoding AEL codes up to capacity}\label{sec:ael_decoding}
Given $g\in \Sigma^R$ and an error radius $\beta \in [0,1]$, the list decoding problem is to output the list 
\[
\calL(g,\beta) ~=~ \inbraces{ h\in \AELC \sub \Sigma^R ~\vert~ \Delta(g,h) \leq \beta } \mper
\]
We will show that for any $\eps>0$, assuming the inner constant-sized code is list decodable up to radius $\beta+\eps$ with a list size independent of its blocklength $d$, the global code $\AELC$ is list decodable up to radius $\beta$, assuming $\lambda$ is small enough in terms of $\eps$. We will also need to assume $\beta + \eps \leq \delta_{\inn}$, so that we are only decoding under the distance $\delta_{\inn}$ of the inner code. We can then use either random linear codes or folded Reed-Solomon codes as inner codes, both of which are decodable up to list decoding capacity with constant (independent of $d$) sized lists, to get that $\AELC$ can also be decoded up to the list decoding capacity. We defer the formal details of these parameters to \cref{sec:main_ael_decoding}.

As remarked before, it is possible to think of $g$ as an element of $\Sigma_{\inn}^E$, simply by unfolding the symbols in $g_{\ri}$ for each $\ri \in R$. For example, this allows us to write $g_{\li}$ to denote the restriction of $g$ to the edges incident on some $\li\in L$. We will assume this freedom for $g$ (and other elements of $\Sigma^R$, such as codewords of $\AELC$) implicitly without always pointing out the transition from $\Sigma^R$ to $\Sigma_{\inn}^E$.

For a given $g$, we will obtain the list of codewords within a given error radius, by first obtaining local lists $\calL_{\ell}$ for each $\ell \in L$, and then using regularity to consistently select codewords from the local lists. 
For the argument below, we consider a \emph{fixed} global codeword $h$ such that $\Delta_R(g,h) \leq \beta$ and show that there exists some element in the set we enumerate, which will allow us to recover $h$. Since $h$ is arbitrary, this will prove that an enumeration based algorithm will recover the entire list.

\paragraph{Local lists.} We define the local lists at vertices $\ell \in L$ as those obtained by decoding $g_{\li}$ up to radius $\beta+\eps$,
\[
\calL_{\li} \defeq \calL(g_{\li},\beta+\eps) ~=~ \inbraces{f_{\ell} \in \cC_{\inn} ~|~ \Delta(f_{\ell}, g_{\ell}) \leq \beta +\eps} \mper
\]
The expander mixing lemma implies that for each global codeword within a list of radius $\beta$, its local projection must be in most local lists of radius $\beta$.
\begin{claim}\label{clm:local-membership}
Let $g \in \Sigma^R$ and $h \in \AELC$ with $\dis_R(g,h) \leq \beta$ and let $\lambda \leq \gamma \cdot \eps$. Then, $\Pr{\ell \in L}{h_{\li} \not\in \listl} \leq \gamma$. 
\end{claim}
\begin{proof}
Note that $h_{\li} \not\in \listl$ only if $\Delta(g_{\li},h_{\li}) > \beta+\eps$. Let $L' = \inbraces{\ell \in L ~|~ \Delta(g_{\ell}, h_{\ell})  > \beta + \eps}$. Then each $\ell \in L'$ has at least $(\beta + \eps) \cdot d$ incident edges on which $g_e \neq h_e$. Let $R' \subseteq R$ be the set of the right endpoints of these edges. Then, we have 
\[
(\beta + \eps) \cdot d \cdot \abs{L'} ~\leq~ E_G(L',R') 
~\leq~ \frac{d}{n} \cdot \abs{L'} \cdot \abs{R'}  + \lambda \cdot dn \mper
\]
Since $g$ and $h$ must differ on $R'$, we have $\abs{R'} \leq \dis_R(g,h) \cdot n \leq \beta \cdot n$, which proves the claim.
\end{proof} 
While we now know that $h_{\li}$ is present in most local lists $\listl$, we do not know \emph{which} element of $\listl$ corresponds to $h_{\li}$. We will next show how to stitch together an assignment almost consistent with $h$ from these local lists $\{\listl\}_{\li\in L}$.
\subsubsection{Picking from local lists via regularity}
Given an upper bound $K$ on the size of the local lists at radius $\beta+\eps$ (for the inner code $\cC_{\inn}$),  we will think of the local lists as sets of size \emph{exactly} $K$, by including arbitrary codewords (distinct from list elements) if necessary. We will also impose an arbitrary ordering for each $\calL_{\ell}$ and use $\calL_{\ell}[i]$ to refer to the $i$-th element. 

\paragraph{Edge-densities in agreement graphs.}  Recall that we have fixed a global codeword $h$ such that $\Delta(g,h) \leq \beta$. 
We consider the (unknown) sets $A_1, \ldots, A_K$ defined by the the position of $h_{\ell}$ in the local lists
\[
A_i ~=~ \inbraces{\ell \in L ~|~ \calL_{\ell}[i] = h_{\ell}} \mper
\]
Note that the sets $\{A_i\}_{i \in [K]}$ are disjoint, and \cref{clm:local-membership} implies that $\sum_{i}\abs{A_i} \geq (1-\gamma) \cdot \abs{L}$.
Also note that if we can recover the sets $\{A_i\}_{i \in [K]}$ up to a small fraction of errors, then we correctly know $h_{\ell}$ for most $\ell \in L$, and can find $h$ by unique decoding the outer code $\cC_0$. 
We will recover the sets $A_i$ by understanding their interaction with the (also unknown) set $B ~=~ \inbraces{r \in R ~|~ g_r = h_r}$. Note that $\abs{B} \geq (1-\beta) \cdot n$.

Consider the graph $H_i = (L,R, E(H_i))$  defined by the agreement of the given $g \in \Sigma^R$ and the $i$-th elements in the local lists
\[
E(H_i) ~=~ \{e = (\ell, r) \in E(G) ~|~ (\calL_{\ell}[i])_e = g_e \} \mper
\]
Note that the graphs $H_i$ can be easily computed using the given $g$.
Also, we must have $E_{H_i}(A_i, B) = E_G(A_i,B)$ since for an edge $e = (\ell, r) \in E_G(A_i,B)$, we have
$(\calL_{\ell}[i])_e = h_e = g_e$.
That is, while some edges of $G$ are deleted to obtain $H_i$, none of the edges between $A_i$ and $B$ were deleted. Thus the induced subgraph of $H_i$ on $A_i \cup B$ looks like a very dense subgraph relative to $G$.

We will recover the sets $A_1, \ldots, A_K$ by using the regularity lemma to characterize dense subgraphs (relative to $G$) of the agreement graphs $H_i$. 
\begin{figure}[h]
\centering
\includegraphics[width=0.6\textwidth]{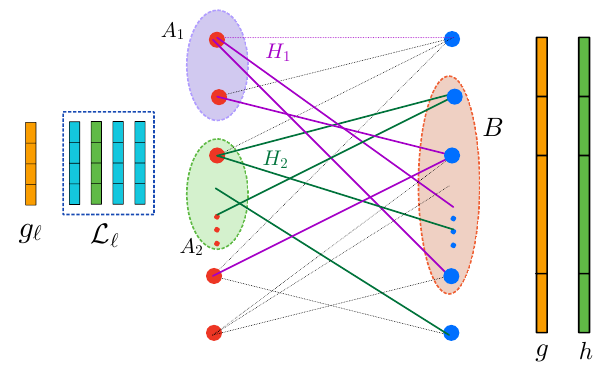}
\vspace{-10 pt}
\caption{Agreement graphs for list decoding AEL codes}
\vspace{-10 pt}
\label{fig:ael-graphs-non-intro}
\end{figure}
\vspace{5pt}
The following key technical lemma shows that a sufficiently regular family $\factor$, which is \emph{simultaneously} $(\eta,\gamma)$-regular with respect to $H_1, \ldots, H_K$, essentially identifies the sets $A_1, \ldots, A_K$.
We can obtain such a family by applying the regularity lemma for each of the graphs $H_i$, and then taking a union.
\begin{lemma}[Rigidity]\label{lem:ael-rigidity}
Let the decoding radius $\beta$ satisfy $\beta \leq  \delta_{\inn} - \eps$ for $\delta_{\inn}=\delta(\cC_{\inn})$, and let the family $\family$ be simultaneously $(\eta,\gamma)$-regular for the graphs $H_1, \ldots, H_K$. Let $S_1, \ldots, S_K \subseteq L$ be any collection of disjoint sets satisfying $\max_{i \in [K]} ~\norm{\indicator{A_i} - \indicator{S_i}}_{\family} ~\leq~ \eta$. Further, assume\footnote{Note that $\lambda$ being at most $\gamma$, and in fact much smaller, is anyway needed to obtain $(\eta,\gamma)$-regular families.} $\lambda \leq \gamma$.
Then, $\sum_{i} \abs{S_i \setminus A_i} \leq (2K\gamma/\eps) \cdot n$.
\end{lemma}
\begin{proof}
By the $\eta$-closeness of the sets with respect to $\family$, we have
\[
\forall i \in [K] \qquad \abs{E_{H_i}(S_i,B)} ~\geq~ \abs{E_{H_i}(A_i,B)} - \gamma \cdot dn ~=~ \abs{E_{G}(A_i,B)} - \gamma \cdot dn \mper
\]
Summing the above and using the bound $\sum_i\abs{A_i} \geq (1-\gamma) \cdot n$ from \cref{clm:local-membership}, we get
\[
\sum_{i \in [K]} \abs{E_{H_i}(S_i,B)} 
~~\geq~~ \abs{E_{G}(\cup_i A_i,B)} - K \cdot \gamma \cdot dn 
~~\geq~~ d \cdot \abs{B} - (K+1) \cdot\gamma \cdot dn \mper
\]
Now, consider the sets $W_i = S_i \setminus A_i$. 
Note that for each $\ell \in W_i$, we have $(\calL_{\ell}[i])_e \neq h_e$ for at least $\delta_{\inn}$ fraction of the edges incident on $\li$ since $\calL_{\ell}[i]$ and $h_{\ell}$ must be distinct codewords in $\cC_{\inn}$. 
Also, such an edge cannot be in $E_{H_i}(S_i,B)$ since we have $(\calL_{\ell}[i])_e \neq h_e = g_e$. Using this, and the fact that $S_i \setminus W_i \subseteq A_i$, we get
\begin{align*}
\sum_{i \in [K]} \abs{E_{H_i}(S_i,B)} 
&~\leq~ \sum_{i \in [K]} \abs{E_{H_i}(S_i \setminus W_i,B)} + \sum_{i \in [K]} \abs{W_i} \cdot (1 - \delta) \cdot d \\
&~=~ \sum_{i \in [K]} \abs{E_{G}(S_i \setminus W_i,B)} + \sum_{i \in [K]} \abs{W_i} \cdot (1 - \delta_{\inn}) \cdot d \mper
\end{align*}
We can now use expander mixing lemma to bound the number of edges in $G$.
\[
\sum_{i \in [K]} \abs{E_{G}(S_i \setminus W_i,B)} 
~\leq~ \inparen{\sum_i \abs{S_i \setminus W_i}} \cdot \abs{B} \cdot \frac{d}{n} + \lambda \cdot dn 
~\leq~ \inparen{n - \sum_i\abs{W_i}} \cdot \abs{B} \cdot \frac{d}{n} + \lambda \cdot dn \mper
\]
Combining the above inequalities, and using $w$ to denote $\sum_i \abs{W_i}/n$, we get
\begin{align*}
d\abs{B} - (K+1)\cdot\gamma \cdot dn ~\leq~ \sum_{i \in [K]} \abs{E_{H_i}(S_i,B)} 
~&\leq~
(1- w) \cdot d \abs{B}  + \lambda \cdot dn + w \cdot (1 - \delta_{\inn}) \cdot dn \\
\implies~
w \cdot (\abs{B} - (1 - \delta_{\inn})\cdot n) ~&\leq~ (\lambda + (K+1)\cdot\gamma) \cdot n \\
\implies~ w \cdot (\delta_{\inn} - \beta) \cdot n ~&\leq~ (\lambda + (K+1)\cdot\gamma) \cdot n \mper && \insquare{~|B| \geq (1-\beta)n~}
\end{align*}
Using $\lambda \leq \gamma$ and $\beta \leq \delta_{\inn} - \eps$ then proves the claim.
\end{proof}

The Rigidity lemma above says that to find a collection $S_1,\ldots,S_K$ that is close to $A_1,\ldots,A_K$, it suffices for their signatures to be close, with respect to a family $\family$ that is $(\eta,\gamma)$-regular for small enough $\gamma$. Since such a family is of constant size (independent of $n$), we can afford to guess all of these signatures. 
We next describe how to use the covering nets of \cref{clm:enumeration} to do this formally.

\subsubsection{The list decoding algorithm}
Let $\AELC$ be a code obtained via the AEL construction, using the outer code $\cC_{\out}$ and inner
code $\cC_{\inn}$, and using the graph $G=(L,R,E)$ which is an $(n,d,\lambda)$-expander for a sufficiently
small $\lambda$. 
%
Then the following algorithm finds the list of radius $\delta_{\inn} - \eps$ around any $g \in
\Sigma^R$, using only $K$ calls to the algorithmic regularity lemma
(\cref{cor:expander_regularity}) and a constant number of calls to the unique-decoder for $\cC_{\out}$, 
where the constant only depends on the parameters $\delta_{\dec}$, $\eps$, and $K$.

\medskip
\begin{new-algorithm}{List Decoding $\AELC$}\label{algo:ael-decoding} 
\begin{tabular}{r l}
\textsf{Input:} & $g \in\Sigma^R$ \\[3 pt]
\textsf{Output:} & List $\calL \subseteq \AELC$\\[3 pt] 
\textsf{Parameters:} & List decoding radius  $\beta \leq \delta(\AELC) - 2\eps$, list size bound $K$
                       for  $\cC_{\inn}$ at radius $\beta + \eps$, \\
& unique decoding radius $\delta_{\dec}$ for $\cC_{\out}$
\end{tabular}
\begin{itemize}
\medskip
\item For each $\ell \in L$, compute the local list $\calL_{\ell} = \inbraces{f_{\ell} \in \cC_{\inn} ~|~ \Delta(f_{\ell}, g_{\ell}) \leq \beta + \eps}$. Pad with additional elements of $\cC_{\inn}$ if needed, to view each $\calL_{\ell}$ is an ordered list of size $K$.
\item For each $i \in [K]$ define a graph $H_i$ with vertex sets $(L,R)$ and edges
\[
E(H_i) ~=~ \{e = (\ell, r) \in E(G) ~|~ (\calL_{\ell}[i])_e = g_e \} \mper
\]
For $\gamma = (\eps \cdot \delta_{\dec})/(4K)$ and sufficiently small $\eta$, compute a family $\family_i$ which is $(\eta,\gamma)$-regular for $H_i$. Let $\family$ be the union of $\family_1, \ldots, \family_K$.
\item For each $\sigma = (\sigma_1, \ldots, \sigma_K) \in \calN_{\eta/4}(\family, K)$
\begin{itemize}
\item Consider disjoint sets $S_1, \ldots, S_K \subseteq L$ given by \cref{clm:enumeration} with $\max_{i \in [K]} \norm{\sigma_i - \sigma(\indicator{S_i})}_{\infty} \leq \eta/2$.
\item Define $h' \in (\cC_{\inn})^L$ as
\[
h'_{\ell} ~=~
\begin{cases}
\calL_{\ell}[i] &~\text{if}~ \ell \in S_i \\
\text{arbitrary} &~\text{if}~ \ell \notin \cup_i S_i
\end{cases}
\]
\item Unique decode $h'$ to find $h_0 \in \cC_{\out}$ and corresponding $h \in \AELC$. 
\item If $\Delta(g,h) \leq \beta$, then $\calL \leftarrow \calL \cup \{h\}$.
\end{itemize}
\item Return $\calL$.
\end{itemize}
\end{new-algorithm}

\medskip
\begin{theorem}\label{thm:ael_decoding_technical}
%
%
Let $\tlocal{\delta}$ be the time needed to list decode the inner code $\calC_{\inn}$ up to radius $\delta$, $\treg{\gamma}$ be the time required to obtain an $(\eta,\gamma)$-regular family for any subgraph of
$G$ with $\eta = \Omega(\gamma^2)$ and let $T_{\dec}$ be the running time of the unique-decoder for the outer code that can decode $\delta_{\dec}$ fraction of errors.

Let $\eps > 0$ be arbitrary. Given any $g \in \Sigma^R$ and $\beta \leq \delta_{\inn} -\eps$, let $K$ be an upper bound on the list size of $\calC_{\inn}$ for decoding up to radius $\beta+\eps$ and let $\gamma = (\eps \cdot \delta_{\dec})/4K$. Assuming $\lambda \leq \gamma$, the \cref{algo:ael-decoding}
runs in time $O(n\cdot \tlocal{\beta+\eps} + K \cdot \treg{\gamma}) + T_{\dec} \cdot (1/\gamma)^{O(K^2/\gamma^2)}$ 
and returns the list $\calL = \inbraces{h \in \AELC ~|~ \dis_R(g,h) \leq
\beta}$. Further, the produced list is of size at most $(1/\gamma)^{{O(K^2/\gamma^2)}}$.

\end{theorem}
\begin{proof}
We note that the computation of the local lists, and the graphs $H_i$ can be done in time $O(n)$. 
By \cref{lem:regularity-family}, we can take $\eta = \Omega(\gamma^2)$ and obtain an $(\eta,\gamma)$-regular family $\family_i$ for each graph $H_i$ in time $\treg{\gamma}$, with $\abs{\family_i} = O(1/\gamma^2)$. Thus, for the common union $\family$, we have $\abs{\family}= O(K/\gamma^2)$
By \cref{clm:enumeration}, we have $\abs{\calN_{\eta/4}(\family,K)} = (1/\eta)^{O(K \cdot \abs{\family})} = (1/\gamma)^{O(K^2/\gamma^2)}$, and the enumeration can be done in time $T_{\dec} \cdot (1/\gamma)^{O(K^2/\gamma^2)}$.

Finally, it remains to prove that the algorithm correctly recovers the list. 
Note that for any $h$ satisfying $\dis(g,h) \leq \beta$, and corresponding sets $\inbraces{A_i}_{i \in [K]}$ indicating its position in local lists, the $h'$ found by our algorithm agrees with $h$ on all vertices $\ell \in \cup_i (S_i \cap A_i)$, and thus
\[
\dis_L(h,h') ~\leq~ 1 - \frac{1}{n} \cdot \sum_i \abs{S_i \cap A_i} 
~=~ 1 - \sum_i \angles{\indicator{S_i},\indicator{A_i}} 
~=~ 1 - \sum_i \ExpOp[\indicator{S_i} - \indicator{S_i \setminus A_i}] \mper
\] 
By \cref{clm:enumeration} there \emph{exists} $\sigma \in
\calN_{\eta/4}(\family,K)$ such that $\max_{i \in [K]}\norm{\sigma_i - \sigma(\indicator{A_i})}_{\infty} \leq \eta/4$. 
Moreover, the  disjoint sets $S_1, \ldots, S_K$ produced using $\sigma$ will satisfy  
$\max_{i \in [K]} \norm{\sigma_i-\sigma(\indicator{S_i})}_{\family} \leq \eta/4 + \abs{\factor(\family)}/n \leq \eta/2$.
Thus, we obtain by triangle inequality
\[
\max_{i \in [K]}\norm{\indicator{A_i} - \indicator{S_i}}_{\family} 
~=~ \max_{i \in [K]}\norm{\sigma(\indicator{A_i}) - \sigma(\indicator{S_i})}_{\infty} 
~\leq~ \eta \mper
\]
Also, by \cref{lem:ael-rigidity}, for any such collection $S_1, \ldots, S_K$, we have $\sum_i \ExpOp[\indicator{S_i \setminus A_i}] \leq (2K\gamma/\eps)$. Moreover, by \cref{clm:local-membership}, we have $\sum_i \abs{A_i} \geq (1-\gamma) \cdot n$ and $\norm{\indicator{A_i} - \indicator{S_i}}_{\family} \leq \eta$ implies
\[
\dis_L(h,h') 
~\leq~ 1 - \sum_i (\ExpOp[\indicator{A_i}] - \gamma) + \sum_{i}\ExpOp[\indicator{S_i \setminus A_i}] 
~\leq~ \gamma + K \gamma + 2K\gamma/\eps ~\leq~ 4K\gamma/\eps\mcom
\]
which is at most $\delta_{\dec}$ by our choice of $\gamma$. 
Thus, for every $h$ with $\dis_R(g,h) \leq \beta$, one of the choices in our enumeration finds an
$h'$ satisfying $\dis_L(h,h') \leq \delta_{\dec}$. Unique decoding this $h'$ then recovers $h$, and
thus the enumeration algorithm recovers all $h$ in the list.
\end{proof}

\subsubsection{Instantiating AEL codes for list decoding}\label{sec:main_ael_decoding}
We instantiate \cref{thm:ael_decoding_technical} with a brute force inner code decoder and the algorithm of \cref{lem:regularity-family} to find regular families. This gives a general reduction from list decoding AEL codes up to the distance of the inner code, to unique decoding the outer code $\calC_{\out}$, assuming the inner code can be list decoded up their distance.
\begin{theorem}[List Decoding of AEL Codes up to Capacity]\label{thm:ael_decoding_instantiation}
	Let $\AELC$ be the code obtained via the AEL construction applied to an outer code $\calC_{\out}$ over alphabet $\Sigma_{\out}$ and an inner code $\calC_{\inn}$ over alphabet $\Sigma_{\inn}$ using an $(n,d,\lambda)$-expander $G=(L,R,E)$. Assume that the outer code can be unique decoded from a fraction $\delta_{\dec}$ of errors in time $T_{\dec}$.
	
	Fix an arbitrary $\eps > 0$. Suppose the list size for $\calC_{\inn}$ when decoding from $\delta_{\inn}-\eps$ errors is $K=K(\eps)$, a constant independent of $d$, and $\lambda \leq \frac{\eps^2\cdot \delta_{\dec}^2}{4\cdot 10^4\cdot K^2}$. Then the code $\AELC$ can be decoded up to radius $\delta_{\inn}-2\eps$ in randomized time $\Ot{n\cdot |\Sigma_{\inn}|^d + nK}+T_{\dec}\cdot \exp(O(K^5/(\eps^3\cdot \delta_{\dec}^3)))$. Further, the produced list is of size at most $\exp(O(K^5/\eps^3\cdot \delta_{\dec}^3))$.
\end{theorem}
\begin{proof}
	Given any $g\in (\Sigma_{\inn}^d)^R$, we can use \cref{thm:ael_decoding_technical} with radius $\beta = \delta_{\inn} - 2\eps$ to recover the list $\calL(g,\beta) = \{h\in \AELC ~\vert~ \Delta_R(g,h) \leq \beta\}$, as long as the following conditions are met:
	\begin{itemize}
		\item The list size for $\calC_{\inn}$ is constant at $\beta+\eps = \delta_{\inn}-\eps$, which is true by our choice of $\calC_{\inn}$.
		\item The outer code can be unique decoded up to a radius $\delta_{\dec}$, which is true by assumption.
		\item $\lambda \leq \gamma$ for $\gamma = \frac{\eps \cdot \delta_{\dec}}{4K}$, which is true by our choice of $\lambda$. 
	\end{itemize}
	To analyze the runtime, we note that the time needed to list decode $\calC_{\inn}$ from radius $\beta+\eps$, $\tlocal{\beta+\eps}$, is at most $|\Sigma_{\inn}|^d$ by a brute force search. As long as $\lambda \leq \gamma^2/2500$, which holds by our choice of $\lambda$, the time needed to compute an $(\eta,\gamma)$-regular family with $\gamma = \frac{\eps \cdot \delta_{\dec}}{4K}$ and $\eta = \Omega(\gamma^2)$ is $\treg{\gamma} = \Ot{nd}$.
	
	The final runtime therefore is $n\cdot |\Sigma_{\inn}|^d+ K\cdot \Ot{nd} + T_{\dec}\cdot \exp(K^5/(\eps^3\cdot \delta_{\dec}^3))$.
\end{proof}
	
	Finally, we make some concrete choices, picking $\calC_{\out}$ be a linear-time unique decodable code of \cite{SS96, Zemor01} and $\calC_{\inn}$ to be a random linear code.
It was shown by Alrabiah, Guruswami and Li~\cite{AGL24} that for any $\eps>0$, a random
linear code of rate $R$ over alphabet $2\cdot 2^{10/\eps^2}$ has distance $1-R-2\eps$ and list decoding radius $1-R-2\eps$ with high probability, with a list size of $1/\eps$. The constant sized lists near the list decoding capacity $1-R$ makes them a suitable candidate to be used as the inner code $\calC_{\inn}$. 
Since the inner code is a constant-sized object, we can search over all linear codes in $(\F_q)^d$
of dimension $\rho \cdot d$ for a given rate $\rho$, and the code $\AELC$ then yields an explicit
construction of codes decodable in near-linear time up to list decoding capacity.

While we choose $\calC_{\inn}$ to be a random linear code to illustrate the power of AEL construction, one may also choose them to be algebraic codes such as folded Reed-Solomon, which are also known to achieve list decoding capacity with constant list sizes.
\begin{corollary}\label{cor:ael_decoding instantiation_final}
For every $\rho, \eps \in (0,1)$, there exists an infinite family of explicit codes $\AELC \subseteq (\F_q^d)^n$ obtained via the AEL construction that satisfy: 
\begin{enumerate}
\item $\rho(\AELC) \geq \rho$.
\item $\delta(\AELC) \geq 1-\rho-\eps$.
\item For any $g\in (\F_q^d)^R$, the list $\calL(g,1-\rho-\eps)$ is of size at most $\exp(O(1/\eps^{17}))$.
\item There is an algorithm that given any $g\in (\F_q^d)^R$, runs in time $\widetilde{O}_{\eps}(n)$ and outputs the list $\calL(g,1-\rho-\eps)$.
\item The alphabet size $q^d$ of the code $\AELC$ can be taken to be $2^{O(1/\eps^{22})}$.
\item $\AELC$ is characterized by parity checks of size $O(1/\eps^{23})$ over the field $\F_q$.
\end{enumerate}
\end{corollary}
\begin{proof}
Let $d = 2^{122}/\eps^{20}$ be such that there exist explicit infinite families of $(n,d,\lambda)$-expander graphs for some $\lambda \leq \frac{\eps^{10}}{2^{60}}$.
Let $\calC_{\inn} \sub \F_q^d$ be a random linear code of rate $\rho_{\inn} = \rho + \frac{\eps}{4}$ over $\F_q$, with $q = 2\cdot 2^{10\cdot 2^6/\eps^2} = 2^{O(1/\eps^2)}$, so that $\calC_{\inn}$ has distance at least $1-\rho_{\inn}-\frac{\eps}{4} =1-\rho-\frac{\eps}{2}$ and is list decodable up to $1-\rho_{\inn}-\frac{\eps}{4} = 1-\rho-\frac{\eps}{2}$ with a list size of at most $K=8/\eps$.

Finally, let $\calC_{\out} \subseteq (\F_q^{\rho_{\inn} \cdot d})^L$ be an outer (linear) code with rate $\rho_{\out} = 1 - \eps/4$ and distance (say) $\delta_{\out} = \eps^3/2^{15}$, and decodable up to radius $\delta_{\dec} = \eps^3/2^{17}$ in linear time $O(n)$. Explicit families of such codes can be obtained via Tanner code constructions and their linear time unique decoders. We note that 
\[
	\rho(\AELC) ~\geq~ \rho_{\out} \cdot \rho_{\inn} ~=~ (1-\eps/4) \cdot (\rho + \eps/4) ~\geq~ \rho \mper
\]

The claim about distance follows from \cref{thm:ael_distance}. Using \cref{thm:ael_decoding_instantiation} with $\delta_{\inn} = 1-\rho-\frac{\eps}{2}$ and decodability\footnote{We use decodability up to a slightly smaller radius to make it compatible with \cref{thm:ael_decoding_instantiation}. Tracing back further, we see that \cref{thm:ael_decoding_technical} only works if the decoding radius (for $\AELC$) $\beta$ is separated from the distance $\delta_{\inn}$ of the inner code. Although \cref{thm:ael_decoding_technical} does not necessarily need such separation between decoding radius of the \emph{inner code} $\calC_{\inn}$ and $\delta_{\inn}$, making them separated does not make much difference for our application.}  of $\calC_{\inn}$ up to $1-\rho-\frac{3\eps}{4}$ with list size $\frac{8}{\eps}$, we can decode $\AELC$ up to radius $1-\rho-\eps$ in time 
\[
	\Ot{n\cdot q^d}+n\cdot \exp(O(1/\eps^{17})) = \widetilde{O}_{\eps}(n) \mcom
\]
as long as $\lambda \leq \frac{(\eps/4)^2\cdot (\eps^3/2^{17})^2}{4\cdot 10^4\cdot (8/\eps)^2}$. With some calculation, we see that it suffices to choose $\lambda \leq \frac{\eps^{10}}{2^{60}}$, and this is true by our choice of degree $d$.

The alphabet size of the code $\AELC$ is $q^d = 2^{O(1/\eps^2)\cdot O(1/\eps^{20})} = 2^{O(1/\eps^{22})}$.

The choice to make $\calC_{\out}$ Tanner codes of rate $1-\eps/4$ allows us to conclude that $\calC_{\out}$ has parity checks of size at most $O(1/\eps^3)$. Therefore, $\AELC$ has parity checks of size at most $O(d \cdot (1/\eps^3)) = O(1/\eps^{23})$.
\end{proof}


%% file: list_recovery.tex
\subsection{List recovering AEL codes up to capacity}\label{sec:list-recovery}
We illustrate the flexibility of regularity lemma based decoders by showing how a relatively simple extension of the list decoder from \cref{sec:ael_decoding} allows us to perform list recovery for the AEL code in near-linear time.

The list recovery problem for a code $\calC$ of blocklength $n$ over alphabet $\Sigma$ is a generalization of list decoding, where the input is a collection $\calL_1,\calL_2,\cdots ,\calL_n$ of $n$ subsets of $\Sigma$, such that $|\calL_i|\leq \lsLR$ for every $i\in [n]$. In other words, $\calL_i \in \binom{\Sigma}{\leq \lsLR}$ for every $i\in [n]$. For the decoding radius $\beta$, the goal is to output the list
\[
	\calL\parens*{\inbraces{\calL_i}_{i\in [n]}, \beta} \defeq \braces[\Big]{h\in \calC ~|~ \Pr{i\in [n]}{h_i \not\in \calL_i}\leq \beta }
\]
We first generalize the notion of Hamming distance to this setting of small sets per coordinate.

\begin{definition}
	Let $\{\calL_i\}_{i\in [n]}$ be a collection of subsets of $\Sigma$, and let $h\in \Sigma^n$. Then we denote
	\[
		\Delta\parens*{\{\calL_i\}_{i\in [n]}, h} ~\defeq~ \Pr{i\in [n]}{h_i \not\in \calL_i}
	\]
\end{definition}

This definition is clearly a generalization of the usual Hamming distance, since for any $g\in \Sigma^n$, one could construct the collection of sets $\calL_i = \{g_i\}$ for each $i\in [n]$. Then, $\Delta\parens*{\{\calL_i\}_{i\in [n]}, h} = \Delta(g,h)$.
%

For the case of AEL codes, since the coordinate set is $R$, these input sets can be represented as $\{\listr\}_{\ri \in R}$, and each input set $\listr \sub \Sigma = \Sigma_{\inn}^d$. The goal then is to output all $h\in \AELC$ such that $\Pr{\ri\in R}{h_{\ri} \not\in \listr} \leq \beta$, or equivalently, all $h \in \AELC$ such that $\Delta\parens*{\{\listr\}_{\ri\in R},h} \leq \beta$.

Similar to the special case of list decoding, we will show how to perform list recovery efficiently when the inner code used in AEL construction is itself list recoverable up to any radius larger than $\beta$, the outer code is unique decodable up to some arbitrarily small radius, and $\lambda$ of the graph used is small enough. We will then instantiate the AEL construction with specific inner and outer codes to obtain final codes with rate $\rho$ that can be list recovered up to radius $1-\rho-\eps$ for arbitrary $\eps > 0$.

\subsubsection{Picking from local lists via regularity}
For a collection $\{\calL_{\ri}\}_{\ri\in R}$, with each $\calL_{\ri} \sub \Sigma$, we will obtain the list of codewords within a given error radius $\beta$, by first obtaining local lists $\calL_{\ell}$ for each $\ell \in L$, and then using the efficient regularity lemma to consistently select codewords from the local lists. 
As before, we consider a \emph{fixed} global codeword $h \in \AELC \sub \Sigma^R$ such that $\Pr{\ri \in R}{h_{\ri} \in \calL_{\ri}} \geq 1-\beta$ and show that there exists some element in the set we enumerate, which will allow us to recover $h$. As in \cref{sec:ael}, since $h$ is arbitrary, this will prove that an enumeration based algorithm will recover the entire list.

\paragraph{Local lists.} Note that for each edge $e = (\li,\ri) \in E(G)$, the input set $\listr$ on its right endpoint suggests at most $\lsLR$ many symbols of $\Sigma_{\inn}$. These symbols are simply $\braces*{\parens{f_{\ri}}_e ~|~ f_{\ri} \in \listr}$. Thus, we may assign to each edge a set $\calL_e \sub \Sigma_{\inn}$ with $|\calL_e| \leq \lsLR$.

For $\eps > 0$ to be chosen later, we define the local lists at vertices $\ell \in L$ as
\[
\calL_{\ell} ~=~ \braces*{f_{\ell} \in \cC_{\inn} ~\big\vert~ \Pr{e\in N(\li)}{\parens{f_{\li}}_e \not\in \calL_e} \leq \beta +\eps} \mper
\]
Just like in the case of list decoding, the expander mixing lemma implies that for each global codeword $h$ that agrees with $\{\listr\}_{\ri \in R}$ in at least $1-\beta$ fraction of $\ri \in R$, its local projection must be in most local lists of radius $\beta$.
\begin{claim}\label{clm:local-membership-lr}
Let $\{\listr\}_{\ri\in R}$ be a collection with $|\listr| \leq \lsLR$ for each $\ri\in R$, and $h$ be an arbitrary codeword of $\AELC$ that satisfies $\Pr{\ri\in R}{h_{\ri} \not\in \listr} \leq \beta$. If $\lambda \leq \gamma \cdot \eps$, then, $\Pr{\li \in L}{h_{\li}\not\in \listl} \leq \gamma$. 
\end{claim}
%
\begin{proof}
We start with the observation that $h_{\li} \not\in \listl$ only if $\Pr{e\in N(\li)}{h_e\not\in \calL_e} > \beta+\eps$. 

Let $L' = \inbraces{\ell \in L ~|~ \Pr{e\in N(\li)}{h_e \not\in \calL_e} \geq \beta + \eps}$. Then each $\ell \in L'$ has at least $(\beta + \eps) \cdot d$ incident edges on which $h_e \not\in \calL_e$. For such an edge $e=(\li,\ri)$, it must also be the case that $h_{\ri} \not\in \listr$. Let $R' \subseteq R$ be the set of the right endpoints of these edges. Then, we have 
\[
(\beta + \eps) \cdot d \cdot \abs{L'} ~\leq~ E_G(L',R') 
~\leq~ \frac{d}{n} \cdot \abs{L'} \cdot \abs{R'}  + \lambda \cdot dn \mper
\]
Since $h_{\ri}\not\in \listr$ for every $\ri \in R'$, we have $\abs{R'} \leq \Pr{\ri\in R}{h_{\ri}\not\in \listr} \cdot |R| \leq \beta \cdot n$, which proves the claim.
\end{proof} 
Now we have two collection of lists, $\{\listl\}_{\li \in L}$ and $\{\listr\}_{\ri\in R}$, such that $\Delta\parens*{\{\listl\}_{\li\in L}, h} \leq \gamma$ and $\Delta\parens*{\{\listr\}_{\ri\in R}, h}~\leq~\beta$. In the case of list decoding, we only had lists on one side of the graph. Nevertheless, our argument from \cref{sec:ael_decoding} for stitching together an assignment close to $h$ from these lists will extend seamlessly to the list recovery setting.

We will assume an upper bound $K$, which is crucially independent of $d$, on the list size when list recovering from $\beta+\eps$ errors for the inner code $\cC_{\inn}$. As before, we will think of the local lists as sets of size \emph{exactly} $K$, by including arbitrary codewords (distinct from list elements) if necessary. On the right, the input lists $\{\listr\}_{\ri\in R}$ are already of size at most $\lsLR$, and we will again think of them as sets of size \emph{exactly} $\lsLR$ by adding new symbols of $\Sigma$ if necessary. We will also impose an arbitrary ordering for each $\calL_{\ell}$ and for each $\listr$. This will allow us to use $\calL_{\ell}[i]$ to refer to the $i$-th element of $\listl$ for $i\in [K]$, and $\listr[j]$ to refer to the $j$-th element of $\listr$ for $j\in [\lsLR]$.

\paragraph{Edge-densities in agreement graphs.}  Fix a global codeword $h$ such that $\Pr{\ri \in R}{h_{\ri}\not\in \listr} \leq \beta$.
As in the case of list decoding, we consider the (unknown) sets $A_1, \ldots, A_K$ defined by the the position of $h_{\ell}$ in the local lists
\[
A_i ~=~ \inbraces{\ell \in L ~|~ \calL_{\ell}[i] = h_{\ell}} \mper
\]
The sets $\{A_i\}_{i \in [K]}$ are disjoint, and \cref{clm:local-membership-lr} implies that $\sum_{i}\abs{A_i} \geq (1-\gamma) \cdot \abs{L}$.
As before, if we can recover the sets $\{A_i\}_{i \in [K]}$ up to a small fraction of errors, then we correctly know $h_{\ell}$ for most $\ell \in L$, and can find $h$ by unique decoding the outer code $\cC_0$. 

In the case of list decoding, we recovered the sets $A_i$ by understanding their interaction with some unknown set $B$ that was defined to be the set of $\ri \in R$ where $h_{\ri}$ agreed with the received word. For list recovery, we will need to study the interaction between sets $A_1,A_2,\cdots ,A_K$ with a \emph{collection} of disjoint sets $B_1, B_2,\cdots ,B_{\lsLR}$. For $j\in [\lsLR]$, we define
\[
	B_j = \inbraces{\ri \in R ~|~ h_{\ri} = \listr[j]}
\]

The condition $\Pr{\ri\in R}{h_{\ri} \not\in \listr} \leq \beta$ on $h$ can now be written in terms of the sets $B_1,B_2,\cdots ,B_{\lsLR}$ as $\sum_{j=1}^{\lsLR} |B_i| \geq (1-\beta) \cdot |R|$.

Consider the graph $H_{ij}= (L,R, E(H_{ij}))$ constructed by deleting all edges of $G$ that are inconsistent between assignments picked according to index $i\in [K]$ on the left and index $j\in [\lsLR]$ on the right,
\[
E(H_{ij}) ~=~ \{e = (\li, \ri) \in E(G) ~|~ (\calL_{\ell}[i])_e = (\listr[j])_e \} \mper
\]
Note that the graphs $H_{ij}$ can be easily computed based on the local lists $\{\listl\}_{\li\in L}$ and $\{\listr\}_{\ri\in R}$.
Also, we must have $E_{H_{ij}}(A_i, B_j) = E_G(A_i,B_j)$ since for an edge $e = (\ell, r) \in E_G(A_i,B_j)$, we have
$(\listl[i])_e = h_e = (\listr[j])_e$, and therefore such an edge will never be deleted when constructing $H_{ij}$.
%

Finally, we prove an analog of the rigidity lemma in the setting of list recovery. This lemma shows that a sufficiently regular family $\family$, which is \emph{simultaneously} $(\eta,\gamma)$-regular with respect to $\{H_{ij}\}_{i\in [K],j\in [\lsLR]}$, allows us to identify the sets $A_1, \ldots, A_K$.
As before, we can obtain such a family by applying the regularity lemma for each of the graphs $H_{ij}$ for $i\in [K]$ and $j\in [\lsLR]$, and then taking the union.

\begin{lemma}[Rigidity]\label{lem:ael-rigidity-lr}
Let the family $\family$ be simultaneously $(\eta,\gamma)$-regular for the graphs $\{H_{ij}\}_{i\in [K],j\in [\lsLR]}$. Let $S_1, \ldots, S_K \subseteq L$ be any collection of disjoint sets satisfying $\max_{i \in [K]} ~\norm{\indicator{A_i} - \indicator{S_i}}_{\family} ~\leq~ \eta$. 
Also, let the list recovery radius $\beta$ satisfy $\beta \leq  \delta_{\inn} - \eps$, and assume $\lambda \leq \gamma$. Then, $\sum_{i} \abs{S_i \setminus A_i} \leq (3\lsLR K\gamma/\eps) \cdot n$.
\end{lemma}
\begin{proof}
By the $\eta$-closeness of the sets $A_i$ and $S_i$ with respect to $\family$, we have
\[
\forall i \in [K], j\in [\lsLR] \qquad \abs{E_{H_{ij}}(S_i,B_j)} ~\geq~ \abs{E_{H_{ij}}(A_i,B_j)} - \gamma \cdot dn ~=~ \abs{E_{G}(A_i,B_j)} - \gamma \cdot dn
~\geq \abs{E_{G}(S_i,B_j)} - 2\gamma \cdot dn \mper
\]
Summing the above and using the bound $\sum_i\abs{A_i} \geq (1-\gamma) \cdot n$ from \cref{clm:local-membership-lr}, we get
\[
\sum_{\substack{i \in [K] \\ j\in [\lsLR]}} \abs{E_{H_{ij}}(S_i,B_j)} 
~~\geq~~ \sum_{\substack{i\in [K] \\ j\in [\lsLR]}} \abs{E_{G}(S_i,B_j)} - 2\lsLR K \cdot \gamma \cdot dn 
\]
Define the sets $W_i = S_i \setminus A_i$. 
For each $\ell \in W_i$, we have $(\calL_{\ell}[i])_e \neq h_e$ for at least $\delta_{\inn}$ fraction of the incident edges on $\li$ since $\calL_{\ell}[i]$ and $h_{\ell}$ must be distinct codewords in $\cC_{\inn}$. 
Also, such an edge $e=(\li,\ri)$ cannot be in $E_{H_{ij}}(S_i,B_j)$ since we have $(\calL_{\ell}[i])_e \neq h_e = (\listr[j])_e$. Using this, and the fact that $S_i \setminus W_i \subseteq A_i$, we get
\begin{align*}
\sum_{\substack{i \in [K] \\ j\in [\lsLR]}} \abs{E_{H_{ij}}(S_i,B_j)} &~=~ \sum_{\substack{i \in [K]\\ j\in [\lsLR]}} \abs{E_{H_{ij}}(S_i \setminus W_i,B_j)} + \sum_{\substack{i \in [K]\\ j\in [\lsLR]}} \abs{E_{H_{ij}}(W_i, B_j)} \\
&~\leq~ \sum_{\substack{i \in [K]\\ j\in [\lsLR]}} \abs{E_{G}(S_i \setminus W_i,B_j)} + \sum_{i \in [K]} \abs{W_i} \cdot (1 - \delta) \cdot d \\
&~=~ \sum_{\substack{i \in [K]\\ j\in [\lsLR]}} \abs{E_{G}(S_i,B_j)} - \sum_{\substack{i \in [K]\\ j\in [\lsLR]}} \abs{E_{G}(W_i,B_j)}+ \sum_{i \in [K]} \abs{W_i} \cdot (1 - \delta) \cdot d \mper
\end{align*}
We can now use expander mixing lemma to bound the number of edges in $G$.
\[
\sum_{\substack{i \in [K] \\ j\in [\lsLR]}} \abs{E_{G}(W_i,B_j)} 
~\geq~ \frac{d}{n} \cdot \inparen{\sum_i \abs{W_i}} \cdot \inparen{\sum_j \abs{B_j}} - \lsLR K \cdot \lambda \cdot dn 
~\geq~ (1-\beta) \cdot d\cdot \inparen{\sum_i \abs{W_i}} - \lsLR K \cdot \lambda \cdot dn 
\]
Combining the above inequalities, and using $w$ to denote $\sum_i \abs{W_i}/n$, we get
\begin{align*}
\sum_{\substack{i\in [K] \\ j\in [\lsLR]}} \abs{E_{G}(S_i,B_j)} - 2\lsLR K \cdot \gamma \cdot dn  &~\leq~ \sum_{\substack{i \in [K] \\ j\in [\lsLR]}} \abs{E_{H_{ij}}(S_i,B_j)} 
~\leq~
\sum_{\substack{i \in [K]\\ j\in [\lsLR]}} \abs{E_{G}(S_i,B_j)} - \sum_{\substack{i \in [K]\\ j\in [\lsLR]}} \abs{E_{G}(W_i,B_j)}+ \sum_{i \in [K]} \abs{W_i} \cdot (1 - \delta) \cdot d\\
\implies~ & -2\lsLR K \cdot \gamma \cdot dn \leq -(1-\beta)w\cdot dn + \lsLR K \cdot \lambda \cdot dn + w(1-\delta) \cdot dn \\
\implies~ & w(\delta-\beta) \leq 2\lsLR K \cdot \gamma +\lsLR K \cdot \lambda\\
\implies~ & w \leq \frac{2\lsLR K \cdot \gamma + \lsLR K \cdot \lambda}{\delta - \beta} \mper
\end{align*}
Using $\lambda \leq \gamma$ and $\beta \leq \delta - \eps$ then proves the claim.
\end{proof}

\subsubsection{The list recovery algorithm}
%
%
Let $\AELC$ be a code obtained via the AEL construction, using the outer code $\cC_{\out}$ and inner
code $\cC_{\inn}$, and using the graph $G=(L,R,E)$ which is an $(n,d,\lambda)$-expander for a sufficiently
small $\lambda$.
Further, we assume that the inner code is list recoverable with input list size $\lsLR$ up to radius $\beta+\eps$ for some $\eps>0$, and the output list size for this list recovery is at most $K$.
%
Then the following algorithm takes as input $\{\listr\}_{\ri\in R}$, where each $|\listr| \leq \lsLR$, and finds the list of all codewords $h\in \AELC$ such that $\Pr{\ri\in R}{h_{\ri} \not\in \listr} \leq \beta$. This is done using $\lsLR K$ many calls to the algorithmic regularity lemma
(\cref{cor:expander_regularity}), and constantly many calls to the unique-decoder for $\cC_{\out}$, where the constant only depends on
the parameters $\delta_{\dec}$, $\eps$, $k$, and $K$.

\medskip
\begin{new-algorithm}{List recovering $\AELC$}\label{algo:ael-lr} 
\begin{tabular}{r l}
\textsf{Input:} & $\{\listr\}_{\ri\in R}$, with each $\listr \sub \Sigma$ and $|\listr| \leq \lsLR$ \\[3 pt]
\textsf{Output:} & List $\calL \subseteq \AELC$\\[3 pt] 
\textsf{Parameters:} & List decoding radius  $\beta \leq \delta_{\inn} - \eps$, \\
& list size bound $K$
                       for  $\cC_{\inn}$ when list recovering from radius $\beta + \eps$, \\
& unique decoding radius $\delta_{\dec}$ for $\cC_{\out}$
\end{tabular}
\begin{itemize}
\medskip
\item For each $e = (\li,\ri) \in E(G)$, compute the local list $\calL_e \sub \Sigma_{\inn}$, defined as \[ \calL_e = \inbraces{ \sigma \in \Sigma_{\inn} ~|~ \exists f_{\ri} \in \listr \text{ such that } (f_{\ri})_e = \sigma } \mper\]
\item For each $\ell \in L$, compute the local list $\calL_{\ell} = \inbraces{f_{\ell} \in \cC_{\inn} ~|~ \Pr{e\in N(\li)}{(f_{\ell})_e \not\in \calL_e} \leq \beta + \eps}$. Pad with additional elements of $\cC_{\inn}$ if needed, to view each $\calL_{\ell}$ as an ordered list of size $K$.
\item Pad each set of $\{\listr\}_{\ri\in R}$ with additional elements of $\Sigma = \Sigma_{\inn}^d$ if needed, to view each $\listr$ as an ordered list of size $\lsLR$.
\item For each $i \in [K], j\in [\lsLR]$ define a bipartite graph $H_{ij}$ with vertex sets $(L,R)$ and edges
\[
E(H_{ij}) ~=~ \{e = (\li, \ri) \in E(G) ~|~ (\listl[i])_e = (\listr[j])_e \} \mper
\]
For $\gamma = (\eps \cdot \delta_{\dec})/(5\lsLR K)$ and sufficiently small $\eta$, compute a family $\family_{ij}$ which is $(\eta,\gamma)$-regular for $H_{ij}$. Let $\family$ be the union of $\{\family_{ij}\}_{i \in [K], j\in [\lsLR]}$.
\item For each $\sigma = (\sigma_1, \ldots, \sigma_K) \in \calN_{\eta/4}(\family, K)$
\begin{itemize}
\item Consider disjoint sets $S_1, \ldots, S_K \subseteq L$ given by \cref{clm:enumeration} with $\max_{i \in [K]} \norm{\sigma_i - \sigma(\indicator{S_i})}_{\infty} \leq \eta/2$.
\item Define $h' \in (\cC_{\inn})^L$ as
\[
h'_{\ell} ~=~
\begin{cases}
\calL_{\ell}[i] &~\text{if}~ \ell \in S_i \\
\text{arbitrary} &~\text{if}~ \ell \notin \cup_i S_i
\end{cases}
\]
\item Unique decode $h'$ to find $h_0 \in \cC_{\out}$ and corresponding $h \in \AELC$. 
\item If $\Pr{\ri \in R}{h_{\ri} \not\in \listr} \leq \beta$, then $\calL \leftarrow \calL \cup \{h\}$.
\end{itemize}
\item Return $\calL$.
\end{itemize}
\end{new-algorithm}

\medskip
\begin{theorem}\label{thm:ael_recovery_technical}
%
%
Let $\tlocal{\delta}$ be the time required to list recover $\calC_{\inn}$ up to radius $\delta$, let $\treg{\gamma}$ be the time required to obtain an $(\eta,\gamma)$-regular family for any subgraph of
$G$ with $\eta = \Omega(\gamma^2)$, and let $T_{\dec}$ be the running time of the unique-decoder for the outer code $\calC_{\out}$ that can decode $\delta_{\dec}$ fraction of errors.
%

Given any $\beta,\eps > 0$ with $\beta \leq \delta_{\inn} -\eps$, let $K$ be an upper bound on the list size of $\calC_{\inn}$ when list recovering up to radius $\beta+\eps$, and let $\gamma = (\eps\cdot \delta_{\dec})/(5\lsLR K)$. If it holds that $\lambda \leq \gamma$, then given any $\{\listr\}_{\ri \in R} \in \binom{\Sigma}{\leq \lsLR}^R$ as input, \cref{algo:ael-lr}
runs in time 
\[
	O(n \cdot \tlocal{\beta+\eps} + \lsLR K \cdot \treg{\gamma}) + T_{\dec} \cdot (1/\gamma)^{O(\lsLR K^2/\gamma^2)} \mcom
\]
and returns the list $\calL = \inbraces{h \in \AELC ~|~ \Pr{\ri\in R}{h_{\ri} \not\in \listr} \leq
\beta}$. Further, the list is of size at most $(1/\gamma)^{O(\lsLR K^2/\gamma^2)}$.
\end{theorem}
\begin{proof}
We note that the computation of the local lists $\{\calL_e\}_{e\in E}$ and $\{\listl\}_{\li \in L}$, and the graphs $H_{ij}$ can be done in time $O( n \cdot \tlocal{\beta+\eps})$. 
By \cref{lem:regularity-family}, we can take $\eta = \Omega(\gamma^2)$ and obtain an $(\eta,\gamma)$-regular family $\family_{ij}$ for each graph $H_{ij}$ in time $\treg{\gamma}$, with $\abs{\family_{ij}} = O(1/\gamma^2)$. Thus, for the union $\family$, we have $\abs{\family}= O(\lsLR K/\gamma^2))$
By \cref{clm:enumeration}, we have $\abs{\calN_{\eta/2}(\family,K)} = (1/\eta)^{O(\lsLR K \cdot \abs{\family})} = (1/\gamma)^{O(\lsLR K^2/\gamma^2)}$, and the enumeration can be done in time $T_{\dec} \cdot (1/\gamma)^{O(\lsLR K^2/\gamma^2)}$.

Finally, it remains to prove that the algorithm correctly recovers the list. 
Note that for any $h\in \calL$, and corresponding sets $\inbraces{A_i}_{i \in [K]}$ indicating its position in local lists, the $h'$ found by our algorithm agrees with $h$ on all vertices $\ell \in \cup_i (S_i \cap A_i)$, and thus
\[
\dis_L(h,h') ~\leq~ 1 - \frac{1}{n} \cdot \sum_i \abs{S_i \cap A_i} 
~=~ 1 - \sum_i \angles{\indicator{S_i},\indicator{A_i}} 
~=~ 1 - \sum_i \ExpOp[\indicator{S_i} - \indicator{S_i \setminus A_i}] \mper
\] 
By \cref{clm:enumeration} there \emph{exists} $\sigma = (\sigma_1, \ldots, \sigma_K) \in
\calN_{\eta/4}(\family,K)$ such that disjoint sets $S_1, \ldots, S_K$ produced using $\sigma$ 
will satisfy $\max_{i \in [K]} \norm{\indicator{A_i}-\indicator{S_i}}_{\family} \leq \eta$.
Also, by \cref{lem:ael-rigidity-lr}, for any such collection $S_1, \ldots, S_K$, we have $\sum_i \ExpOp[\indicator{S_i \setminus A_i}] \leq (3\lsLR K\gamma/\eps)$. Moreover, by \cref{clm:local-membership-lr}, we have $\sum_i \abs{A_i} \geq (1-\gamma) \cdot n$ and $\norm{\indicator{A_i} - \indicator{S_i}}_{\factor} \leq \eta$ implies
\[
\dis_L(h,h') 
~\leq~ 1 - \sum_i (\ExpOp[\indicator{A_i}] - \gamma) + \sum_{i}\ExpOp[\indicator{S_i \setminus A_i}] 
~\leq~ \gamma + K \gamma + 3\lsLR K\gamma/\eps ~\leq~ 5\lsLR K\gamma/\eps\mcom
\]
which is at most $\delta_{\dec}$ by our choice of $\gamma$. 
Thus, for every $h \in \calL$, one of the choices in our enumeration finds an
$h'$ satisfying $\dis_L(h,h') \leq \delta_{\dec}$. Unique decoding this $h'$ then recovers $h$, and
thus the enumeration algorithm recovers all $h$ in the list $\calL$.
\end{proof}

\subsubsection{Instantiating AEL codes for list recovery}
As in the case of list decoding, we instantiate \cref{thm:ael_recovery_technical} with a brute force inner code decoder and the algorithm of \cref{lem:regularity-family} for finding regular factors, to obtain a general reduction of list recovery of AEL codes up to the distance of the inner code, to unique decoding the outer code $\calC_{\out}$. Again, this is assuming the inner code can be list recovered from a radius close to its distance.

\begin{theorem}[List Recovery of AEL Codes up to Capacity]\label{thm:ael_recovery_instantiation}
	Let $\AELC$ be the code obtained via the AEL construction applied to an outer code $\calC_{\out}$ over alphabet $\Sigma_{\out}$ and an inner code $\calC_{\inn}$ over alphabet $\Sigma_{\inn}$ using an $(n,d,\lambda)$-expander $G=(L,R,E)$. Assume that the outer code can be unique decoded from a fraction $\delta_{\dec}$ of errors in time $T_{\dec}$.
	
	Fix an arbitrary $\eps > 0$. Suppose the list size for $\calC_{\inn}$ when list recovering from $\delta_{\inn}-\eps$ errors is $K=K(\eps)$, a constant independent of $d$, and $\lambda \leq \frac{\eps^2\cdot \delta_{\dec}^2}{4\cdot 5^6 \cdot \lsLR^2 K^2}$. Then the code $\AELC$ can be list recovered up to radius $\delta_{\inn}-2\eps$ in randomized time $\Ot{n\cdot |\Sigma_{\inn}|^d + \lsLR K \cdot n}+T_{\dec}\cdot \exp(O(\lsLR^4K^5/(\eps^3\cdot \delta_{\dec}^3)))$. Further, the produced list is of size at most $\exp(O(\lsLR^4K^5/\eps^3\cdot \delta_{\dec}^3)))$.
\end{theorem}
\begin{proof}
	Given any input $\{\listr\}_{\ri\in R}$ to the list recovery problem, with each $|\listr|\leq \lsLR$, we can use \cref{thm:ael_recovery_technical} with radius $\beta = \delta_{\inn} - 2\eps$ to recover the list $\calL(g,\beta) = \{h\in \AELC ~\vert~ \Delta_R(g,h) \leq \beta\}$, as long as the following conditions are met:
	\begin{itemize}
		\item The output list size is constant when list recovering $\calC_{\inn}$ from radius $\beta+\eps = \delta_{\inn}-\eps$, which is true by our choice of $\calC_{\inn}$.
		\item The outer code can be unique decoded up to a radius $\delta_{\dec}$, which is true by assumption.
		\item $\lambda \leq \gamma$ for $\gamma = \frac{\eps \cdot \delta_{\dec}}{5\lsLR K}$, which is true by our choice of $\lambda$. 
	\end{itemize}
	To analyze the runtime, we note that the time needed to list recover $\calC_{\inn}$ from radius $\beta+\eps$, $\tlocal{\beta+\eps}$, is at most $|\Sigma_{\inn}|^d$ by a brute force search. As long as $\lambda \leq \gamma^2/2500$, which holds by our choice of $\lambda$, the time needed to compute an $(\eta,\gamma)$-regular family with $\gamma = \frac{\eps \cdot \delta_{\dec}}{5\lsLR K}$ and $\eta = \Omega(\gamma^2)$ is $\treg{\gamma} = \Ot{nd}$.
	
	The final runtime therefore is $n\cdot |\Sigma_{\inn}|^d+ \lsLR K \cdot \Ot{nd} + T_{\dec}\cdot \exp(\lsLR^4 K^5/(\eps^3\cdot \delta_{\dec}^3)))$.
\end{proof}
	
Finally, we make the same choice for $\calC_{\out}$ as in the case of list decoding, and again choose $\calC_{\inn}$ to be a random linear code. For the inner code's list recoverability, we rely on a recent result of Li and Shagrithaya~\cite{LS25} that shows that rate $R$ random linear codes over a large enough alphabet can be list recovered from radius $1-R-\eps$ with an output list size of $\parens*{2\lsLR/\eps}^{2\lsLR/\eps}$, where $k$ denotes the size of the input lists in the list recovery problem. As before, we could have also used folded Reed-Solomon codes which have a list size of $\parens*{\lsLR/\eps}^{O((\log \lsLR)/\eps)}$ \cite{KRSW23, Tamo24}, but we choose to proceed with the random linear code result to highlight the combinatorial flavor of the our results.
%
%
As before, we can rely on a brute force search over all (constant-sized) linear codes of a fixed rate to find a list recoverable inner code, whose existence is guaranteed by \cite{LS25}.

\begin{theorem}[Theorem 3.1 in \cite{LS25}]
	Let $\rho,\eps>0$ such that $1-\rho -\eps>0$, and let $\lsLR \in \N$. For $q \geq k^{14/\eps}$, a random linear code of rate $\rho$ in $\F_q^n$ can be list recovered with input list size $\lsLR$ and output list size at most $(2\lsLR/\eps)^{2\lsLR/\eps}$.
\end{theorem}

%
\begin{corollary}\label{cor:ael_recovery_instantiation_final}
For every $\rho, \eps \in (0,1)$ and $\lsLR \in \N$, there exists an infinite family of explicit codes $\AELC \subseteq (\F_q^d)^n$ obtained via the AEL construction that satisfy: 
\begin{enumerate}
\item $\rho(\AELC) \geq \rho$.
\item $\delta(\AELC) \geq 1-\rho-\eps$.
\item For any $\{\listr\}_{\ri\in R}$, with each $\listr \sub \F_q^d$ and $|\listr|\leq \lsLR$, the list $\calL\parens*{\{\listr\}_{\ri\in R},1-\rho-\eps}$ is of size at most $\exp(\exp(O(\frac{\lsLR}{\eps}  \log \frac{k}{\eps} )))$.
\item There is an algorithm that given any $\{\listr\}_{\ri\in R}$, with each $\listr \sub \F_q^d$ and $|\listr|\leq \lsLR$, runs in time $\widetilde{O}_{\lsLR, \eps}(n)$ and outputs the list $\calL\parens*{\{\listr\}_{\ri\in R},1-\rho-\eps}$.
\item The alphabet size $q^d$ of the code $\AELC$ can be taken to be $\exp(\exp(O(\frac{k}{\eps} \log \frac{k}{\eps})))$.
\item $\AELC$ is characterized by parity checks of size $\exp(\frac{k}{\eps}\log\frac{k}{\eps})$ over the field $\F_q$.
\end{enumerate}
\end{corollary}
\begin{proof}

Let $d = 2^{O(\frac{k}{\eps} \log \frac{k}{\eps})}$ be such that there exist explicit infinite families of $(n,d,\lambda)$-expander graphs for some $\lambda \leq 2^{-C\frac{k}{\eps}\log\frac{k}{\eps}}$, where $C$ is a large constant.
Let $\calC_{\inn} \sub \F_q^d$ be a random linear code of rate $\rho_{\inn} = \rho + \frac{\eps}{4}$ over $\F_q$, with $q = \lsLR^{2^6/\eps} = \lsLR^{O(1/\eps)}$, so that $\calC_{\inn}$ has distance at least $1-\rho_{\inn}-\frac{\eps}{4} =1-\rho-\frac{\eps}{2}$ and is list recoverable from radius $1-\rho_{\inn}-\frac{\eps}{4} = 1-\rho-\frac{\eps}{2}$ with a list size of at most $K=(8\lsLR/\eps)^{8\lsLR/\eps}$.

Let $\calC_{\out} \subseteq (\F_q^{\rho_{\inn} \cdot d})^L$ be the same choice as in the case of list decoding. That is, $\calC_{\out}$ is a linear code with rate $\rho_{\out} = 1 - \eps/4$ and distance (say) $\delta_{\out} = \eps^3/2^{15}$, and decodable up to radius $\delta_{\dec} = \eps^3/2^{17}$ in linear time $O(n)$. 
By our choice of parameters,
\[
	\rho(\AELC) ~\geq~ \rho_{\out} \cdot \rho_{\inn} ~=~ (1-\eps/4) \cdot (\rho + \eps/4) ~\geq~ \rho \mper
\]

The claim about distance again follows from \cref{thm:ael_distance}. Using \cref{thm:ael_recovery_instantiation} with $\delta_{\inn} = 1-\rho-\frac{\eps}{2}$ and decodability
of $\calC_{\inn}$ up to $1-\rho-\frac{3\eps}{4}$ with output list size $(8\lsLR/\eps)^{8\lsLR/\eps}$, we can decode $\AELC$ up to radius $1-\rho-\eps$ in time 
\[
	\Ot{n\cdot q^d + \lsLR K \cdot n}+n\cdot \exp(O(\lsLR K^2/\eps^{12})) = \widetilde{O}_{\lsLR, \eps}(n) \mcom
\]
as long as $\lambda \leq \frac{(\eps/4)^2\cdot (\eps^3/2^{17})^2}{4\cdot 5^6\cdot \lsLR^2 \cdot(8\lsLR /\eps)^{16\lsLR/\eps}} = \frac{1}{2^{O(\frac{k}{\eps}\log\frac{k}{\eps})}}$. We chose the degree $d$ so that this is true.

The alphabet size of the code $\AELC$ is $q^d = 2^{O(k/\eps)\cdot 2^{O(\frac{k}{\eps}\log \frac{k}{\eps})}} = 2^{2^{O(\frac{k}{\eps}\log \frac{k}{\eps})}}$, and $\AELC$ has parity checks of size at most $O(d \cdot (1/\eps^3)) = 2^{O(\frac{k}{\eps}\log\frac{k}{\eps})}$.
\end{proof}

\begin{remark}
Note that compared to the list decoding case, the bounds for list recovery are off by an exponential factor. The source of this loss is the exponential list size bound from \cite{LS25} for random linear codes that we use as inner codes. This can be remedied by using a uniformly random (non-linear) code instead of the random linear code as an inner code, which are known to achieve the optimal list size of $O(k/\eps)$ \cite{Res20:thesis} instead of the $(k/\eps)^{O(k/\eps)}$ that we use right now. However, we chose to instantiate with the suboptimal (in terms of list size) random \emph{linear} code because they preserve additional structural properties such as linearity, and allow us to get LDPC codes that achieves list recovery capacity.
\end{remark}



%% file: tanner.tex
\section{Algorithms for Tanner Codes}\label{sec:tanner}

Let $G = (L,R,E)$ be a bipartite $(n,d,\lambda)$-expander. Let $\codeL$ and $\codeR$ be two codes over a common alphabet $\Sigma$, common blocklength $d$ and distances $\delL$ and $\delR$ respectively, so that the Tanner code on graph $G$ with left code $\codeL$ and right code $\codeR$ is defined as
\[
	\TanC = \{ h\in \Sigma^E ~|~ \forall \li \in L, ~ h_{\li} \in \codeL \text{ and } \forall \ri\in R, ~ h_{\ri} \in \codeR\} \mper
\] 

\begin{lemma}[Distance of Tanner codes]\label{lem:tanner_dist}
	The distance of the code $\TanC$ is at least $\sqrt{\delL \delR}\parens*{\sqrt{\delL \delR} - \lambda} \geq \delL \delR - \lambda$.
\end{lemma}

That is, by making $\lambda$ small enough, the distance of Tanner codes approaches the product of distance $\delL \delR$. We will not be much concerned with the rate of this code, but we mention that if the codes $\codeL$ and $\codeR$ are linear codes with rates $\rho_L$ and $\rho_R$ respectively, then the rate of $\TanC$ is at least $\rho_L+\rho_R-1$. In particular, if both $\rho_L$ and $\rho_R$ are close to 1, so is the rate of $\TanC$, and this makes Tanner codes particularly suitable as an outer code in the AEL code construction.

The following lemma is well-known as unique decodability of Tanner codes and goes back to \cite{SS96, Zemor01}, but we need it with two minor modifications: (1) we need an error-and-erasure version \cite{SkaRoth03}, and (2) we will think of decoding from a $g \in (\codeL \cup \{\bot\})^L$ instead of from $g\in (\Sigma \cup \{\bot\})^E$. For this reason, we adapt the proof of \cite[Lemma 2.4]{RWZ21} to this setting in \cref{sec:appendix_tanner}.

\begin{lemma}[Unique-decodability of Tanner codes]\label{lem:err_and_erasure}
	Let $\TanC$ be a Tanner code defined as in \cref{lem:tanner_dist}. Fix $\eps > 0$ such that $\lambda \leq \frac{\eps}{8} \cdot \min(\delL,\delR)$. Given a $g\in (\codeL \cup \{\bot\})^L$ with $s\cdot |L|$ many erasures, we have the following:
	\begin{itemize}
	\item there is at most one $h\in \TanC$ that satisfies $2 \dis_L(g,h) + s \leq \delR - 4\eps$.
	\item If such an $h$ exists, then there is an algorithm that given $g$, runs in linear time $O(n\cdot |\Sigma|^d)$, and outputs $h$.
	\end{itemize}
\end{lemma}

Tanner codes are well known for this linear time unique decodability, and were the first explicit construction of LDPC codes. Beyond their practical uses, they were used as the base (outer) code in an application of the AEL procedure by Guruswami and Indyk \cite{GI05} to obtain linear-time unique decodable near the Singleton bound. However, not much was known about the list decoding of Tanner codes until the work of \cite{JST23}, who used the Sum-of-Squares hierarchy to decode them up to their Johnson bound.

In this section, we show how Tanner codes can be list decoded even beyond the Johnson bound, using simple extensions of the techniques we already saw \cref{sec:ael} for AEL codes. Suppose the code $\codeL$ is decodable with constant-sized lists (independent of its blocklength $d$) up to a radius $\decL$, and the code $\codeR$ is decodable with constant-sized lists up to a radius $\decR$. We will design an algorithm that can decode arbitrarily close to the radius $\decT$ (by making $\lambda$ small enough), where 
\[
	\decT ~\defeq \min\parens*{\delL\decR,\decL\delR}
\]
This radius matches the decoding radius obtained for tensor codes, which are dense analogs of Tanner codes, by \cite{GGR09}. Note that in the worst case, the decoding radius (with constant sized lists) for $\codeL$ is just $\decL = \calJ(\delL)$, where $\calJ(\delL)$ denotes the Johnson bound corresponding to distance $\delL$.

Suppose the codes $\codeL$ and $\codeR$ are same (this was the setting addressed by \cite{JST23}), and $\lambda$ is small enough. Let $\delta$ be their common distance, and let $\delta^{\dec}$ be their common decoding radius, which is at least $\calJ(\delta)$. Then based on just the distance $\delta$ of the inner code, our algorithm will decode up to a radius $\delta\cdot \calJ(\delta)$, which is already better than $\calJ(\delta^2)$. We note that even combinatorial decodability beyond Johnson bound was not known for Tanner codes at all. When we know more about the inner code, such as if it is chosen to be decodable up to its distance $\delta$ (\ie $\delta^{\dec} \approx \delta$), then we can conclude that the Tanner code is decodable all the way up to its (designed) distance of $\delta^2$.

\subsection{List decoding Tanner codes}
Let $g\in \Sigma^E$ be a received word, and our goal is to design an algorithm that for any $\eps>0$, outputs the list 
\[
	\calL = \calL(g, \decT-\eps) = \{h\in \TanC ~\vert~ \Delta(g,h) \leq \decT-\eps\} \mper
\] 
\paragraph{Local lists.}
We start by $\codeL$-decoding $g_{\li}$ for each $\li \in L$ up to radius $\decL$ to obtain a collection of lists $\{\listl\}_{\li\in L}$ with each list containing at most $\lsL$ codewords of $\codeL$. Likewise, we also $\codeR$-decode $g_{\ri}$ for each $\ri\in R$ up to radius $\decR$ to obtain another collection of lists $\{\listr\}_{\ri\in R}$, with each list containing at most $\lsR$ codewords of $\codeR$.
	
	Let $h$ be an arbitrary element of $\calL(g,\decT-\eps)$, which in particular means that $h\in \TanC$.
	\begin{proposition}\label{prop:local_presence}
		For any $h\in \calL(g,\decT-\eps)$, it holds that $\Pr{\li\in L}{h_{\li}\not\in \listl} \leq \delR-\eps$. Similarly, on the right, it holds that $\Pr{\ri\in R}{h_{\ri}\not\in \listr} \leq \delL-\eps$.
	\end{proposition}
	\begin{proof}
		Unlike in the case of AEL codes, this does not require the expander mixing lemma, but is based on simple counting using our choice of $\decT$.
		\begin{align*}
			\Delta(g,h) = \Pr{e\in E}{g_e\neq h_e} &= \Ex{\li\in L}{\Delta(g_{\li},h_{\li})} \\
			&= \Ex{\li\in L}{\Delta(g_{\li},h_{\li}) ~\big\vert~ h_{\li}\in \listl} \cdot \Pr{\li\in L}{h_{\li}\in \listl} + \Ex{\li\in L}{\Delta(g_{\li},h_{\li}) ~\big\vert~ h_{\li}\not\in \listl} \cdot \Pr{\li\in L}{h_{\li}\not\in \listl} \\
			&> 0\cdot \Pr{\li\in L}{h_{\li}\in \listl} + \decL\cdot \Pr{\li\in L}{h_{\li}\not\in \listl} \\
			&= \decL\cdot \Pr{\li\in L}{h_{\li}\not\in \listl}
		\end{align*}
		Using $\Delta(g,h) \leq \decT-\eps < \decL\cdot \delR-\eps$ gives 
		\[
			\Pr{\li\in L}{h_{\li}\not\in \listl} \leq \frac{\decL\cdot \delR-\eps}{\decL} = \delR - \frac{\eps}{\decL} \leq \delR - \eps \mper
		\]
		A similar argument on the right shows $\Pr{\ri\in R}{h_{\ri}\in \listr} \leq \delL$.
	\end{proof}
	
	Thus, $h_{\li}$ appears in many lists obtained by the vertex by vertex decoding above. To recover $h$, the algorithm must discover which lists it appears in, and among them which element of the list equals $h_{\li}$. This is exactly the type of problem we solved for AEL codes in \cref{sec:ael}, using tools from the regularity lemma of \cref{sec:reg_lemma}.
	
\subsubsection{Picking from local lists via regularity}	
If a list $\listl$ has $ < \lsL$ codewords, we add arbitrary codewords from outside the list to ensure that all lists $\{\listl\}_{\li\in L}$ are of size \emph{exactly} $\lsL$. Further, we assume that the lists are ordered, so that the $\codeL$ codewords in each $\li\in L$ can be indexed using an index $i\in [\lsL]$. Consequently, we will denote the $i^{th}$ codeword in $\listl$ by $\listl[i]$. 
	
	We repeat the same process for lists $\{\listr\}_{\ri\in R}$ so that they are all ordered lists of size exactly $\lsR$.
%
\paragraph{Subgraphs of $G$.}
We next use the local lists $\{\listl\}_{\li\in L}$ and $\{\listr\}_{\ri\in R}$ to create $\lsL \lsR$ many subgraphs of $G$. 
	
For each $i\in [\lsL]$ and $j\in [\lsR]$, we define the $H_{ij}$ to be the subgraph of $G$ obtained by deleting all edges $e = (\li,\ri)\in E(G)$ such that the symbol induced on $e$ by $\listl[i]$ and $\listr[j]$ do not match. Formally, $H_{ij}=(L,R,E(H_{ij}))$, where
\[
	E(H_{ij}) = \inbraces{(\li,\ri)\in E(G)~\vert~ \listl[i] = \listr[j]}
\]
Note that the graphs $H_{ij}$ only depend on the input $g$, and so can be computed in linear time.

For every $i\in [\lsL]$ and $j\in [\lsR]$, we can use \cref{lem:regularity-family} on $H_{ij}$ to obtain an $(\eta,\gamma)$-regular family for $H_{ij}$. Let $\family$ be the union of all of these families, so that $\family$ is simultaneously $(\eta,\gamma)$-regular for all $H_{ij}$, where $i\in [\lsL]$ and $j\in [\lsR]$.

\paragraph{From codewords to sets.}
Fix an $h\in \calL(g,\decT-\eps)$. We define $\lsL$ subsets $A_1,A_2,\cdots ,A_{\lsL}$ of $L$ based on the positions of $h_{\li}$ in $\listl$.
		\[
			A_i = \{\li\in L ~:~ \listl[i] = h_{\li}\} \mper
		\]
		Similarly, we also define $\lsR$ subsets $B_1,B_2,\cdots ,B_{\lsL}$ of $R$ based on the positions of $h_{\ri}$ in $\listr$., defined as
		\[
			B_j = \{\ri \in R ~:~ \listr[j] = h_{\ri}\} \mper
		\]
		Note that both $A_1,A_2,\cdots ,A_{\lsL}$ and $B_1,B_2,\cdots ,B_{\lsR}$ are pairwise disjoint collections of sets. Also note that a vertex $\li\in L$ is present in $\cup_{i\in [\lsL]}A_i$ if and only if $h_{\li}\in \listl$. This means that \cref{prop:local_presence} translates to 
		\begin{align*}
		\sum_{i} |A_i| \geq (1-\delR+\eps)|L| \qquad \text{and} \qquad \sum_{j} |B_j| \geq (1-\delL+\eps)|R|
		\end{align*}
		Further, we note that for any $i\in [\lsL]$ and $j\in [\lsR]$, no edges between $A_i$ and $B_j$ are deleted when constructing $H_{ij}$. That is, $E_{H_{ij}}(A_i,B_j) = E_G(A_i,B_j)$. This is because for any edge $e=(\li,\ri)\in A_i\times B_j$, $(\listl)_e = h_e = (\listr)_e$.
%
%

If we knew either the collection of sets $A_1,A_2,\cdots ,A_{\lsL}$ or the collection of sets $B_1,B_2,\cdots ,B_{\lsR}$, then we could recover $h$ by marking $A_0$ or $B_0$ as erasures and using \cref{lem:err_and_erasure}. Therefore, finding these sets is as hard as finding $h$ itself.
However, given any family $\family$ on $L$, we can find a collection of disjoint sets $S_1,S_2,\cdots ,S_{\lsL} \sub L$ using \cref{clm:enumeration} such that 
\begin{align*}
	&\forall i \in [\lsL], ~~\norm{\one_{A_i} - \one_{S_i}}_{\family} \leq \eta
\end{align*}

The following key claim shows that if the family $\family$ is simultaneously $(\eta,\gamma)$-regular for all the graphs $H_{ij}$, with $i\in [\lsL]$ and $j\in [\lsR]$, then any such collection $S_1,S_2,\cdots ,S_{\lsL}$ must be very close to the target sets $A_1,A_2,\cdots ,A_{\lsL}$.

%
%

\begin{lemma}[Rigidity]\label{lem:tanner-rigidity}
	Let $\family$ be a family that is simultaneously $(\eta,\gamma)$-regular for all $H_{ij}$, where $i\in [\lsL]$ and $j\in [\lsR]$. Let $S_1, S_2,\cdots , S_{\lsL} \sub L$
	be a disjoint collection of sets satisfying $\max_i \norm{\one_{A_i} - \one_{S_i}}_{\family} \leq \eta$. Assume\footnote{As in the case of AEL codes, $\lambda$ being much smaller will be anyway needed to obtain $(\eta,\gamma)$-regular families.} $\lambda \leq \gamma$.
	Then, 
	\[
		\sum_{i\in [\lsL]} |S_i\setminus A_i| \leq \parens*{3\lsL \lsR \cdot \gamma /\eps}n \mper
	\]
\end{lemma}

\begin{proof}
	By the $\eta$-closeness of sets with respect to the family $\family$,
	\begin{align*}
		\sum_{i,j} \abs{E_{H_{ij}}(S_i,B_j)} &\geq \sum_{i,j} \parens*{\abs{E_{H_{ij}}(A_i,B_j)} - \gamma \cdot dn} \\
		&= \sum_{i,j} \abs{E_{H_{ij}}(A_i,B_j)} - \lsL \lsR \cdot \gamma \cdot dn \\
		&= \sum_{i,j} \abs{E_{G}(A_i,B_j)} - \lsL \lsR \cdot \gamma \cdot dn \\
		&\geq \sum_{i,j} \abs{E_{G}(S_i,B_j)} - 2 \lsL \lsR \cdot \gamma \cdot dn
	\end{align*}
	where in the last step we used that $\family$ is also $(\eta,\gamma)$-regular for $G$ (see \cref{remark:trivial_factor}).
	
	Fix some $i\in [\lsL]$. We define $W_i = S_i \setminus A_i$. For any $\li \in W_i$, it holds that $\listl[i] \neq h_{\li}$, and so $\listl[i]$ and $h_{\li}$ must differ on at least $\delL$ fraction of the edges incident on $\li$. If such an edge goes to $B_j$ for some $j$, it will get deleted in the corresponding $H_{ij}$. As a result, any $\li\in W_i$ can contribute at most $(1-\delL)d$ edges to $\sum_j E_{H_{ij}}(S_i, B_j)$.
	\begin{align*}
	 \sum_{j} |E_{H_{ij}}(S_i, B_j)| &= \sum_{j} |E_{H_{ij}}(S_i \setminus W_i, B_j)| + \sum_{j} |E_{H_{ij}}(W_i, B_j)|\\
	 &\leq \sum_{j} |E_{G}(S_i \setminus W_i, B_j)| + |W_i| (1-\delL) d\\
	 &= \sum_{j} |E_{G}(S_i, B_j)| - \sum_j |E_G(W_i,B_j)|+ |W_i| (1-\delL) d\\
	 &\leq \sum_{j} |E_{G}(S_i, B_j)| - \frac{d}{n} |W_i| \parens*{\sum_j |B_j|} + \lsR \cdot \lambda \cdot dn + |W_i| (1-\delL) d\\
	 &\leq \sum_{j} |E_{G}(S_i, B_j)| - \eps\cdot d |W_i| + \lsR \cdot \lambda \cdot dn
	\end{align*}
	where the second last inequality uses expander mixing lemma, and the last inequality uses $\sum_j |B_j| \geq (1-\delL+\eps)n$.
	Summing over $i$ and comparing,
	\begin{align*}
		\sum_{i,j} \abs{E_{G}(S_i,B_j)} - 2 \lsL \lsR \cdot \gamma \cdot dn &\leq \sum_{i,j} \abs{E_{H_{ij}}(S_i,B_j)} \leq \sum_{i,j} |E_{G}(S_i, B_j)| - \eps\cdot d \cdot \sum_i |W_i| + \lsL \lsR \cdot \lambda \cdot dn \\
		\implies \sum_i |W_i| &\leq \parens*{\frac{2\lsL \lsR \cdot \gamma + \lsL \lsR \cdot \lambda}{\eps}}n \leq \parens*{\frac{3\lsL \lsR \cdot \gamma}{\eps}}n
	\end{align*}
	where in the last inequality we used $\lambda \leq \gamma$.
\end{proof}
%
%
%

\subsubsection{The list decoding algorithm for Tanner codes}

Next, we use the rigidity property established above to design an algorithm to list decode. The theorem statement below abstracts out the runtimes for the list decoding in each vertex neighborhood of $\{g_{\li}\}_{\li\in L}$ and $\{g_{\ri}\}_{\ri\in R}$ to obtain the local lists, the runtime to obtain $(\eta,\gamma)$-regular families, and the runtime to unique decode the Tanner code once we have figured out the sets $A_1,\ldots, A_{\lsL}$.

As before, let $\TanC$ be a Tanner code obtained by using an $(n,d,\lambda)$-expander with two left and right codes $\codeL$ and $\codeR$ respectively. The list decoder below finds the list using $\lsL \lsR$ many calls to the algorithmic regularity lemma \cref{cor:expander_regularity} and constantly many calls to the unique decoder of \cref{lem:err_and_erasure}, where the constant only depends on $\eps$, $\lsL$ and $\lsR$.

\begin{new-algorithm}{List Decoding $\TanC$}\label{algo:tanner-decoding} 
\begin{tabular}{r l}
\textsf{Input:} & $g \in\Sigma^E$ \\[3 pt]
\textsf{Output:} & List $\calL \subseteq \TanC$\\[3 pt] 
\textsf{Parameters:} & Distance $\delL$ for $\codeL$ and distance $\delR$ for $\codeR$,\\
& list size bound $\lsL$
                       for  $\codeL$ at radius $\decL$ and list size bound $\lsR$ for $\codeR$ at radius $\decR$, \\
& list decoding radius $\decT - \eps = \min\parens*{\delL \decR, \decL \delR} - \eps$
\end{tabular}
\begin{itemize}
\medskip
\item For each $\li \in L$, compute the local list $\listl = \inbraces{f_{\li} \in \codeL ~|~ \Delta(f_{\li}, g_{\li}) \leq \decL} = \calL_{\codeL}(g_{\li},\decL)$. Pad with additional elements of $\codeL$ if needed, to view each $\listl$ as an ordered list of size $\lsL$.
\item For each $\ri \in R$, compute the local list $\listr = \inbraces{f_{\ri} \in \codeR ~|~ \Delta(f_{\ri}, g_{\ri}) \leq \decR} = \calL_{\codeR}(g_{\ri},\decR)$. Pad with additional elements of $\codeR$ if needed, to view each $\listr$ as an ordered list of size $\lsR$.
\item For each $i \in [\lsL]$ and each $j\in [\lsR]$ define a graph $H_{ij}$ with vertex sets $(L,R)$ and edges
\[
E(H_{ij}) ~=~ \{e = (\ell, r) \in E(G) ~|~ (\listl[i])_e = (\listr[j])_e \} \mper
\]
For $\gamma = \eps^2/(14\lsL \lsR)$ and sufficiently small $\eta$, compute a family $\family_{ij}$ which is $(\eta,\gamma)$-regular for $H_{ij}$. Let $\family$ be the union of all the $\lsL \lsR$ many families obtained.
\item For each $\sigma = (\sigma_1, \cdots, \sigma_{\lsL}) \in \calN_{\eta/4}(\family, \lsL)$
\begin{itemize}
\item Consider arbitrary disjoint sets $S_1, \ldots, S_{\lsL} \subseteq L$ given by \cref{clm:enumeration} such that\\ $\max_{i \in [\lsL]} \norm{\sigma_i - \sigma(\indicator{S_i})}_{\infty} \leq \eta/2$.
\item Define $h' \in (\codeL \cup \{\bot\})^L$ as
\[
h'_{\ell} ~=~
\begin{cases}
\calL_{\ell}[i] &~\text{if}~ \ell \in S_i \\
\bot &~\text{if}~ \ell \notin \cup_i S_i
\end{cases}
\]
\item Use the errors-and-erasures decoder of \cref{lem:err_and_erasure} on $h'$.
\item If the decoder succeeds, let $h\in \TanC$ be the output of the decoder. If $\Delta(g,h) \leq \decT - \eps$, then $\calL \leftarrow \calL \cup \{h\}$.
\end{itemize}
\item Return $\calL$.
\end{itemize}
\end{new-algorithm}
\medskip

\begin{theorem}\label{thm:tanner-technical}
	Let $\tlocalL{\decL}$ be the time required to list decode $\codeL$ up to radius $\decL$, $\tlocalR{\decR}$ be the time required to list decode $\codeR$ up to radius $\decR$, $\treg{\gamma}$ be the time required to obtain an $(\eta,\gamma)$-regular decomposition for any subgraph of
$G$ for $\gamma = \frac{\eps^2}{14 \lsL \lsR}$ and some $\eta = \Omega(\gamma^2)$, and $\tdec{\delta}$ be the time required to run the errors-and-erasures decoder up to radius $\delta$. Assume that $\lambda \leq \min\parens*{\gamma, \frac{\eps\cdot \delL}{64}, \frac{\eps\cdot \delR}{64}}$. 

Then for any input $g\in \Sigma^E$, \cref{algo:tanner-decoding} runs in time 
\[
	O\parens*{n\cdot \tlocalL{\decL} + n\cdot \tlocalR{\decR} + \lsL \lsR \cdot \treg{\gamma}} + \parens*{|E|+\tdec{\delR - \frac{\eps}{2}}}\cdot \parens*{\frac{1}{\gamma}}^{O(\lsL^2 \lsR / \gamma^2)} \mcom
\]
and outputs the list $\calL_{\TanC}(g, \decT - \eps) = \{h\in \TanC ~|~ \Delta(g,h) < \decT -\eps\}$. Further, the list is of size at most $(1/\gamma)^{O(\lsL^2 \lsR / \gamma^2)}$.
\end{theorem}
\begin{proof}
We note that the computation of the local lists, and the graphs $H_{ij}$ can be done in time $O(n \cdot \tlocalL{\decL} + n \cdot \tlocalR{\decR})$. 
By \cref{lem:regularity-family}, we can take $\eta = \Omega(\gamma^2)$ and obtain an $(\eta,\gamma)$-regular family $\family_{ij}$ for each graph $H_{ij}$ in time $\treg{\gamma}$, with $\abs{\family_{ij}} = O(1/\gamma^2)$. Thus, for the union $\family$, we have $\abs{\family}= O(\lsL \lsR/\gamma^2)$
By \cref{clm:enumeration}, we have $\abs{\calN_{\eta/4}(\family,\lsL)} = (1/\eta)^{O(\lsL \cdot \abs{\family})} = (1/\gamma)^{O(\lsL^2 \lsR/\gamma^2)}$, and the enumeration can be done in time $\tdec{\delR - \frac{\eps}{2}} \cdot (1/\gamma)^{O(\lsL^2 \lsR/\gamma^2)}$.

Finally, it remains to prove that the algorithm correctly recovers the list. 
Note that for any $h$ satisfying $\dis(g,h) \leq \decT - \eps$, and corresponding sets $\inbraces{A_i}_{i \in [K]}$ indicating its position in local lists, the $h'$ found by our algorithm agrees with $h$ on all vertices $\ell \in \cup_i (S_i \cap A_i)$ and disagrees on all vertices $\ell \in \cup_i (S_i \setminus A_i)$. 
The remaining vertices in $L \setminus \cup_i S_i$ are marked as erasures ($\bot$) in $h'$. Therefore,
\[
\dis_L(h,h') ~=~ \sum_i \ExpOp[\indicator{S_i \setminus A_i}] \mper
\] 
%
where the $\Delta_L(\cdot,\cdot)$ counts only errors and not erasures. 
By \cref{clm:enumeration} there \emph{exists} $\sigma = (\sigma_1, \ldots, \sigma_{\lsL}) \in
\calN_{\eta/4}(\family,\lsL)$ such that disjoint sets $S_1, \ldots, S_{\lsL}$ produced using $\sigma$ will satisfy $\max_{i \in [\lsL]} \norm{\indicator{A_i}-\indicator{S_i}}_{\family} \leq \eta$.
Also, by \cref{lem:tanner-rigidity}, for any such collection $S_1, \ldots, S_{\lsL}$, we have $\sum_i \ExpOp[\indicator{S_i \setminus A_i}] \leq (3\lsL \lsR \gamma/\eps)$. Moreover, by \cref{clm:local-membership}, we have $\sum_i \abs{A_i} \geq (1-\delR+\eps) \cdot |L|$ and $\norm{\indicator{A_i} - \indicator{S_i}}_{\factor} \leq \eta$ implies
\[
\sum_i \Ex{\indicator{S_i}} \geq \sum_i (\Ex{\indicator{A_i}}-\gamma) \geq (1-\delR + \eps - \lsL \gamma)
\]
Suppose $h'$ has $s_0 \cdot |L|$ many erasure symbols, then $s_0 = 1 - \sum_i \Ex{\indicator{S_i}}$. We compute,
\[
	2\dis_L(h,h') + s_0 = 2 \sum_i \Ex{\indicator{S_i \setminus A_i}} + \parens*{1-\sum_i \Ex{\indicator{S_i}}} \leq \frac{6\lsL \lsR \gamma}{\eps} + \delR -\eps + \lsL \gamma \leq \delR - \eps + \frac{7\lsL \lsR \gamma}{\eps} \mcom
\]
which is at most $\delR - \frac{\eps}{2}$ because $\gamma = \frac{\eps^2}{14 \lsL\lsR}$.
Thus, for every $h$ with $\dis_R(g,h) \leq \decT - \eps$, one of the choices in our enumeration finds an
$h'$ with $s_0\cdot |L|$ erasures satisfying $2\dis_L(h,h') + s_0 \leq \delR - \frac{\eps}{2}$. Further, it can be verified that our choice of $\lambda$ allows us to use the errors-and-erasures unique decoder of \cref{lem:err_and_erasure} on this $h'$ to recover $h$, and
thus the enumeration algorithm recovers all $h$ in the list.
\end{proof}

\subsection{Main result for Tanner codes}

We instantiate the \cref{thm:tanner-technical} using a brute force list decoder for $\codeL$ and $\codeR$, \cref{lem:regularity-family} for algorithmic regularity, and \cref{lem:err_and_erasure} for unique decoding, to obtain our main result for the case of Tanner codes.
\begin{theorem}
	Suppose the code $\codeL$ is list decodable up to radius $\decL$ with list size $\lsL$ and the code $\codeR$ is list decodable up to radius $\decR$ with list size $\lsR$. Let $\decT = \min\parens*{\delL \decR, \decL \delR}$ and let $\eps > 0$ be arbitrary. If $\lambda \leq \frac{\eps^4}{5\cdot 10^5\cdot K_1^2 K_2^2}$, then the code $\TanC$ can be list decoded up to radius $\decT-\eps$ with list size $\exp(O(K_1^5 K_2^4/\eps^3))$ in time $\widetilde{O}_{d,|\Sigma|}(n)$.
\end{theorem}

\begin{proof}
	Given any $g\in \Sigma^E$, we can use the list decoder of \cref{thm:tanner-technical} to recover $\calL(g,\decT-\eps)$, as long as $\lambda\leq \gamma = \frac{\eps^2}{14\lsL\lsR}$.
	
	We first observe that $\tlocalL{\decL}$ and $\tlocalR{\decR}$ can be simply taken to be $O(|\Sigma|^d)$ in the worst case, by simply trying out all possible candidates for the list. Note that we are promised bounds of $\lsL$ and $\lsR$ on this list size in the theorem statement.
	
	The time needed to find an $(\eta,\gamma)$-regular family is $\Ot{|E|} = \widetilde{O}(n)$ by \cref{lem:regularity-family} if $\lambda \leq \frac{\gamma^2}{2500}$.
	
	Finally, the errors-and-erasures decoder of \cref{lem:err_and_erasure} runs in time $O(n\cdot |\Sigma|^d)$ in the worst case, assuming $\lambda\leq \min\parens*{\frac{\eps\cdot \delL}{64}, \frac{\eps\cdot \delR}{64}}$.
	
	It therefore suffices to choose $\lambda$ to be at most $\frac{\eps^4}{5\cdot 10^5\cdot K_1^2 K_2^2}$. The total runtime is $\Ot{n \cdot |\Sigma|^d} = \widetilde{O}_{d,|\Sigma|}(n)$.
\end{proof}

We get the following simplified version when both left and right codes are the same.
\begin{corollary}\label{cor:tanner-instantiation}
	Let $\TanC$ be the Tanner code defined using an $(n,d,\lambda)$-expander and a single code $\calC_0$ over some constant sized alphabet $\Sigma$, used on both left and right sides. Suppose $\calC_0$ has distance at least $\delta$, and is list decodable up to radius $\delta^{\dec}$ with a list size of $K$ that is independent of $d$. Then given any $\eps>0$ if $\lambda \leq c\cdot \frac{\eps^4}{K^4}$ for some small constant $c$, then there is an algorithm that takes as input $g\in \Sigma^E$, runs in time $\Ot{n}$, and returns the list $\calL(g,\delta\cdot \delta^{\dec}-\eps)$, which is of size at most $\exp(O(K^9/\eps^3))$.
\end{corollary}

Of course, if the code $\calC_0$ above is decodable up to its distance $\delta$, then $\TanC$ is also decodable up to $\delta^2$.



%% file: appendix.tex
\section{Errors-and-Erasures Decoding of Tanner Codes}\label{sec:appendix_tanner}

\begin{lemma}[Restatement of \cref{lem:err_and_erasure}]\label{lem:err_and_erasure_appendix}
	Let $\TanC$ be a Tanner code defined using left code $\codeL$ with distance $\delL$ and right code $\codeR$ with distance $\delR$, both defined over the alphabet $\Sigma$, on the graph $G = (L,R,E)$ that is an $(n,d,\lambda)$-expander. Fix $\eps > 0$ such that $\lambda \leq \frac{\eps}{8} \cdot \min(\delL,\delR)$. Given a $g\in (\codeL \cup \{\bot\})^L$ with $s\cdot |L|$ many erasures, we have the following:
	\begin{itemize}
	\item there is at most one $h\in \TanC$ that satisfies $2 \dis_L(g,h) + s \leq \delR - 4\eps$.
	\item If such an $h$ exists, then there is an algorithm that given $g$, runs in linear time $O(n\cdot |\Sigma|^d)$, and outputs $h$.
	\end{itemize}
\end{lemma}
\begin{proof}
	Let $h\in \TanC$ be an arbitrary codeword that satisfies $2\dis_L(g,h) + s \leq \delR - 4\eps$. We will give an algorithm that always outputs $h$ (when given $g$ as input). This will prove both the parts in one shot.
	
	We will construct a sequence of $g_2,g_3,\ldots, $ such that for each $i\in \N$, $g_{2i} \in (\codeR)^R$ and $g_{2i+1} \in (\codeL)^L$. The main idea will be that these $g_2,g_3,\ldots $ grow closer and closer to $h$.
	
	To construct $g_2$, we look at the projections $g_{\ri} \in (\Sigma \cup \{\bot\})^{N(\ri)}$ of $g$ on the local neighborhoods of $\ri\in R$. Then $(g_2)_{\ri}$ is defined to be the closest codeword to $g_{\ri}$, that is,
	\[
		\forall \ri \in R, \qquad (g_2)_{\ri} = {\arg\min}_{f_{\ri} \in \codeR} \Delta(f_{\ri}, g_{\ri})
	\]
	Ties are broken arbitrarily (for example, always pick one according to some canonical ordering). We claim that $g_2$ can be made arbitrarily close to $h$ by decreasing $\lambda$. This claim is slightly different from the general claims later because it involves correcting both errors and erasures from $g$. For later claims, we will only deal with errors, which will simplify the analysis.
	\begin{claim}
		$\Delta_R(g_2,h) \leq 2\lambda/\eps \leq \frac{\delL}{4}$.
	\end{claim}
	\begin{proof}
		Let $L' = \{ \li \in L ~|~ g_{\li} = \bot\}$ and $L'' = \{ \li \in L \setminus L' ~|~ g_{\li} \neq h_{\li}\}$. The condition $2\Delta_L(g,h) + s \leq \delR -2\lambda$ then translates to $2|L''| + |L'| \leq (\delR - 4\eps) n$.
		
		For an $\ri \in R$, if the number of edges from $\ri$ to $L'$ is at most $\parens*{\frac{|L'|}{n} + \eps}\cdot d$ and the number of edges from $\ri$ to $L''$ is at most $\parens*{\frac{|L''|}{n} + \eps}\cdot d$, then $(g_2)_{\ri} = h_{\ri}$. This is because for such an $\ri$, if the number of erasure symbols in $g_{\ri}$ is $s_{\ri}\cdot d$, then 
	\[
		2\Delta((g_2)_{\ri},h_{\ri}) + s_{\ri} \leq 2\parens*{\frac{|L''|}{n} + \eps} + \frac{|L'|}{n}+\eps \leq \delR - 4\eps + 3\eps = \delR -\eps
	\]
	and so $h_{\ri}$ is the nearest codeword to $g_{\ri}$. We will therefore upper bound $\Delta_R(g_2,h)$ by showing that not too many vertices in $R$ can have more than expected number of edges going to $L'$ or $L''$.
	
	Let $R' = \inbraces{\ri \in R ~\big\vert~ \abs{N(\ri) \cap L'} \geq \parens*{\frac{|L'|}{n}+\eps}\cdot d}$. Then,
	\[
		E(L',R') \geq |R'| \cdot \parens*{\frac{|L'|}{n}+\eps}\cdot d
	\]
	and using expander mixing lemma,
	\[
		E(L',R') \leq \frac{d}{n} |L'|\cdot |R'| + \lambda d n
	\]
	Comparing the upper and lower bounds,
	\begin{align*}
		& |R'| \cdot \parens*{\frac{|L'|}{n}+\eps}\cdot d \leq \frac{d}{n} |L'|\cdot |R'| + \lambda d n\\
		\implies & \quad |R'| \leq \frac{\lambda}{\eps}n
	\end{align*}
	Similarly, define $R'' = \inbraces{\ri \in R ~\big\vert~ \abs{N(\ri) \cap L''} \geq \parens*{\frac{|L''|}{n}+\eps}\cdot d}$, so that
	\[
		|R''| \leq \frac{\lambda}{\eps}\cdot n
	\]
	Finally, $\Delta(g_2,h) \leq \frac{|R'|+|R''|}{n} \leq \frac{2\lambda}{\eps}$.
	\end{proof}
	Now we move to the general case. First we show how to construct the sequence $g_3,g_4,\ldots$.
	For any $i\in \N$, $g_{2i+1} \in (\codeL)^L$ is constructed as 
	\[
		\forall \li \in L, \qquad (g_{2i+1})_{\li} = {\arg\min}_{f_{\li} \in \codeL} \Delta(f_{\li}, (g_{2i})_{\li})
	\]
	Similarly, for any $i=2,3,\ldots$, the corresponding $g_{2i} \in (\codeR)^R$ is constructed as
	\[
		\forall \ri \in R, \qquad (g_{2i})_{\ri} = {\arg\min}_{f_{\ri} \in \codeR} \Delta(f_{\ri}, (g_{2i-1})_{\ri})
	\]
	The following claim says that if $g_{2i}$ is close to $h$ on $R$, then $g_{2i+1}$ is significantly closer to $h$ on $L$.
	\begin{claim}
		Let $i \geq 1$ be an integer. Suppose $\Delta_R(g_{2i},h) \leq \frac{\delL}{4}$, then $\Delta_L(g_{2i+1},h) \leq \frac{\eps^2}{4} \cdot \Delta_R(g_{2i},h)$.
	\end{claim}
	\begin{proof}
		Let $R' \sub R$ be the set of right vertices where $g_{2i}$ and $h$ differ. Any $\li \in L$ that sends less than $\frac{\delL}{2}\cdot d$ edges to $R'$ gets corrected to $h_{\li}$ in $g_{2i+1}$. Therefore, we only need to count how many left vertices send at least $ \frac{\delL}{2} \cdot d$ edges to $R'$. Let $L' = \{ \li \in L ~\big\vert~ |N(\li) \cap R'| \geq \frac{\delL}{2}\cdot d\}$.
		\begin{align*}
			& |L'| \cdot \frac{\delL}{2}\cdot d\leq E(L',R') \leq \frac{d}{n}\cdot |L'|\cdot |R'| + \lambda d\cdot \sqrt{|L'| \cdot |R'|}\\
			\implies & \qquad \frac{|R'|}{|L'|} \geq \frac{1}{\lambda^2} \parens*{\frac{\delL}{2} - \frac{|R'|}{n}}^2 \geq \frac{1}{\lambda^2} \frac{\delta_L^2}{16}\\
			\implies & \qquad |L'| \leq \frac{\lambda^2 }{\delta_L^2/16} \cdot |R'| \leq \frac{\eps^2}{4} \cdot |R'| \mper
		\end{align*}
	\end{proof}
	In particular, $\Delta_L(g_3,h) \leq \frac{\eps^2}{4} \cdot \frac{2\lambda}{\eps} = \frac{\lambda \eps}{2}$, which is at most $\delR/4$.
	Similar to above, we have a claim that says that if $g_{2i-1}$ and $h$ are close on $L$, then $g_{2i}$ and $h$ are even closer on $R$.
	\begin{claim}
		Let $i \geq 2$ be an integer. Suppose $\Delta_L(g_{2i-1},h) \leq \frac{\delR}{4}$, then $\Delta_R(g_{2i},h) \leq \frac{\eps^2}{4} \cdot \Delta_L(g_{2i-1},h)$.
	\end{claim}
	\begin{proof}
		Let $L' \sub L$ be the set of left vertices where $g_{2i-1}$ and $h$ differ. Any $\ri \in R$ that sends less than $\frac{\delR}{2}\cdot d$ edges to $L'$ gets corrected to $h_{\ri}$ in $g_{2i}$. Therefore, we only need to count how many right vertices send at least $ \frac{\delR}{2} \cdot d$ edges to $L'$. Let $R' = \{ \ri \in R ~\big\vert~ |N(\ri) \cap L'| \geq \frac{\delR}{2}\cdot d\}$.
		\begin{align*}
			& |R'| \cdot \frac{\delR}{2}\cdot d\leq E(L',R') \leq \frac{d}{n}\cdot |L'|\cdot |R'| + \lambda d\cdot \sqrt{|L'| \cdot |R'|}\\
			\implies & \qquad \frac{|L'|}{|R'|} \geq \frac{1}{\lambda^2} \parens*{\frac{\delR}{2} - \frac{|L'|}{n}}^2 \geq \frac{1}{\lambda^2} \frac{\delta_R^2}{16}\\
			\implies & \qquad |R'| \leq \frac{\lambda^2 }{\delta_R^2/16} \cdot |L'| \leq \frac{\eps^2}{4} \cdot |R'| \mper
		\end{align*}
	\end{proof}
	It is now easy to see that the distances $\Delta_R(g_{2i},h)$ are monotonically decreasing with 
	\[	
		\Delta_R(g_{2i+2},h) \leq \parens*{\frac{\eps^2}{4}}^2 \cdot \Delta_R(
	g_{2i},h) \mper
	\]
	Thus for some $i = \Theta(\log n)$, $g_{2i}$ will converge to $h$. 
	
	This immediately implies a runtime bound of $O(\log n \cdot |\Sigma| \cdot n)$, since there are $O(\log n)$ iterations, and each iteration only involves locally decoding for each left or right vertex. This decoding takes time $|\Sigma|^d$ per vertex in the worst case, so the total cost per iteration is $O(|\Sigma|^d\cdot n)$. For our purposes, this near-linear running time suffices.
	
	To get a linear time implementation, observe that to obtain $g_{2i+2}$ from $g_{2i}$, it suffices to change only the neighborhoods of vertices where $g_{2i+1}$ changed from $g_{2i-1}$. This difference is at most
	\[
		\Delta_L(g_{2i-1},g_{2i+1})\cdot n \leq \Delta_L(g_{2i-1},h)\cdot n + \Delta_L(g_{2i+1},h)\cdot n \leq 2 \Delta_L(g_{2i-1},h)\cdot n \mper
	\] 
	A similar claim also holds true for computing the odd indices $g_5,g_7,\ldots$. Thus the total amount of work done to compute $g_6, g_8, \ldots$ is
	\[
		\sum_{i=2}^{\infty} |\Sigma|^d \cdot 2\Delta_L(g_{2i-1},h) n \leq 2|\Sigma|^d n \cdot \Delta_L(g_3,h) \sum_{i=0}^{\infty} \parens*{\frac{\eps^2}{4}}^i = \frac{2|\Sigma|^dn \Delta_L(g_3,h)}{1-\frac{\eps^2}{4}} = O(n \cdot |\Sigma|^d) \quad \qedhere
	\]
\end{proof}

Note that runtime analysis assumed a worst case bound of $|\Sigma|^d$ for the inner code, which is constant if the alphabet $\Sigma$ is of constant size. In the case when faster errors-and-erasures decoders are available for the inner code, such as in the case of polynomial codes like Reed-Solomon, we can get rid of the exponential dependence on $d$ and replace it with a polynomial one, as done in \cite{RWZ21}, for instance.